\documentclass{article}

\usepackage[english]{babel}
\usepackage{natbib}
\usepackage[letterpaper,top=2cm,bottom=2cm,left=3cm,right=3cm,marginparwidth=1.75cm]{geometry}
\usepackage{amsmath, amsthm, amssymb}
\usepackage{graphicx}
\usepackage[colorlinks=true, allcolors=blue]{hyperref}
\usepackage{subcaption}
\usepackage{mathtools}
\usepackage{wrapfig}
\usepackage{booktabs}
\usepackage{multirow}
\usepackage{multicol}
\usepackage[table]{xcolor}
\usepackage{tabularx}
\makeatletter
\@ifundefined{insert@pcolumn}{\let\insert@pcolumn\insert@column}{}
\makeatother
\graphicspath{{pictures/}}
\usepackage{subcaption}

\usepackage{algorithm}
\usepackage{algpseudocode}

\newcommand{\blind}{1}%JASA-template
\newcommand{\jbes}{0}

\DeclareMathOperator*{\argmax}{argmax}

\DeclareMathOperator{\logit}{logit}
\DeclareMathOperator{\expit}{expit}
\newtheorem{theorem}{Theorem}[section]
\newtheorem{proposition}[theorem]{Proposition}

\newtheorem{assumption}[theorem]{Assumption}
\newtheorem{lemma}[theorem]{Lemma}

\usepackage{pifont}
\newcommand{\cmark}{\ding{51}}
\newcommand{\xmark}{\ding{55}}

\newcolumntype{L}[1]{>{\raggedright\arraybackslash}p{#1}}
\newcolumntype{C}{>{\centering\arraybackslash}X}

\title{A Censored Transformed Model for Proportional Outcomes with Boundary Mass and an Application to Loss Given Default Modeling}
\author{
  Yuan Christopher Qiang\footnotemark[1] \\
  \and
  Fabio Sigrist\footnotemark[1] \footnotemark[3]% chktex 42
}

\begin{document}
\maketitle

\begin{abstract}
We introduce the zero-one censored transformed normal (ZOC-TN) model for proportional responses with potential probability mass at the boundaries 0 and 1. The model combines a censored Gaussian variable with a two-parameter affine-logit transformation on the interior (0,1). We characterize the transformation parameters, establish large-sample properties, and relate the affine-logit specification to broader classes of interior distributions. Theoretical and experimental results demonstrate that the proposed model can capture a wider range of qualitative density shapes than several benchmark models while remaining parsimonious, computationally efficient, and numerically stable. Furthermore, the ZOC-TN model can be extended (i) to account for nonlinearities and interactions in a tree-boosting machine learning framework and (ii) to explicitly model residual spatio-temporal variability. We apply the ZOC-TN model to loss given default (LGD) modeling for a large dataset of U.S. residential mortgages and compare it to multiple benchmark models. We find that a tree-boosted ZOC-TN model with a spatio-temporal frailty Gaussian process delivers the strongest out-of-sample performance, indicating that mortgage losses are shaped by nonlinear covariate effects and by unaccounted-for space-time variation.
\end{abstract}

\noindent%

% The resulting likelihood is parsimonious, computationally stable, and flexible enough to capture a wide range of asymmetric and nonstandard interior density shapes commonly observed in fractional-response data.

\footnotetext[1]{Seminar for Statistics, ETH Zurich}
% Personal email address: christopher.yq@gmail.com
\footnotetext[3]{Corresponding author: fabio.sigrist@stat.math.ethz.ch}

\section{Introduction}
Responses constrained to the unit interval arise in many fields such as econometrics and risk modeling, with rates, proportions, shares, utilization indices, recovery rates, and standardized losses all taking values in $[0,1]$. Such outcomes are commonly referred to as \emph{fractional} or \emph{bounded} response variables, or simply \emph{proportional data}; outside of finance and economics, they are common in fields such as neuroscience and ecology. The main modeling challenges are that these responses are bounded, can be asymmetric and multimodal, and may exhibit mass at the boundaries. %As a result, bounded-response regression constitutes a distinct modeling problem for which ordinary Gaussian methods are generally inappropriate.
 % \citep{Geissinger2022, Leng2021, Douma2019}
 
 Classical linear regression models can yield predicted values outside $[0,1]$, do not reflect the changing variance structure near the boundaries, and are often too rigid for the skewness and kurtosis observed in proportional data. Transformations such as the logit or probit of a rescaled response mitigate the support problem only partially, while introducing ad hoc preprocessing choices and reducing interpretability on the original scale. To address this, one strand of the fractional-response literature models only the conditional mean $\mathbb{E}[Y\mid \mathbf{x}]$. \citet{Papke1996} proposed quasi-likelihood estimation using a bounded response function, typically logistic, and subsequent work extended this framework in semiparametric and two-part directions \citep{Ramalho2011}. These models can handle boundary observations for estimation and are useful when inference on average effects is the primary objective, but they are less suited when the full predictive distribution or boundary probabilities are of interest.

An alternative option is to adopt fully parametric likelihoods for bounded outcomes. For example, beta regression can model fractional responses with support in $(0,1)$ \citep{Ferrari2004}. If observed responses attain the boundary values at 0 or 1, \citet{Smithson2006} propose the technique of ``nudging" to transform data into the interior before fitting a beta model. To explicitly model proportional data with boundary masses, two common approaches are mixture (two-stage) constructions and censored models. Mixture models, such as the zero-one-inflated beta model of \citet{Ospina2010}, treat boundary and interior outcomes as arising from separate components, whereas censored models, such as the two-limit Tobit model of \citet{Rosett1975}, model all observations through a single latent variable that is censored at $0$ and $1$. The latter approach is more parsimonious and yields a unified interpretation of covariate effects, which can be useful even when no explicit censoring occurs in the data collection. The censored Gamma model of \citet{Sigrist2011} is an extension based on a shifted and censored Gamma distribution, and \citet{Kosmidis2025} proposed extended-support beta models using a censored and transformed four-parameter beta family. Recently, \citet{lee2025} developed cobin and micobin regression models as robust extensions of beta regression.  
 % Within this class, the two-limit Tobit model can be too restrictive.

In this article, we develop a novel censored model for proportional outcomes denoted as the \emph{zero-one censored transformed normal} (ZOC-TN) distribution. This distribution arises from a censored normal random variable whose interior part is transformed according to an affine-logit transformation. The latter enables the interior density to capture a broad spectrum of distribution shapes including skewness and kurtosis. We characterize the transformation parameters, establish large-sample properties, derive the Fisher information, and show how the model approximates and can be extended to general distributions on the interior $(0,1)$. Moreover, we demonstrate experimentally and theoretically that the model can capture a wider range of qualitative density shapes than several benchmark models while remaining parsimonious, computationally efficient, and numerically stable. This computational tractability is an important contribution: it makes full-likelihood distributional models feasible in settings where repeated estimation, nonlinear machine learning predictors, or latent dependence structures would otherwise be difficult to combine with flexible censored likelihoods. We apply the proposed model to loss given default (LGD) forecasting using a large U.S. mortgage dataset. We find that a tree-boosted ZOC-TN model with a spatio-temporal frailty Gaussian process yields the highest out-of-sample prediction accuracy. This provides strong evidence of nonlinear effects and interactions in the determination of LGD as well as residual correlations across time and space.

% Moreover, the censored transformed normal construction can be extended naturally to include more general and flexible families of proportional likelihoods beyond our chosen affine-logit model. %In addition to performing a simulation study to evaluate model fit against established benchmark models, in the second half of the paper we apply our findings to an empirical loss given default dataset in the task of forecasting loss severity on single-family mortgage defaults in the continental United States. 

% In summary, this article makes three contributions. First, we introduce the zero-one censored transformed normal likelihood for proportional responses with potential boundary mass.  Second, through simulations, we show that ZOC-TN provides competitive accuracy across a range of censored and mixture data-generating processes, while retaining a clear computational advantage over censored beta and gamma alternatives. 

\subsection{Background on Loss Given Default Modeling}
The loss given default (LGD) is the proportion of the outstanding balance lost in the event a borrower defaults. Under Basel II, the LGD is a key risk parameter, alongside the probability of default (PD) and the exposure at default (EAD), in the provisioning and pricing of loans, whereby the total nominal loss is the product of the PD, EAD, and LGD. Through the Internal Ratings-Based (IRB) approach, Basel II and the 2017 Basel III reforms permit financial institutions to develop and deploy proprietary models to estimate the three components PD, EAD, and LGD. Despite this apparent need for robust modeling methods, the LGD has received notably less attention in the academic literature than the PD, especially in the case of retail exposure \citep{Leow2012, Bellotti2012}. We aim to address this by investigating a battery of fractional response models including traditional regression models and more modern machine learning methods.
% \citep{BaselIII}

% Early work on LGD modeling treats mortgage default as an embedded option on homeowner equity, and derives implications for loss severity from contingent-claims models \citep{Lekkas1993, Kau1999}. In this way, these models serve only to estimate loss given “optimal” default, and fail to account for nonfinancial contributors (or indeed inhibitors) of default. A discussion and empirical analysis on the realism of these ruthless models is provided in \cite{Ambrose2001}.

For our analysis, we model the LGD as a fractional response variable, taking values in $[0, 1]$, with $0$ representing a complete recovery, and $1$ being a total loss of the principal at the time of default. Several fractional response regression methods have been used in the literature. \cite{Bellotti2012}, for example, forecast LGD on a U.K. credit card dataset using Tobit regression, while \cite{Bastos2010} estimates the conditional mean of LGD using a Bernoulli quasi-likelihood approach. \citet{Sigrist2011} develop a censored gamma model and apply it to an insurance LGD dataset. \cite{Calabrese2014} considers recovery rates ($=1 - \mathrm{LGD}$) on Italian bank loans, building a mixture distribution with a Bernoulli random variable for the boundaries and a beta random variable on the interior. \cite{Yao2017} use a least squares support vector classifier to detect boundary observations and consider the effects of a wide array of interior regression models. Recent papers have proposed to apply machine learning methods to the task of forecasting LGD \citep{Tobback2014, Loterman2012, Kellner2022}. However, the matter of accounting for spatial correlations in LGD modeling has not been addressed before. As a secondary objective of this paper, we fill this gap by explicitly modeling residual spatial and temporal correlations in LGD through a latent Gaussian process.
%This approach follows from the work of \cite{Kundig2025} which develops a novel machine learning model for forecasting PD by combining tree-boosting with a latent spatio-temporal Gaussian process.

\section{The Zero-One Censored Transformed Normal (ZOC-TN) Model}\label{sec:zoctn_model}
As in the two-limit Tobit model \citep{Rosett1975}, we assume that a latent Gaussian variable is censored at $0$ and $1$:
\begin{equation}\label{def_Z}
    Z = \max\{\min(Z^*, 1), 0\}, ~~ Z^* \sim \mathcal{N}(\mu, \sigma^2),~~ \mu\in \mathbb{R}, \sigma^2>0.
\end{equation}
In a second step, the interior portion is mapped through a strictly monotone transformation
\begin{equation}\label{eq:zoctn_affine}
    Y = g_{a, b}(Z) := \begin{cases}
    0, & Z=0 \\
    \mathrm{expit}\left(a + b\,\mathrm{logit}(Z)\right), &Z\in(0, 1) \\
    1, &Z=1
    \end{cases}~~\text{for } a\in \mathbb{R}, b>0,
\end{equation}
where $\mathrm{logit}(z) = \log(z/(1-z))$ and $\mathrm{expit}(x) = \frac{1}{1 + e^{-x}}$. This transformation does not affect the boundary mass at $0$ and $1$, while allowing flexible skewness and kurtosis in the interior of the unit interval. The corresponding density is
\begin{equation}\label{eq:zoctn_pdf}
    f^{\mathrm{ZOC\text{-}TN}}(y \mid \mu,\sigma,a, b)
    =
    \begin{cases}
            \Phi\!\left(-\dfrac{\mu}{\sigma}\right),
        & y = 0 \\
            \phi\!\left(
                \dfrac{z_{a, b}(y) - \mu}{\sigma}
            \right)
            \dfrac{z_{a, b}(y)(1-z_{a, b}(y))}{\,b\sigma y(1-y)},
        & y \in (0,1) \\
            1 -
            \Phi\!\left(
                \dfrac{1-\mu}{\sigma}
            \right),
        & y = 1
    \end{cases}
\end{equation}
where $z_{a, b}(y) := g_{a, b}^{-1}(y) = \mathrm{expit}\left((\text{logit}(y) - a)/b\right)$ for $y\in(0,1 )$, with respect to the dominating measure $\delta_0+\lambda_{(0,1)}+\delta_1$, where \(\lambda_{(0,1)}\) is the Lebesgue measure on \((0,1)\), and $\delta_0$ and $\delta_1$ are Dirac measures on $0$ and $1$, respectively. 

\subsection{Parameter Interpretation}
We refer to \(a\) as an \emph{interior location} parameter and \(b\) as an \emph{interior concentration} parameter. Together with \((\mu,\sigma)\), these transformation parameters determine the shape of the density on the interior \((0,1)\). Proposition~\ref{prop:zoc-tn-interpretation} provides an interpretation of the parameters $a$ and $b$  through the interior distribution on the logit scale. In particular, using the relation
\[
\logit(Y)=a+b\,\logit(Z)~~ \Leftrightarrow ~~\frac{Y}{1-Y} = e^a\left(\frac{Z}{1-Z}\right)^b \quad\text{on }\{0<Y<1\},
\]
we see that, for a fixed $b$, $a<0$ shifts the interior distribution toward $0$, while $a > 0$ pushes it to $1$. Adjusting $b$, meanwhile, controls the central concentration of the interior density: values \(0<b<1\) pull interior observations toward the center $0.5$, while values \(b>1\) push them toward the boundaries. The stochastic ordering result further clarifies that increasing \(a\) moves the entire interior distribution upward in a strong sense, namely by first-order stochastic dominance. %Likewise, the median formula shows that the transformation preserves the ordering structure of the latent interior variable and maps its central location directly through the affine-logit transformation.
% \footnote{This is because $Z < 0.5 \implies \left(\frac{Z}{1-Z}\right) <1 \implies \left(\frac{Z}{1-Z}\right)^b > \left(\frac{Z}{1-Z}\right)$ for $b \in (0, 1)$, and vice versa for $Z > 0.5$.}
% Thus, \(a\) shifts the interior distribution to the right or left without changing its qualitative shape on the logit scale, while \(b\) expands or compresses the spread of the interior distribution around its center.

\begin{proposition}[Interpretation of the transformation parameters on the interior]
\label{prop:zoc-tn-interpretation}
Define $Z$ and $Z^*$ as in \eqref{def_Z}, and let $Z^\circ$ denote an interior draw from $Z$, $Z^\circ \sim Z\mid(0<Z<1)$. Under the ZOC-TN model, define the corresponding interior response by $Y^\circ=g_{a,b}(Z^\circ) = \expit\{a+b\logit(Z^\circ)\}$, and let $
U=\logit(Z^\circ)$, $V=\logit(Y^\circ)$. Then $V=a+bU$ a.s. and the following holds:

\begin{itemize}
    % \item[(i)] \textbf{Location-scale interpretation on the logit scale.}
    % For every integer $m\ge 1$ with $\mathbb E|U|^m<\infty$,
    % $
    % \mathbb{E}[V]
    % =
    % a+b\,\mathbb{E}[U]
    % $ holds.
    % Moreover, for every integer $m\ge 2$ such that $\mathbb E|U|^m<\infty$,
    % $
    % \mathbb E\!\left[(V-\mathbb EV)^m\right]
    % =
    % b^m
    % \mathbb E\!\left[(U-\mathbb EU)^m\right].
    % $
    % In particular,
    % $
    % \operatorname{Var}(V)
    % =
    % b^2\operatorname{Var}(U).
    % $

    \item %\textbf{Monotonicity in the shift parameter $a$.}
    For fixed $b>0$, if $a_2>a_1$, then
    $
    Y^\circ_{a_2,b}
    \;\ge_{\mathrm{st}}\;
    Y^\circ_{a_1,b}.
    $
    % Equivalently,
    % $
    % Y_{a_2,b}\mid(0<Y_{a_2,b}<1)
    % \;\ge_{\mathrm{st}}\;
    % Y_{a_1,b}\mid(0<Y_{a_1,b}<1).
    % $
    That is, increasing \(a\) shifts the interior conditional law to the right in the sense of first-order stochastic dominance.

    % \item \textbf{Median transformation.}
    % If the median $m_U$ of $U$ exists and is unique, then the median of
    % $Y^\circ$, equivalently the conditional median of $Y\mid(0<Y<1)$, is
    % $
    % m_Y=\expit\!\bigl(a+b\,m_U\bigr).
    % $
\end{itemize}
\end{proposition}
A proof of Proposition~\ref{prop:zoc-tn-interpretation} can be found in Appendix~\ref{app:proof_of_zoctn_interpretation}. Figure~\ref{fig:zoctn_graphical_distributions} in Appendix~\ref{app:density_shapes} illustrates the effects of the four parameters on the ZOC-TN density. We see that increasing $a$ results in a rightward shift of the center of mass in the interior, while increasing $b$ produces a displacement effect away from the center and toward the boundaries. This breadth of qualitative behaviors suggests that the affine-logit transformation provides a flexible low-dimensional mechanism for reshaping the interior density on \((0,1)\). This is the subject of the following subsection.

\subsection{Approximation and Extension to General Interior Distributions}
The affine-logit transformation defined in \eqref{eq:zoctn_affine} can be viewed as a parsimonious member of a much broader class of smooth monotone transformations on \((0,1)\). If one allows for arbitrary increasing functions \(h\), transformations of the form
\[
T_h(z)=\mathrm{expit}\!\bigl(h(\mathrm{logit}\, z)\bigr)
\]
are sufficiently rich to generate, in principle, any strictly positive continuous density on \((0,1)\). The ZOC-TN specification corresponds to the special case in which \(h\) is affine, \(h(t)=a+bt\). A direct extension would be to replace the affine map \(h(t)=a+bt\) by a more flexible monotone transformation, for example using constrained Bernstein polynomials, monotone \(I\)-splines, integrated exponentiated spline bases, or monotone rational-quadratic splines. These parameterizations preserve the increasing nature of \(h\), and hence of \(T_h\), while allowing substantially richer interior distributional shapes.

%A natural extension would be to approximate h, on a suitably chosen compact logit range or after an additional rescaling, by a constrained Bernstein-polynomial expansion, with coefficient restrictions imposed to preserve monotonicity \citep{lorentz1986bernstein, farouki2012bernstein}.

The following Proposition~\ref{prop:zoc-tn-local-density}, which is proved in Appendix~\ref{app:local_density_proof}, shows that the affine-logit family provides a first-order local approximation, not only to the underlying transformation \(T_h\), but also to the induced interior density on compact subsets of \((0,1)\). Thus, while the ZOC-TN family is not globally universal for continuous densities on \((0,1)\), it may still be interpreted as a locally flexible low-dimensional approximation to that broader class.

\begin{proposition}[Local approximation of induced interior densities]
\label{prop:zoc-tn-local-density}
Let \(Z\) be a random variable with continuously differentiable density \(f_Z\) on \((0,1)\), and let
\[
Y_h=T_h(Z):=\mathrm{expit}\!\bigl(h(\mathrm{logit}\, Z)\bigr),
\]
where \(h : \mathbb{R} \to \mathbb{R}\) is twice continuously differentiable with \(h'(t)>0\) for all \(t\in\mathbb{R}\).
Let \(K \subset (0,1)\) be compact, and assume that \(f_Z\) is bounded and Lipschitz on a
neighborhood of \(T_h^{-1}(K) \cup T_{a_0,b_0}^{-1}(K)\), where \(a_0 := h(t_0) - h'(t_0)t_0\), \(b_0 := h'(t_0) > 0\), for some \(t_0 \in \mathbb{R}\),
and
\[
T_{a_0,b_0}(z):=\mathrm{expit}\!\bigl(a_0+b_0\mathrm{logit}\, z\bigr).
\]
Let \(f_h\) and \(f_{a_0,b_0}\) denote the densities of \(Y_h\) and \(Y_{a_0,b_0}:=T_{a_0,b_0}(Z)\), respectively. Then there exists a constant \(C_K < \infty\), depending only on \(K\), \(h\), \(f_Z\), and \(t_0\), such that
\[
\sup_{y\in K}\bigl|f_h(y)-f_{a_0,b_0}(y)\bigr|
\le
C_K
\sup_{t\in \mathrm{logit}(T_h^{-1}(K)\cup T_{a_0,b_0}^{-1}(K))}|t-t_0|.
\]
In particular, the affine-logit ZOC-TN transformation induces a first-order local approximation to the target interior density generated by any smooth monotone logit-scale transformation \(h\).
\end{proposition}

% Proposition~\ref{prop:zoc-tn-local-density} highlights the value of the ZOC-TN model in the case of modeling empirical distributions where the exact form of the underlying generating process is not known \textit{a priori}. In this way the ZOC-TN distribution represents a reasonable ``one-size-fits-all" modeling choice that can be fit to a wide range of fractional response datasets. 

% (in our regression we include an intercept term to get $\hat{\mathbf{x}}_i = (1 \; \mathbf{x}_i)^\top$)
\subsection{Linear Regression ZOC-TN Model}\label{subsec:zoctn_independent_linear}
To use the ZOC-TN likelihood for regression modeling, we assume that there are $N$ independent samples $(Y_i, \mathbf{X}_i)_{i=1}^N$ consisting of labels $Y_i \in [0, 1]$ and covariates $\mathbf{X}_i \in \mathbb{R}^d$. We then assume that
\begin{equation}\label{eq:ind_lm}
    Y_i \sim \text{ZOC-TN}(\mu_i=\mathbf{X}_i^\top\boldsymbol{\beta}, \sigma, a, b)
\end{equation}
where $\boldsymbol{\beta} \in \mathbb{R}^{d}$. The maximum likelihood estimator (MLE) for this model is given by 
\begin{align}
    (\hat{\boldsymbol{\beta}}, \hat{\sigma}, \hat{a}, \hat{b}) %&= \argmax_{(\boldsymbol{\beta}, \hat{\sigma}, \hat{a}, \hat{b})} \mathcal{L}_N(\boldsymbol{\beta} , \hat{\sigma}, \hat{a}, \hat{b}) \nonumber \\
    &= \argmax_{(\boldsymbol{\beta}, {\sigma}, {a}, {b})} \sum_{i=1}^N \ell^{\text{ZOC-TN}}(\boldsymbol{\beta}, {\sigma}, {a}, {b} \mid Y_i, \mathbf{X}_i) \label{eq:ind_ll}
\end{align}
where $\ell^{\text{ZOC-TN}}(\boldsymbol{\beta}, {\sigma}, {a}, {b} \mid Y_i, \mathbf{X}_i)=\log(f^{\mathrm{ZOC\text{-}TN}}(Y_i \mid \mathbf{X}_i^\top\boldsymbol{\beta},\sigma,a, b))$ is the log-likelihood of the $i^{\text{th}}$ observation, and $f^{\mathrm{ZOC\text{-}TN}}(y \mid \mu,\sigma,a, b)$ is given in Equation~\eqref{eq:zoctn_pdf}. The following result holds for the MLE of this model. 

\begin{theorem}[Existence and large-sample behavior of the MLE / quasi-MLE]\label{thm:zoctn_mle_qmle}
Let $\Theta \subset \mathbb{R}^d \times (0,\infty)\times \mathbb{R}\times (0,\infty)$ be compact, and suppose that $\boldsymbol{\theta}=(\boldsymbol{\beta},\sigma,a,b)\mapsto \ell^{\mathrm{ZOC-TN}}(\boldsymbol{\theta} \mid y, \mathbf{x})$ is continuous for every $(y,\mathbf{x})$, with $\sigma$ and $b$ uniformly bounded away from $0$ on $\Theta$. For i.i.d.\ observations $(Y_i,\mathbf{X}_i)_{i=1}^N$, define 
$
\mathcal{L}^{\text{ZOC-TN}}_N(\boldsymbol{\theta})=\sum_{i=1}^N \ell^{\mathrm{ZOC-TN}}(\boldsymbol{\theta} \mid Y_i, \mathbf{X}_i),$
$
\hat{\boldsymbol{\theta}}_N \in \arg\max_{\boldsymbol{\theta}\in\Theta} \mathcal{L}^{\text{ZOC-TN}}_N(\boldsymbol{\theta})
$, $\mathbf{A} =-\mathbb{E}\!\left[\nabla_{\boldsymbol{\theta}}^2 \ell^{\mathrm{ZOC-TN}}(\boldsymbol{\theta}^\star\mid Y, \mathbf{X})\right]$, and $\mathbf{B}=\mathbb{E}\!\left[\nabla_{\boldsymbol{\theta}} \ell^{\mathrm{ZOC-TN}}(\boldsymbol{\theta}^\star \mid Y, \mathbf{X})\nabla_{\boldsymbol{\theta}} \ell^{\mathrm{ZOC-TN}}(\boldsymbol{\theta}^\star \mid Y, \mathbf{X})^\top\right]$.
Assume that:
\begin{enumerate}
    \item $\mathbb{E}[\sup_{\boldsymbol{\theta}\in\Theta} |\ell^{\mathrm{ZOC-TN}}(\boldsymbol{\theta}\mid Y, \mathbf{X})|] < \infty$;
    \item The population criterion $ M(\boldsymbol{\theta})=\mathbb{E}[\ell^{\mathrm{ZOC-TN}}(\boldsymbol{\theta} \mid Y, \mathbf{X})]$ has a unique maximizer $\boldsymbol{\theta}^\star \in \Theta$, $\boldsymbol{\theta}^\star$ lies in the interior of $\Theta$, and $\ell^{\mathrm{ZOC-TN}}(\boldsymbol{\theta} \mid y, \mathbf{x})$ is twice continuously differentiable in a neighborhood of $\boldsymbol{\theta}^\star$ for almost every $(y,\mathbf{x})$;
    \item The score and Hessian admit integrable envelopes in a neighborhood of $\boldsymbol{\theta}^\star$;
    \item The matrix $\mathbf{A}$ exists and is nonsingular, and $\mathbf{B}$  exists and is finite.
\end{enumerate}

Then:
\begin{enumerate}
    \item[(i)] For every $N$, a maximizer $\hat{\boldsymbol{\theta}}_N$ exists.
    \item[(ii)] $\hat{\boldsymbol{\theta}}_N \xrightarrow{P} \boldsymbol{\theta}^\star$.
    \item[(iii)] $ \sqrt{N}\,(\hat{\boldsymbol{\theta}}_N-\boldsymbol{\theta}^\star) \xrightarrow{d} \mathcal{N}\!\left(0,\; \mathbf{A}^{-1}\mathbf{B}\mathbf{A}^{-1}\right)$.
\end{enumerate}
If the ZOC-TN model is correctly specified and identifiable, then $\boldsymbol{\theta}^\star=\boldsymbol{\theta}_0$ is the true parameter, and, by standard results, $\mathbf{A}=\mathbf{B}=\mathcal{I}(\boldsymbol{\theta}_0)$, where $\mathcal{I}(\boldsymbol{\theta}_0)$ is the Fisher information matrix at $\boldsymbol{\theta}_0$, so that
\[
\sqrt{N}\,(\hat{\boldsymbol{\theta}}_N-\boldsymbol{\theta}_0)
\xrightarrow{d}
\mathcal{N}\!\left(\boldsymbol{0},\; \mathcal{I}(\boldsymbol{\theta}_0)^{-1}\right).
\]
If the model is misspecified, then $\boldsymbol{\theta}^\star$ is the unique pseudo-true parameter, i.e.\ the unique maximizer of $M(\boldsymbol{\theta})=\mathbb{E}[\ell^{\mathrm{ZOC-TN}}(\boldsymbol{\theta} \mid Y, \mathbf{X})]$, i.e., the Kullback-Leibler projection of the true conditional law onto the ZOC-TN family.
\end{theorem}

Theorem~\ref{thm:zoctn_mle_qmle} shows that, under standard compactness, identification, and smoothness conditions, the ZOC-TN maximum likelihood estimator has the usual M-estimation properties. In particular, the estimator is consistent for the unique population maximizer and is asymptotically normal with sandwich covariance matrix. Under correct specification, this population maximizer coincides with the true parameter, and the sandwich covariance reduces to the inverse Fisher information matrix. Under misspecification, the estimator instead converges to the unique pseudo-true parameter, corresponding to the Kullback--Leibler projection of the true conditional law onto the ZOC-TN family. A proof is provided in Appendix~\ref{proof:thm_zoctn_mle_qmle}. Additionally, Proposition~\ref{prop:zoc-tn-information} in Appendix \ref{proof:zoc-tn-information} derives the Fisher information, including the contribution of individual observation types (boundary/interior).

\subsection{Comparing the ZOC-TN Model to Existing Models}\label{subsec:zoctn_vs_benchmarks}
Before moving on to empirical experiments, we present a qualitative comparison of the proposed ZOC-TN regression model to a series of established fractional response regression models including the Bernoulli quasi-likelihood \citep{Papke1996}, two-limit censored normal  \citep{Rosett1975}, zero-and-one inflated beta \citep{Ospina2010}, zero-one censored shifted gamma \citep{Sigrist2011}, and zero-one censored transformed beta \citep{Kosmidis2025} models. We refer to these likelihoods and the corresponding linear regression models with the following abbreviations: Quasi-B, ZOC-N, BE-INF, ZOC-SG, ZOC-TB. Table~\ref{tab:simstudy_models} in Appendix~\ref{app:simstudy_notation} contains the density functions associated with each of these models alongside the choice of regression parameter.

Compared to these existing fractional response models, the ZOC-TN model offers several structural advantages in terms of flexibility. We now highlight some of these key differences. For fixed \(\mathbf{x}\) and \(\mu=\mathbf{x}^\top\boldsymbol{\beta}\), the following properties hold:

% \begin{proposition}[Where ZOC-TN strictly improves on standard benchmarks]
% \label{prop:zoc-tn-benchmark-dominance}

% Let \(\mathcal P_{\mathrm{TN}}(x)\), \(\mathcal P_{\mathrm{N}}(x)\), \(\mathcal P_{\mathrm{SG}}(x)\), and \(\mathcal P_{\mathrm{TB}}(x)\) denote the conditional model classes generated, respectively, by the ZOC-TN, ZOC-N (Tobit), censored shifted gamma, and censored transformed beta likelihoods at covariate value \(x\).

\begin{enumerate}
    \item[1)] \textbf{ZOC-TN vs. Quasi-B:} 
    The Bernoulli quasi-likelihood model specifies only a conditional mean, $\mathbb{E}[Y_i \mid \mathbf{x}_i] = \expit(\mathbf{x}_i^\top \boldsymbol{\beta})$, rather than a full predictive distribution on \([0,1]\). Moreover, due to the restricted range of the logit-link, $\expit:\mathbb{R}\to(0, 1)$, it is unable to predict boundary values. The ZOC-TN likelihood has full support over $[0, 1]$, with positive probability at $0$ and $1$, and it models the full distribution.

    \item[2)] \textbf{ZOC-TN vs. ZOC-N:}  
    The ZOC-TN model decouples the boundary fit from the interior shape. In particular, under the ZOC-TN model, the boundary masses depend only on \((\mu,\sigma)\), while the transformation parameters \((a,b)\) affect only the interior law. Thus, \((a,b)\) allow for reshaping of the interior density without altering the fitted masses at \(0\) and \(1\). This is impossible in the Tobit/ZOC-N model, which has no additional interior-shape parameters once \((\mu,\sigma)\) are fixed.

    \item[3)] \textbf{ZOC-TN vs. ZOC-SG, ZOC-TB, BE-INF}: Unlike the ZOC-SG, ZOC-TB, and BE-INF distributions which have at most one interior stationary point, ZOC-TN can have multiple ones. This provides one source of additional flexibility for the ZOC-TN regression model. This claim is formally proved in Appendix~\ref{app:zoctn_vs_benchmark_proof}.
\end{enumerate}
% \end{proposition}

\newcolumntype{D}[1]{>{\centering\arraybackslash}p{#1}}
\newcolumntype{Y}{>{\raggedright\arraybackslash}X}

\newcommand{\typentry}[3]{%
\begin{minipage}[c]{0.35\linewidth}
    \centering
    \includegraphics[width=\linewidth]{#1}
\end{minipage}%
\hfill
\begin{minipage}[c]{0.6\linewidth}
    \raggedright
    \small
    \textit{#2}: #3
\end{minipage}%
}

\begin{table}[ht!]
\centering
\caption{Qualitative comparison of models for different fractional response distributions. Yellow bars represent atoms at $0$, $1$, shaded purple curves show interior density.}
\label{tab:qualitative_fractional_response}
\scriptsize
\begin{tabularx}{\textwidth}{|Y|D{0.9cm}|D{0.9cm}|D{0.9cm}|D{0.9cm}|D{0.9cm}|}
\toprule
\textbf{Type (shape)}
& \textbf{ZOC-N}
& \textbf{BE-INF}
& \textbf{ZOC-SG}
& \textbf{ZOC-TB}
& \textbf{ZOC-TN} \\
\midrule\midrule
\typentry
{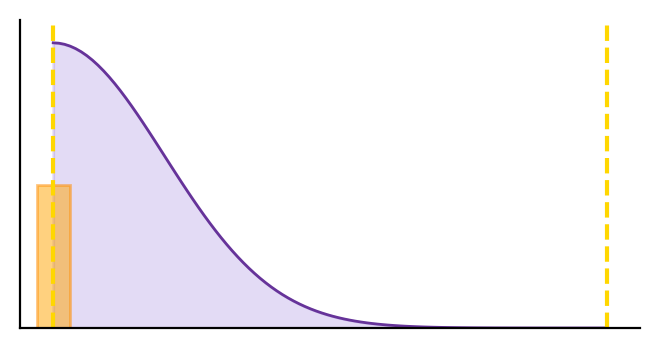}
{Lopsided}
{Nonzero and vanishing point-masses on opposing boundaries, nonzero density over interior.}
& \cmark& \cmark& \cmark& \cmark& \cmark \\\midrule
\typentry
{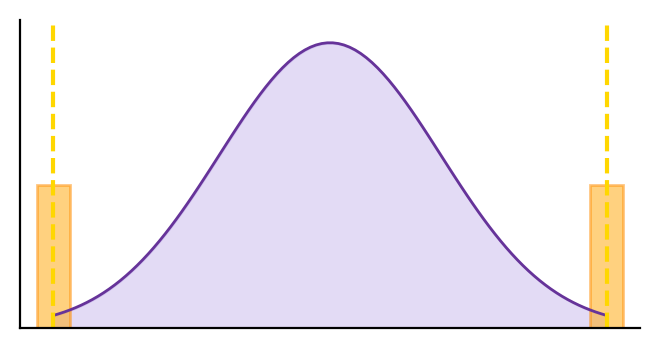}
{Trident}
{Point-masses on boundaries, bell-shaped density over interior.}
& \cmark& \cmark& \cmark& \cmark& \cmark \\\midrule
\typentry
{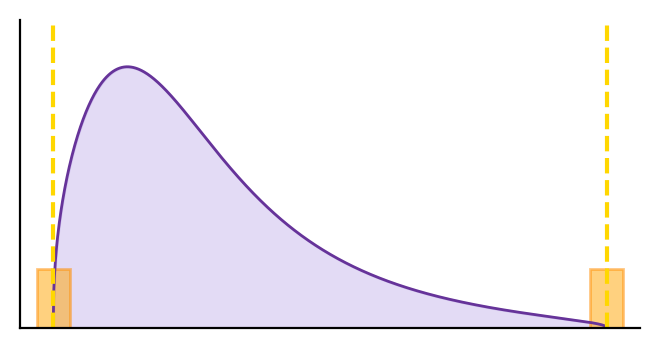}
{Shifted Trident}
{Point-masses on boundaries, off-centre single peak over interior.}
& \cmark& \cmark& \cmark& \cmark& \cmark \\\midrule
% \typentry
% {Figures/Simulation_Study/example_slope.png}
% {Slope}
% {Point-masses on boundaries, non-zero monotone sloping density over interior.}
% & \cmark& \cmark& \cmark& \cmark& \cmark \\ \midrule
\typentry
{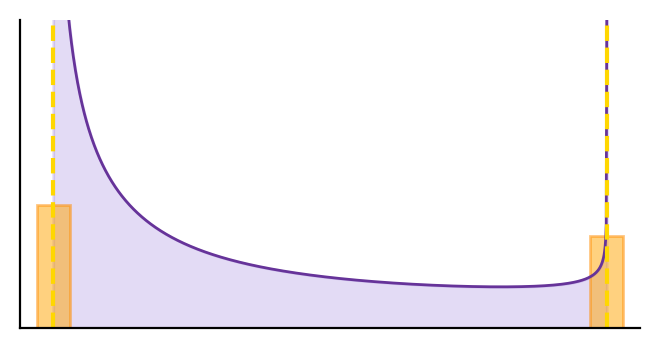}
{Asymmetric U}
{Point-masses on boundaries, non-symmetric U-shaped density over interior.}
& \xmark& \cmark& \xmark& \cmark& \cmark \\\midrule
\typentry
{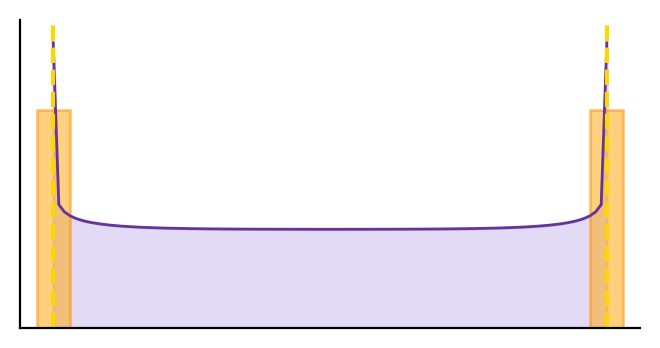}
{U}
{Point-masses on boundaries, U-shaped density over interior.}
& \xmark& \cmark& \xmark& \cmark& \cmark \\\midrule
\typentry
{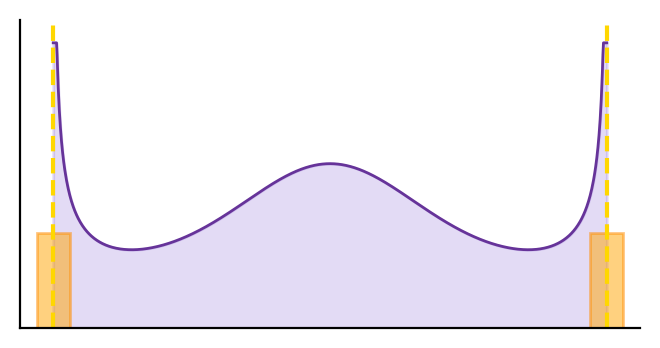}
{W}
{Point-masses on boundaries, ``trimodal'' nonzero density over interior.}
& \xmark& \xmark& \xmark& \xmark& \cmark \\\midrule
\bottomrule
\end{tabularx}
\end{table}

Table~\ref{tab:qualitative_fractional_response} visualizes some of the properties mentioned above by showcasing a selection of distribution shapes for fractional data. These categories are not exhaustive but are intended to represent some of the qualitative distribution types each likelihood is able to reproduce. Of all the models considered, only the ZOC-TN distribution is able to capture all distribution shapes shown in the table.

\section{Simulation Study}\label{sec:sim_study}
In this section, we conduct a simulation study to analyze how the ZOC-TN model compares to existing methods. The fractional response models we choose to examine in this section are the same as those considered in the last section (see Table~\ref{tab:simstudy_models} in Appendix \ref{app:simstudy_notation}). The models are compared across three censored and two mixture likelihoods. Table~\ref{tab:simstudy_dgps} and Appendix~\ref{app:simstudy_mds} contain detailed descriptions of the five data generating processes (DGPs), and Figure~\ref{fig:simulated_empirical_distribution} in Appendix~\ref{app:simstudy_mds} illustrates the five DGPs. The DGPs are chosen to provide targeted stress tests under correct specification, distributional misspecification, boundary-mass misspecification, and multimodal interior behavior.

We simulate a total of $100$ datasets for each DGP, comprising $1,000$ training observations and $1,000$ test observations. We compare the models' goodness-of-fit using the Bayesian information criterion (BIC) on the training data. Moreover, we assess the out-of-sample prediction accuracy on the test data using the mean square error (MSE), the negative log-likelihood (log score), the continuous ranked probability score (CRPS), probability integral transform (PIT) reliability diagrams, and absolute errors of boundary probabilities $\mathrm{AE}(p_0)$ and $\mathrm{AE}(p_1)$. For the test MSE, the mean of the predictive distribution is used as a point prediction. Both the log score and the CRPS measure the accuracy of probabilistic predictions, while PIT diagrams and $\mathrm{AE}(p_0)$ and $\mathrm{AE}(p_1)$ measure the calibration. The latter two measure the calibration of the models at the boundaries. They are defined as $\mathrm{AE}(p) = \sum_{i=1}^{1000} |\hat{p}(\hat{\mathbf{x}}_i) - p(\hat{\mathbf{x}}_i)|$. In addition, we measure the wall-clock time in seconds for estimating the models. The simulation study was conducted in Python using the \texttt{JAX} library to facilitate automatic differentiation of the log-likelihood functions; minimization was done using SciPy's \texttt{L-BFGS-B} implementation. All runtimes reported in this paper are based on wall clock times on an Apple M1 Pro MacBook with 16 GB of random-access memory. The code for reproducing the results presented in this paper is available at \if1\blind{ \url{https://github.com/EnCue/zoctn_paper}}\else BLIND PLACEHOLDER \fi.
% \citep{virtanen2020scipy}
% \citep{jax2018github} 

\subsection{Simulation Study Results}
Table~\ref{tab:simstudy_summary_all} in Appendix \ref{app:simstudy_vanillapit} reports average performance measures and standard errors for each model across the $100$ trials. The Bernoulli QMLE achieves the lowest MSE across all datasets, which is not surprising given its focus solely on the conditional mean. Among the models that produce a full predictive distribution, the ZOC-TN model delivers the best or second-best performance across all data regimes with respect to the BIC, log score, and CRPS. In the cases where it is second-best, the ZOC-TN model offers the best results among all misspecified models, except under the right-skewed mixture distribution, where the BE-INF model achieves a lower BIC and log score. The latter is due to the fact that the beta interior of the BE-INF model can be made to match the single-component beta mixture of the $\text{MD}_{\text{R}}$ DGP. However, the lack of accuracy of the BE-INF model at the boundaries is reflected in the average $\mathrm{AE}(p_0)$ and $\mathrm{AE}(p_1)$ scores, for which the ZOC-TN model achieves lower values. We also find that the censored beta and gamma models have substantially longer runtimes than the ZOC-TN model. Specifically, the ZOC-TN model provides an average $8.6$x speed-up relative to the censored beta model, and an even larger speed-up compared to the censored gamma model. This is due to the need for evaluating the regularized incomplete beta and lower incomplete gamma functions to compute boundary probability masses in the ZOC-TB and ZOC-SG models. % Overall, the Bernoulli QMLE converges fastest across all DGPs, followed by the two-limit Tobit model, while the ZOC-TN model delivers runtimes on par with the inflated beta model.
Moreover, the probability integral transform (PIT) reliability residual diagrams in Figure~\ref{fig:simstudy_pitresidual} in Appendix \ref{app:simstudy_vanillapit} similarly show robust calibration of the ZOC-TN model across all DGPs. These diagrams plot the difference between empirical CDF of the pooled test-sets and their predicted percentiles under each model. In all cases the ZOC-TN model produces a curve close to the horizontal zero line, unlike the BE-INF model, which struggles on the W-Shape MD, and two-limit Tobit model, which falters on the ZOC-TN and right-skewed DGPs.

In summary, the main practical implication from this simulation study is that the preferred model depends on the forecasting target. The Bernoulli quasi-likelihood is attractive when the interest is restricted to conditional-mean predictions, as reflected by its low MSE and short runtime. For distributional forecasting, however, the ZOC-TN model gives the most stable overall performance. It is best or near-best among all considered scenarios, and it is much less costly to estimate than the censored beta and shifted-gamma alternatives. This makes the ZOC-TN model a useful default candidate when boundary probabilities and predictive distributions are of interest but the shape of the bounded response distribution is not known a priori.

\section{Extensions: Nonlinearities and Dependence}\label{sec:zoctn_regression}
In the following, we show how the independent ZOC-TN linear regression model can be extended (i) to account for nonlinearities and interactions in a machine learning framework and (ii) to explicitly model unaccounted spatial and spatio-temporal dependence using Gaussian processes. 

\subsection{Data-driven Nonlinearities and Interactions with Tree-boosting}\label{subsec:zoctn_independent_tb}

We first extend the linear regression model to account for nonlinearities and interactions through the use of gradient tree-boosting \citep{friedman2001greedy}, which is a widely used machine learning technique that achieves state-of-the-art prediction accuracy on tabular data \citep{januschowski2022forecasting}. Tree boosting can learn nonlinearities, discontinuities, and complex interactions in a data-driven manner. It is also invariant to monotone transformations of predictor variables, can accommodate missing predictor values without explicit imputation, and is insensitive to multicollinearity and outliers among predictors. Specifically, we assume
\begin{equation}\label{eq:ind_tbm}
    Y_i \sim \text{ZOC-TN}(\mu_i=F(\mathbf{X}_i), \sigma, a, b),
\end{equation}
where \(F \in \mathcal{H}\), and \(\mathcal{H}\) is the linear span of a set of base learners. In our case, the base learners consist of regression trees \(f:\mathbb{R}^d \to \mathbb{R}\). The function \(F\) is estimated by minimizing the empirical negative log-likelihood
\[
    \mathcal{R}(F,\sigma,a,b)=-\sum_{i=1}^N  \ell^{\mathrm{ZOC-TN}}(\mu_i=F(\mathbf{X}_i),\sigma,a,b\mid Y_i).
\]
Starting from an initial constant prediction \(F_0\), gradient boosting constructs an additive expansion
\[
    F_M(\mathbf{x})
    =
    F_0(\mathbf{x})
    +
    \sum_{m=1}^M \nu f_m(\mathbf{x}),
\]
where \(f_m\) is a regression tree fitted at iteration \(m\), and \(\nu\in(0,1]\) is a learning rate. At each iteration, the tree is fitted to the current likelihood gradients with respect to the location predictor \(\mu_i=F(\mathbf{X}_i)\). In a second-order implementation \citep{sigrist2021gradient}, both first and second derivatives of the ZOC-TN log-likelihood with respect to \(\mu_i\) are used to form a local quadratic approximation to the objective. The required derivatives are given in Appendix~\ref{app:zoc_tn_boosting_gradients}.

\subsection{Modeling Dependence with Gaussian Process Models}\label{subsec:zoctn_gpms}
Next, we relax the conditional independence assumptions in both linear regression and tree-boosted ZOC-TN models described so far. We do this by adding a latent frailty in the form of a zero-mean Gaussian process, to model spatial and spatio-temporal variability which is not accounted for by observable predictor variables. Specifically, we replace Equations \eqref{eq:ind_lm} and \eqref{eq:ind_tbm} with
\begin{equation}\label{eq:gpm}
    Y_i \sim \text{ZOC-TN}(\mu_i=F(\mathbf{X}_i) + \mathcal{G}(\mathbf{s}_i), \sigma, a, b)
\end{equation}
for a linear or tree-boosted regression function $F:\mathbb{R}^d \to \mathbb{R}$, and Gaussian process $\mathcal{G}(\mathbf{s}_i)$ depending on input locations $\mathbf{s}_i\in \mathbb{R}^p$. In our application below, we consider spatial and spatio-temporal Gaussian processes, i.e., $\mathbf{s}_i$ consists of spatial coordinates and time. But, in general, $\mathbf{s}_i$ could also consist of other covariates. The zero-mean latent Gaussian process is specified by a parametric covariance function $\mathrm{Cov}(\mathcal{G}(\mathbf{s}_i), \mathcal{G}(\mathbf{s}_j)) = c_{\boldsymbol{\gamma}}(\mathbf{s}_i, \mathbf{s}_j)$, where $\boldsymbol{\gamma} \in \mathbb{R}^{n_{\mathrm{cov}}}$ are the GP's hyperparameters. 

For such Gaussian process models, the joint marginal likelihood no longer factorizes as a product of individual observations, as in Equation~\eqref{eq:ind_ll} and is given by
\begin{align} \label{eq:gp_likelihood}
    \mathcal{L}(F, \boldsymbol{\gamma}, \boldsymbol{\alpha}) 
    &= \int \left(\prod_{i=1}^N \mathcal{L}^{\text{ZOC-TN}}(F(\mathbf{X}_i) + \mathcal{G}(\mathbf{s}_i), \sigma, a, b)\right) p(\boldsymbol{\mathcal{G}} | \boldsymbol{\gamma}) d\boldsymbol{\mathcal{G}}
\end{align}
where $\boldsymbol{\mathcal{G}} = (\mathcal{G}(\mathbf{s}_1), ..., \mathcal{G}(\mathbf{s}_N))^\top$, $\boldsymbol{\alpha} = (\sigma, a, b)$ are the remaining ZOC-TN parameters, $p(\cdot \mid \boldsymbol{\gamma})$ is the density of the GP, and $\mathcal{L}^{\text{ZOC-TN}}(\cdot)$ is the ZOC-TN likelihood. Implicit in \eqref{eq:gp_likelihood} is the dependence of our chosen regression function $F$ on a set of parameters that must also be learned from the data. For the two types of regression functions used in this paper, linear models depend simply on the set of linear coefficients, $\boldsymbol{\beta}$, while in the case of tree-boosting models, parameters include the structure and leaf weights of each tree in the ensemble. 

We train these models in the \texttt{GPBoost} Python package which handles estimating the set of all optimal parameters, $(\hat{F}, \hat{\boldsymbol{\gamma}}, \hat{\boldsymbol{\alpha}})$, simultaneously by maximizing the Laplace-approximated marginal likelihood $\mathcal{L}^{\text{ZOC-TN}}(\cdot)$. In the case of boosting models, training is done via a type of functional gradient descent referred to as GPBoost algorithm, which iteratively alternates between optimizing the regression function $F$, and the latent GP covariance parameters $\boldsymbol{\gamma}$; see  \citep{Sigrist2022GPBoost} for more details. In addition to the Laplace approximation, we approximate the Gaussian process $\boldsymbol{\mathcal{G}}\in \mathbb{R}^N$ using a Vecchia approximation \citep{Vecchia1988}. This approximation works by simplifying the dependence structure of the coupled random effects to interactions only between a local neighborhood of each latent variable. That is,
% ?: j \ne i ?
\begin{align*}
    p(\boldsymbol{\mathcal{G}} | \boldsymbol{\gamma}) &= \prod_{i=1}^N p(\mathcal{G}_i | (\mathcal{G}_1, ..., \mathcal{G}_{i-1}), \boldsymbol{\gamma}) \approx \prod_{i=1}^N p(\mathcal{G}_i | (\mathcal{G}_j)_{j \in N_{\tilde{m}}(\mathcal{G}_i)}, \boldsymbol{\gamma})
\end{align*}
where $N_{\tilde{m}}(\mathcal{G}_i)$ is the collection of the nearest ${\tilde{m}}$ neighbors of $\mathcal{G}_i$. This produces a sparse representation of the latent GP's precision matrix, which incurs a computational cost of order $\mathcal{O}(N\tilde{m}^3)$. In our application, we set ${\tilde{m}}=20$. For a detailed discussion on the Vecchia-Laplace approximations, we refer to \citet{Kundig2024}.

\section{Application to Mortgage Loss Given Default Modeling}\label{lgd_data}
In the following, we consider the task of modeling the loss given default (LGD) of U.S. 30-year fixed-rate mortgages on single-family homes as issued by the Federal Home Loan Mortgage Corporation, also known as (and referred to hereinafter as) Freddie Mac. As a government-sponsored enterprise (GSE), Freddie Mac maintains the publicly available Single Family Loan-Level Dataset (SFLLD) comprising mortgages purchased or guaranteed by the GSE from January 1, 1999, through June 30, 2025. In addition to providing information on the origination of the approximately 48.3 million mortgages, Freddie Mac publishes a table of monthly performance logs for each contract over its lifespan. The loss given default is a natural candidate for fractional response models. When the ZOC-TN model, or another censored model, is applied to modeling LGD data, we can interpret the uncensored normal variable $Z^*$ as a \textit{loss potential} attached to a loan. In this way, although the LGD is between $[0, 1]$, the loss potential can exceed these bounds to represent the financial outlook of the loan at the time of default. 

 %A similar interpretation is common in meteorological forecasts of precipitation accumulation, wherein rainfall is modeled through a censored latent random variable even as negative rainfall is physically impossible \citep{Scheuerer2015}.

\subsection{Default and Loss Given Default Definitions Used}\label{subsec:lgd_def}
Because Freddie Mac provides no explicit indication whether a loan defaulted or on the LGD incurred by a defaulted mortgage, we use the following definitions. First, we define a mortgage to be in default only once it has been delinquent for 6 months, or at the earliest date available after this time. This default definition is consistent with U.S. Federal Reserve’s Code of Federal Regulations and the Basel framework. %This 180-day delinquency cut-off is consistent with the Basel Framework, which provides national discretion in the case of default of retail or public sector entity obligations. Hence,  %Although a simplification of the official default monitoring process, which relies on the specific circumstances of a given mortgage, at the discretion of ``prudent business judgment" \citep[§~9202.1(a)]{freddiemac9202.1}, this provides us with a reasonable operational definition for default, subject to the following consideration for calculating LGD.
% \citep{FedDefault} which states that “A retail exposure of a Board-regulated institution is in default if: [the] exposure is 180 days past due, in the case of a residential mortgage exposure”
 % \citep{BaselFrameworkCRE3668, BCBSNationalDiscretions}
The LGD is defined as the ratio of the net loss incurred through the foreclosure process (\texttt{actual\_loss}) and the unpaid balance left on the mortgage at the time of default (\texttt{current\_upb}). If recoveries exceed the unpaid principal at default, or additional losses are incurred and not offset (i.e., measured LGD is not in [0, 1]), we censor these observations to the nearest boundary. I.e., the LGD is calculated as
\begin{equation}\label{eq:lgd_def}
    \text{LGD} = \max\left(\min\left( \frac{\texttt{actual\_loss}}{\texttt{upb\_at\_default}}, 1\right), 0\right).
\end{equation}
In addition, Freddie Mac only reports the actual loss on a subset of mortgages, specifically those terminated with a \textit{zero-balance code} (ZBC) relating to one of the following events: third-party sale, short sale or charge-off, REO disposition, or whole loan sale. Hence, we further restrict our dataset to defaulted mortgages with one of these four termination events. It should be noted that although \texttt{actual\_loss} is measured only after the foreclosure process has been completed – often multiple quarters after the mortgage defaults – we do not apply explicit discounting in the calculation of \eqref{eq:lgd_def}. This is because Freddie Mac includes a term for delinquent accrued interest in the calculation of actual loss that we treat as a sufficient offsetting term. %This is broadly consistent with the literature \citep{Leow2012, An2021}.

\subsection{Dataset Description}\label{subsec:dataset_desc}
The dataset we use contains a total of $N_{\text{total}} = 366,859$ defaulted mortgages, which have originated between January 1999 and December 2022 across the mainland United States and defaulted between January 2000 and December 2022. We exclude from our dataset mortgages defaulting in the years 2023, 2024, and 2025 due to delays in the foreclosure process. This is because the \texttt{actual\_loss} variable that is needed for our calculation of LGD only becomes available once loan termination is finalized. The median liquidation time across the dataset, measured as the number of months between default and final loan termination, was a delay of 15 months (empirical distribution shown in Appendix~\ref{app:data_exploration}, Figure~\ref{fig:lgd_liquidationdist}). %To avoid potential biases from including only mortgages with shorter foreclosure timelines, we only consider defaults occurring a full three years prior to the writing of the paper.

As covariates, we use a set of loan origination features and macroeconomic indicators taken from the year prior to default summarized in Table \ref{tab:predictors} in Appendix \ref{app:data_exploration}. Mortgage-specific information such as loan value, borrower credit score, and default insurance coverage are all taken from the SFLLD directly. These features are augmented by state and national macroeconomic statistics sourced from public datasets maintained by the U.S. Bureau of Economic Analysis. National interest rate data is also included in the calculation of \texttt{ir\_spread}, taking the difference between a given mortgage rate and the average 30-year fixed-rate offered by Freddie Mac that year, as reported by the Federal Reserve Bank of St. Louis. For spatial coordinates in GP models, we follow \cite{Kundig2025} and map property ZIP3 codes (the first three digits of the address's postal code) to longitude and latitude centroid coordinates. In total, our dataset includes defaults occurring in $874$ ZIP3 regions across the continental United States.
%\citep{FREDmortgage30us}

Figure~\ref{fig:lgd_dist} shows the empirical distribution of the LGD in our dataset. Roughly $13.1\%$ of all mortgages are associated with a boundary LGD observation (either $0$ or $1$), of which a larger number suffer a complete loss. Over the interior $(0,1)$, the LGD first decreases rapidly and reaches a trough at around $0.1$, then starts to increase peaking around $0.5$ after which it decreases again. These features, in particular the spike at the lower boundary, highlight the value of flexible distributions in modeling LGD. 
\begin{figure}[ht!]
  \centering
  \includegraphics[width=0.6\linewidth]{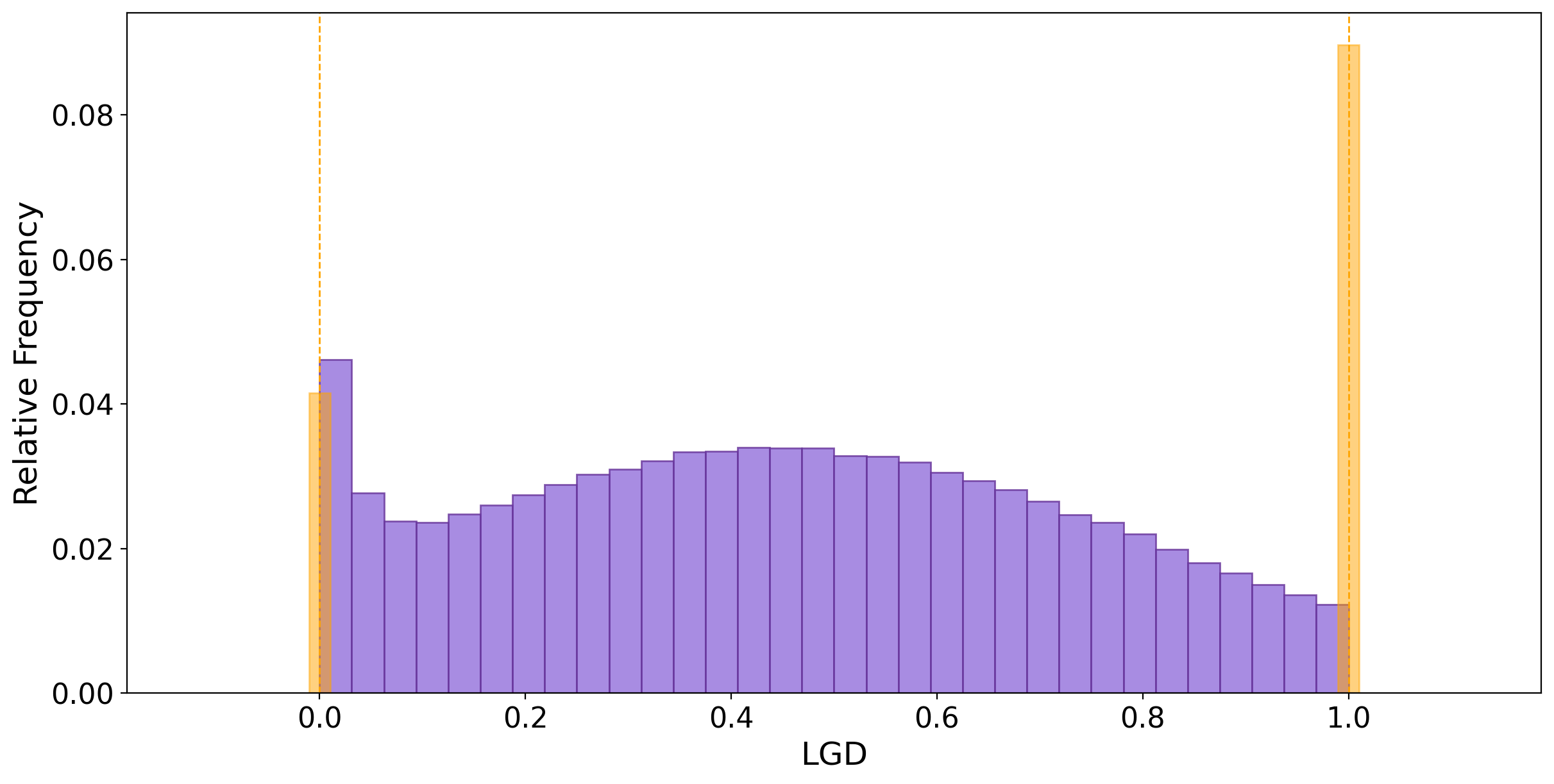}
  \caption{Empirical LGD distribution. Yellow bars represent boundary observations.}
  \label{fig:lgd_dist}
\end{figure}

Figure~\ref{fig:lgd_spatial_heatmap} contains a spatial map with average LGDs in the ZIP3 locations and a time series plot of average LGDs, and Figure \ref{fig:lgd_spatiotemporal_heatmap} in Appendix \ref{app:data_exploration} visualizes spatio-temporal average LGDs. These figures show the clear presence of trends in LGD across space and time. Concerning spatial structure, we observe areas of high regional LGD around Florida, New Mexico, and northern Appalachia into the so-called ``Rustbelt" including Ohio, Michigan, and Indiana. Moreover, we see an increase in the LGD in the fallout of the 2008 subprime mortgage crisis. There is also a very small uptick in LGD in 2021, potentially attributable to the COVID-19 pandemic. %On this we note that the absence of a spike in number of defaults in our dataset in 2020 is due to a large number of defaulted observations in the original SFLLD still being unresolved as of late 2025.
\begin{figure}[ht!]
  \centering
  \includegraphics[width=0.49\linewidth]{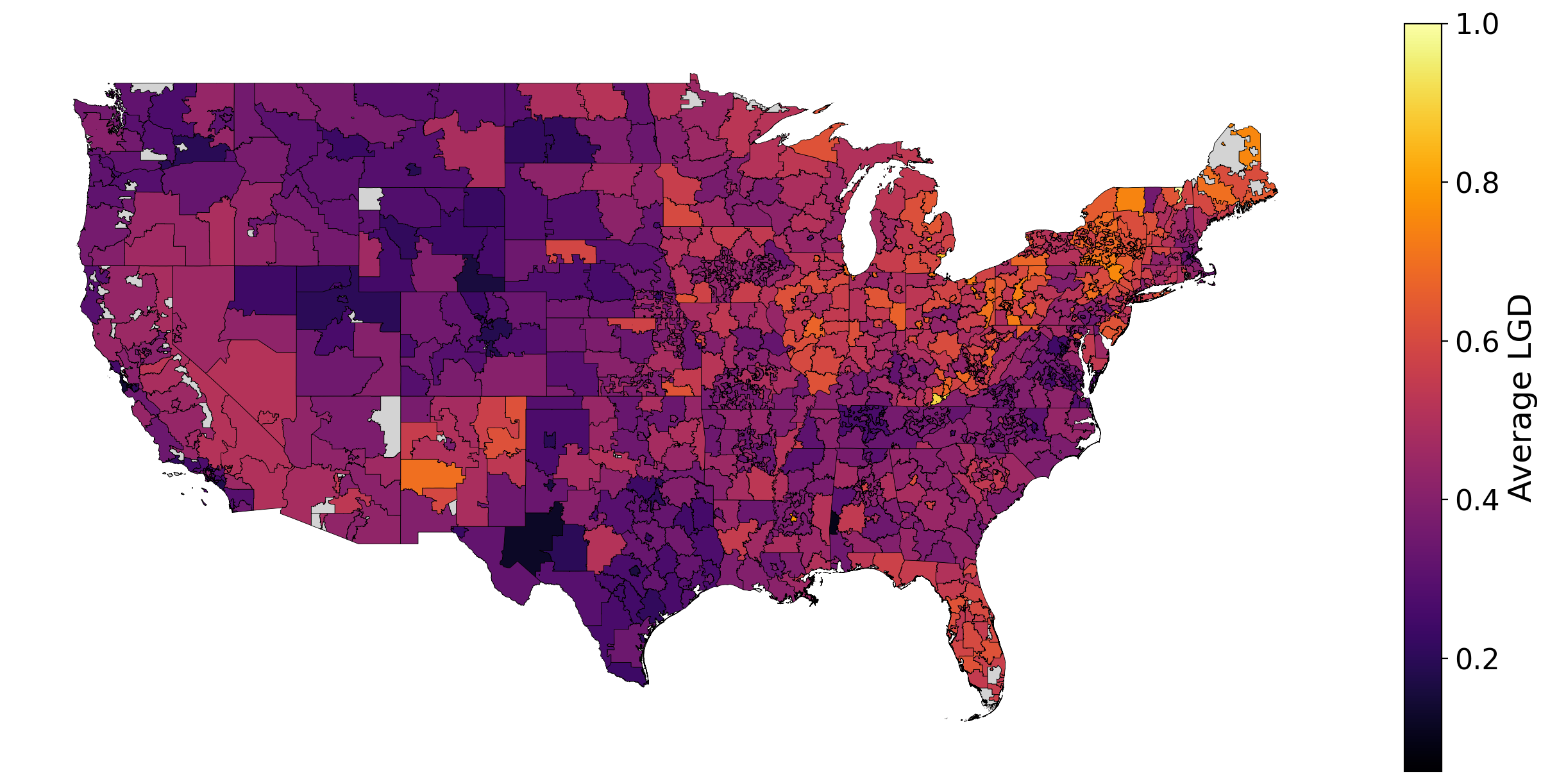}
  \includegraphics[width=0.49\linewidth]{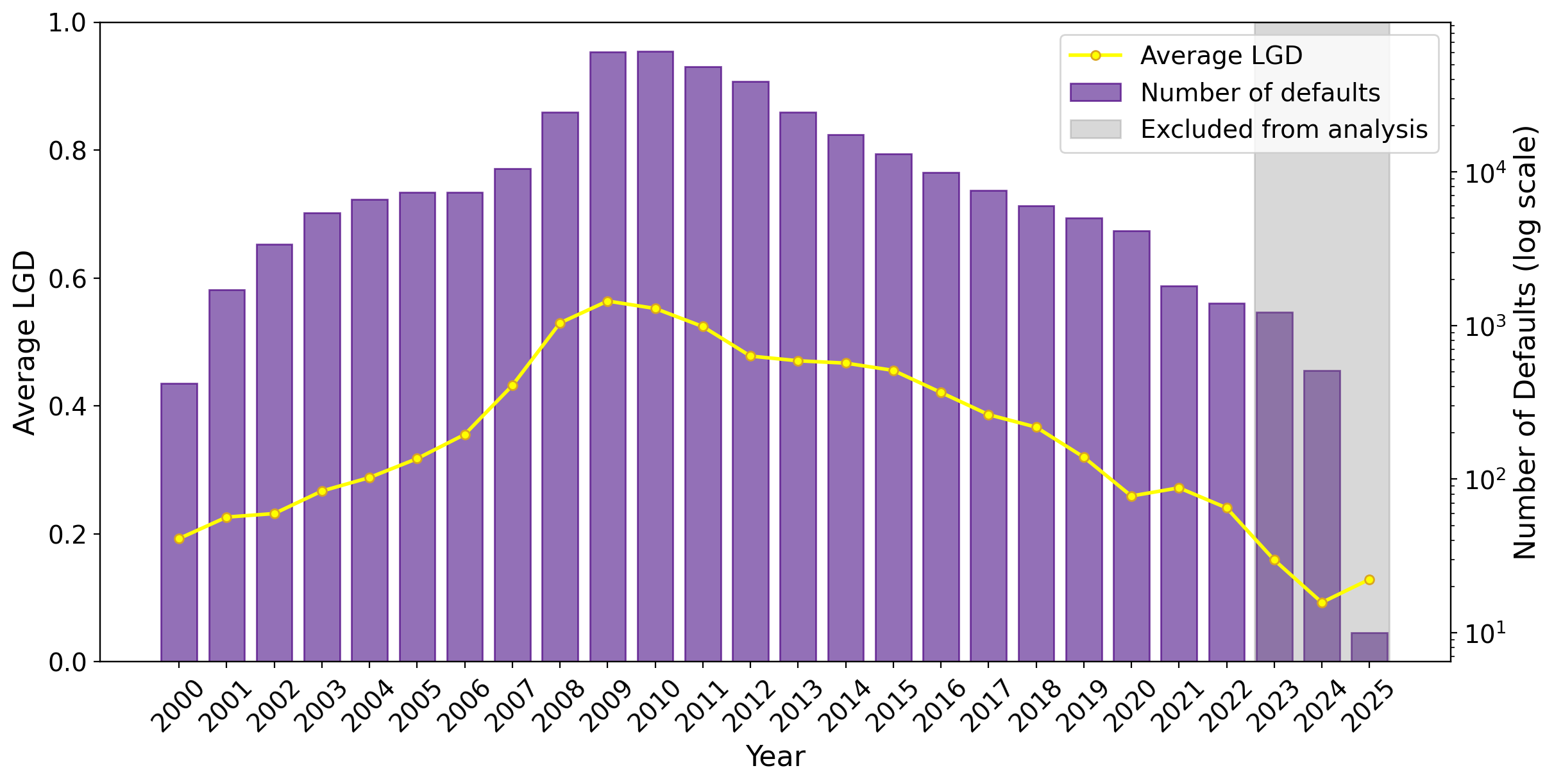}
  \caption{Left plot: Heatmap of average LGDs by ZIP3 locality (no data for gray regions). Right plot: average LGD (left axis) and number of defaults (right axis, log-scale) by year.}
  \label{fig:lgd_spatial_heatmap}
\end{figure}

\subsection{Results of Linear Regression Models}\label{subsec:stage2_results}
To start, we compare the fit of the (independent) linear regression models considered in Sections \ref{subsec:zoctn_vs_benchmarks} and \ref{sec:sim_study} and listed in Table~\ref{tab:simstudy_models} in Appendix \ref{app:simstudy_notation} on the full LGD dataset including all observations from the start of 2000 to the end of 2022. Table \ref{tab:stage1_model_performance} reports the negative log-likelihood (NLL), the Akaike information criterion (AIC), and the Bayesian information criterion (BIC). We find that the ZOC-TN clearly yields the best fit in all considered metrics. In addition, Figure~\ref{fig:stage1_aic} in Appendix \ref{app:lgd_ilms} shows how the AIC evolves as more data is introduced across the years. We observe that the ranking of the model's quality-of-fit is very stable across time. The estimated coefficients and additional parameters for each model are presented in Appendix~\ref{app:lgd_ilms}. The runtimes in Table~\ref{tab:stage1_model_performance} for fitting each of the independent linear regression models on the complete LGD dataset are qualitatively similar to those in the simulation study. We see that among the three models (ZOC-TB, ZOC-SG, and ZOC-TN) with the best goodness-of-fit (AIC, BIC), our proposed censored transformed normal distribution is clearly the fastest.% See Appendix~\ref{app:lgd_ilms} for runtimes across the expanding window dataset.
\begin{table}[ht!]
\centering
\caption{Goodness-of-fit metrics for independent linear models, based on fit of entire LGD dataset. Runtime is based on numerical minimization of NLL. NLL, AIC, and BIC rounded to nearest integer.}
\label{tab:stage1_model_performance}
\footnotesize
\begin{tabular}{lrrrr}
\toprule
  & Runtime & NLL & AIC & BIC \\
\midrule
Quasi-B & $\mathbf{0.51}$ & $232,477$ & $465,009$ & $465,312$ \\
ZOC-N & $1.24$ & $117,164$ & $234,386$ & $234,699$ \\
BE-INF & $1.29$ & $128,114$ & $256,289$ & $256,624$ \\
ZOC-SG & $60.8$ & $110,998$ &	$222,055$ &	$222,380$ \\
ZOC-TB &  $155$  & $114,238$ &	$228,535$ & $228,859$ \\
ZOC-TN &  $2.06$ &  $\mathbf{105,519}$ & $\mathbf{211,100}$ & $\mathbf{211,435}$ \\
\bottomrule
\end{tabular}
\end{table}

In addition, we present modified \textit{suspended rootograms} \citep{Kleiber_2016} in Figure~\ref{fig:stage1_rootograms} to visually inspect the goodness-of-fit of the models. These figures show in the foreground the difference in the square-root of fitted and observed frequencies across a total of 18 bins (16 evenly spaced across the interior, and one at each of the boundaries). The empirical distribution is additionally shown in the background.% at a finer resolution (32 inner bins) to highlight differences between models. The square-root scale is used to adjust for scale differences across the bins.
\begin{figure}[ht!]
  \centering
  \includegraphics[width=1\linewidth]{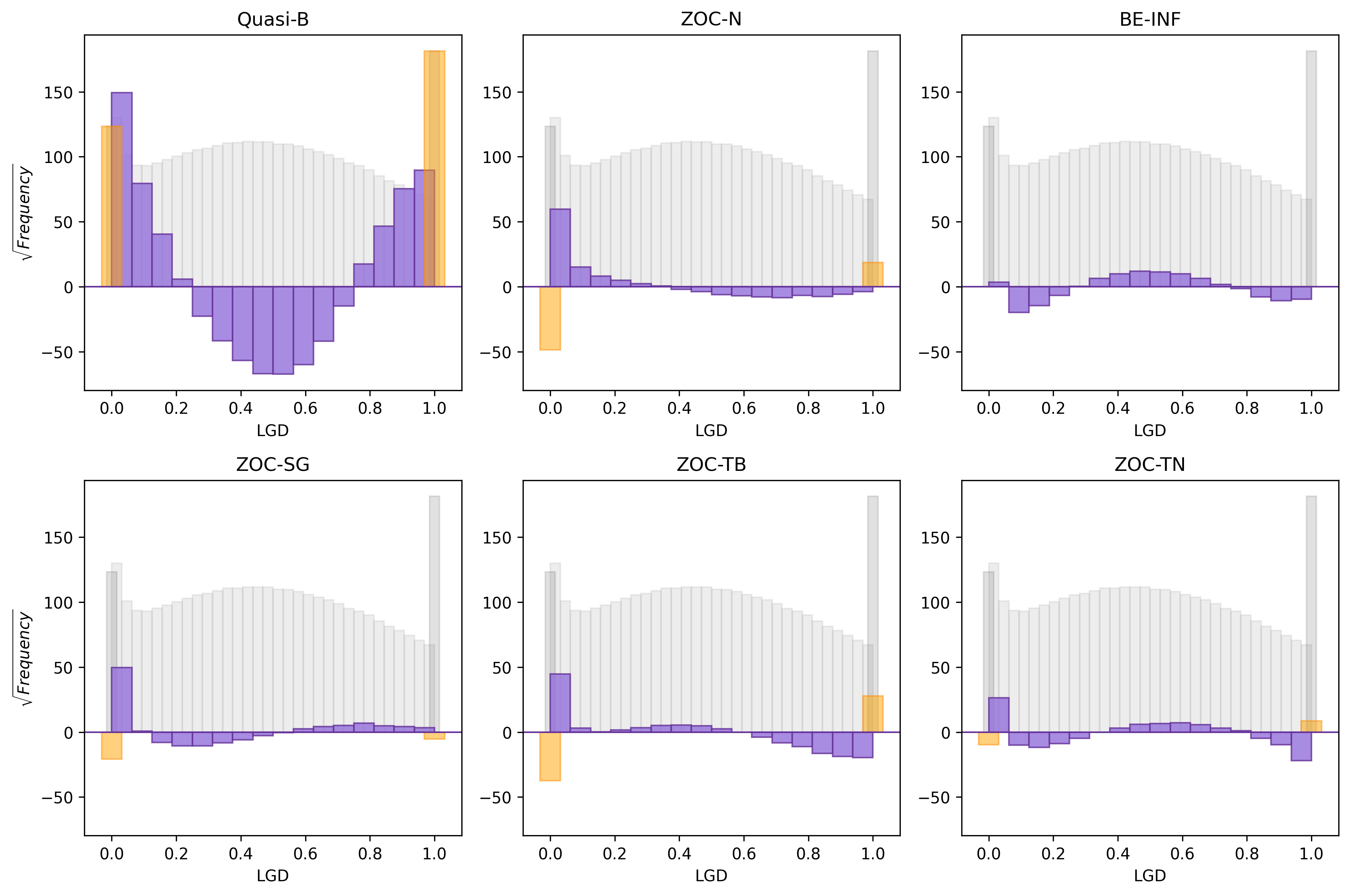}
  \caption{Suspended rootograms for independent linear models compared in Table~\ref{tab:stage1_model_performance}. Colored bars represent $\sqrt{O_j} - \sqrt{E_j}$ for $O_j$ observations in $j^{\text{th}}$ bin, and $E_j$ expected frequency under fitted models.}
  \label{fig:stage1_rootograms}
\end{figure}
Figure~\ref{fig:stage1_rootograms} shows that the ZOC-TN model achieves smaller residuals at the boundary probability masses than the other likelihoods, except for the zero-and-one inflated beta (BE-INF) model. In the interior, the zero-and-one inflated beta model and the ZOC-TN model achieve the best marginal fit. Note that the zero-and-one inflated beta model has zero residuals at the boundaries $0$ and $1$, since these probabilities are estimated globally in the BE-INF model through \((\hat{\gamma},\hat{\alpha})\), and hence fixed across all observations. This means that although at the aggregate level -- at which the suspended rootograms in Figure~\ref{fig:stage1_rootograms} are computed -- the BE-INF likelihood is able to fit the data perfectly, on an individual observation basis this inflexibility is a drawback which results in worse AIC and BIC values. The worst fit is observed for the Bernoulli quasi-likelihood which is unable to capture any nonzero boundary probability mass. %Notable also is the way that each of the likelihoods handles the spike in LGD observations just above $0$, with most models' greatest residual being in the leftmost inner bin. Of these, however, the ZOC-TN likelihood achieves the closest fit, potentially due to its unique ability to capture \emph{W}-type fractional response distributions as presented in Table~\ref{tab:qualitative_fractional_response}. Further evidence of this can be found in the negative residual on the rightmost inner bin where the empirical distribution does not match the expected right peak of a \emph{W} distribution.
In Appendix~\ref{app_tests_tobit}, we also conduct log-likelihood ratio and Wald tests to compare the vanilla two-limit Tobit model to the ZOC-TN model. Both test statistics are highly significant, thus providing strong evidence that the transformation in \eqref{eq:zoctn_affine} meaningfully improves the fit.

\subsection{LGD Prediction}
We now move on to evaluate the prediction accuracy of various models on the task of forecasting one-year-ahead LGDs. This is done by calculating one-year-ahead, temporal out-of-sample LGDs using expanding window training data sets. This means that for each model, we fit a total of $15$ models, starting in 2008, to predict the LGD in the upcoming year. For example, a model used to predict LGDs in 2008 is trained on mortgages defaulting between 2000 and the end of 2007. This rolling-origin design mimics a practical risk management application in which an LGD model is periodically re-estimated using all information available at the time and then deployed to forecast loss severities for future defaults.
% , while a model used to predict LGD in 2009 additionally contains observations from 2008 
In addition to the independent linear models considered in Section \ref{subsec:stage2_results}, we use independent tree-boosting models, Gaussian process models with linear fixed-effects regression terms, and Gaussian process tree-boosting models as presented in Section \ref{subsec:zoctn_gpms}. Specifically, we use independent tree-boosting models with Bernoulli, ZOC-TN, censored beta, and censored gamma (quasi-)likelihoods as well as spatial and spatio-temporal GP models with both linear and nonlinear tree-boosted regression terms. We restrict our analysis of the spatial and spatio-temporal GP models to the ZOC-TN likelihood. This choice is motivated by the best-in-class performance of the ZOC-TN model on the LGD dataset observed in Section \ref{subsec:stage2_results}, as well as numerical instabilities (crashes) encountered in the training of both censored beta and gamma models. The models containing a tree-boosting or GP component were trained using the \texttt{GPBoost} Python library version 1.6.7. The zero-one-inflated Beta and two-limit Tobit likelihoods are not supported in \texttt{GPBoost} at the time of writing, and so these are excluded from the analysis. The prediction accuracy is measured using three metrics: MSE, log score, and CRPS. Details for the log score and CRPS calculations can be found in Appendix~\ref{app:ls_crps}. Moreover, we assess the calibration of the predictive distribution using PIT reliability diagrams. 

\subsubsection{Choice of Hyperparameters}\label{subsec:datasplit}
For the tree-boosting models, we conduct hyperparameter tuning by further splitting the training data into validation data consisting of the final year in the training data and inner training data as the remaining data, and select hyperparameters that minimize the  mean squared error on the validation data. Table~\ref{tab:tuning_param_desc} in Appendix~\ref{app:results} lists the hyperparameters considered as well as grid values used for the randomized search. We perform randomized grid search with 100 iterations. Tuning for the spatio-temporal tree-boosted models is restricted to number of trees due to time constraints, with all other hyperparameters being carried over from the corresponding spatial-only model. Information on the tuned hyperparameters for the spatial and spatio-temporal GP tree-boosted models can be found in Figure~\ref{fig:gpb_hps} in Appendix~\ref{app:results}.
For the Gaussian process models, we use a Matérn covariance function for the spatial GP ($\mathbf{s}_i = (x_i, y_i)$),
\begin{equation}\label{eq:matern_cov}
    c_{\boldsymbol{\gamma}}(\mathbf{s}_i, \mathbf{s}_j) = \sigma_{\mathrm{Matern}}^2 \frac{2^{\nu-1}}{\Gamma(\nu)}
\left(\frac{\sqrt{2\nu}\,\lVert \mathbf{s}_i - \mathbf{s}_j \rVert_2}{\rho_s}\right)^\nu
K_\nu\left(\frac{\sqrt{2\nu}\,\lVert \mathbf{s}_i - \mathbf{s}_j \rVert_2}{\rho_s}\right)
\end{equation}
where $\Gamma(\cdot)$ is the Gamma function, and $K_{\nu}(\cdot)$ is the modified Bessel function of the second kind, and $\boldsymbol{\gamma} = \{\sigma_{\mathrm{Matern}}^2, \rho_s, \nu\}$ are the marginal variance, spatial range, and smoothness (or shape) parameters, respectively. For the spatio-temporal models, we use an anisotropic spatio-temporal Matérn covariance function:
\begin{equation*}
    c_{\boldsymbol{\gamma}}(\mathbf{s}_i,\mathbf{s}_j)
=
\sigma_{\mathrm{Matern}}^2 \frac{2^{\nu-1}}{\Gamma(\nu)}
\left(
\sqrt{2\nu}\left\|\mathbf{R}(\mathbf{s}_i-\mathbf{s}_j)\right\|_2
\right)^\nu
K_\nu\left(
\sqrt{2\nu}\left\|\mathbf{R}(\mathbf{s}_i-\mathbf{s}_j)\right\|_2
\right)
\end{equation*}
where $\mathbf{R} = \mathrm{diag}(\rho_t^{-1}, \rho_s^{-1}, \rho_s^{-1})$. Hence in the spatio-temporal case, $\boldsymbol{\gamma} = \{\sigma_{\mathrm{Matern}}^2, \rho_t, \rho_s, \nu\}$. The hyperparameters $\boldsymbol{\gamma}$  are estimated jointly with the rest of the model parameters. For the smoothness parameter, we use $\nu=1.5$. This is done since estimating the smoothness parameter requires close-by observations, while in our data coordinates are aggregated by ZIP3 region and hence lack granularity. However, a robustness check is provided in Appendix~\ref{app:gp_cov_smoothness}.

\subsubsection{Results}

Table~\ref{tab:stage2_model_performance} presents the average prediction accuracy metrics across the 15 test years (weighted by the number of observations in each year). Missing values for the tree-boosted censored beta model reflect a crash due to numerical issues. Concerning the independent linear regression models, we see that the ZOC-TN model offers competitive prediction accuracy compared to the other likelihoods. Specifically, for probabilistic predictions, the censored beta, censored gamma, and ZOC-TN independent linear regression models yield similar log score and CRPS values. And with the exception of the BE-INF model, all vanilla linear models achieve comparable MSE scores. Considering the independent tree-boosting models, we first observe that tree-boosting clearly yields more accurate predictions compared to linear models. This highlights the need for flexible nonlinear models. Moreover, the ZOC-TN model delivers the best performance of all considered independent tree-boosting models for both the log score and CRPS. We also see a noticeable improvement in accuracy by including GPs to model residual spatially and spatio-temporally structured variability. For instance, the higher accuracy of the spatial GP linear model over independent linear models suggests the presence of unaccounted spatial effects, while the difference between spatial and spatio-temporal models additionally suggests that these correlations vary from year-to-year. Finally, the tree-boosted spatio-temporal ZOC-TN model clearly yields the most accurate predictions in all three metrics. This suggests the presence of both nonlinear effects of the observable covariates and unaccounted spatio-temporal effects. Table~\ref{tab:stage2_model_performance} also reports training runtimes. Note that, for tree-boosting models, these runtimes depend heavily on the selected hyperparameters.%However, the fact that no such improvement is observed between the independent and GP tree-boosted models indicates that some of these spatial and spatio-temporal effects can be compensated by nonlinear effects of observable covariates. %For instance, the low runtime of the ZOC-SG independent tree-boosting model is slightly misleading since the number of boosting iterations was tuned independently across models; the average number of trees across the ZOC-SG models was 270 trees, while the average numbers for the ZOC-TN and Quasi-B models were 366 and 489, respectively.
\begin{table}[ht!]
\centering
\caption{Average model performance across all test years with standard deviation given in parentheses and runtimes (in seconds).}
\label{tab:stage2_model_performance}
\footnotesize
\begin{tabular}{lrrrr}
\toprule
& Runtime & MSE & Log Score & CRPS \\
\midrule
\textit{Independent Linear Models} && \\
\midrule
\rowcolor{gray!15}
Quasi-B
& $0.377$s& $0.0847$& $0.691$& $0.245$ \\& 
& {\scriptsize $(0.0206)$}& {\scriptsize $(0.0596)$}& {\scriptsize $(0.0184)$} \\
\rowcolor{gray!15}
ZOC-N
& $0.998$s & $0.0829$& $0.441$& $0.168$ \\
& & {\scriptsize $(0.0193)$}& {\scriptsize $(0.192)$}& {\scriptsize $(0.0243)$} \\
\rowcolor{gray!15}
BE-INF
& $1.05$s& $0.0800$& $0.464$& $0.311$ \\
& & {\scriptsize $(0.0123)$}& {\scriptsize $(0.174)$}& {\scriptsize $(0.0652)$} \\
\rowcolor{gray!15}
ZOC-TN
& $1.69$s & $0.0846$& $0.430$& $0.171$ \\
& & {\scriptsize $(0.0190)$}& {\scriptsize $(0.196)$}& {\scriptsize $(0.0249)$} \\
\rowcolor{gray!15}
ZOC-SG
& $49.9$s& $0.0826$& $0.420$& $0.168$ \\
& & {\scriptsize $(0.0180)$}& {\scriptsize $(0.179)$}& {\scriptsize $(0.0230)$} \\
\rowcolor{gray!15}
ZOC-TB
& $132$s& $0.0832$ & $0.433$ & $0.169$ \\
& & {\scriptsize $(0.0193)$}& {\scriptsize $(0.178)$}& {\scriptsize $(0.0249)$} \\
\midrule
\textit{Independent Tree-Boosting} \\
\midrule
\rowcolor{gray!15}
Quasi-B & $75.4$s & $0.0657$ & $0.638$ & $0.226$ \\
& & {\scriptsize $(0.0123)$}& {\scriptsize $(0.0305)$}& {\scriptsize $(0.0135)$} \\
\rowcolor{gray!15}
ZOC-TN & $95.7$s & $0.0663$ & $-0.131$ & $0.146$ \\
& & {\scriptsize $(0.0127)$}& {\scriptsize $(0.130)$}& {\scriptsize $(0.0162)$} \\
\rowcolor{gray!15}
ZOC-SG & $70.7$s & $0.0669$ & $0.325$ & $0.149$ \\
& & {\scriptsize $(0.0106)$}& {\scriptsize $(0.149)$}& {\scriptsize $(0.0150)$} \\
\rowcolor{gray!15}
ZOC-TB & N/A & N/A & N/A & N/A \\
\midrule
\textit{GP Linear Models} \\
\midrule
\rowcolor{gray!15}
Spatial (\texttt{ZOCTN\_S\_Lin}) & $68.3$s & $0.0721$ & $-0.0816$ & $0.154$ \\
& & {\scriptsize $(0.0156)$}& {\scriptsize $(0.118)$}& {\scriptsize $(0.0168)$} \\
\rowcolor{gray!15}
Spatio-temporal (\texttt{ZOCTN\_ST\_Lin}) & $1,370$s & $0.0689$ & $-0.0979$ & $0.151$ \\
& & {\scriptsize $(0.0201)$}& {\scriptsize $(0.212)$}& {\scriptsize $(0.0279)$} \\
\midrule
\textit{GP Tree-Boosting} \\
\midrule
\rowcolor{gray!15}
Spatial (\texttt{ZOCTN\_S\_Bst}) & $764$s & $0.0722$ & $-0.0919$ & $0.153$ \\
& & {\scriptsize $(0.0199)$}& {\scriptsize $(0.177)$}& {\scriptsize $(0.0251)$} \\
\rowcolor{gray!15}
Spatio-temporal (\texttt{ZOCTN\_ST\_Bst}) & $7,510$s & $\mathbf{0.0572}$ & $\mathbf{-0.210}$ & $\mathbf{0.134}$ \\
& & {\scriptsize $(0.00593)$}& {\scriptsize $(0.0816)$}& {\scriptsize $(0.00732)$} \\
\bottomrule
\end{tabular}
\end{table}

Figure~\ref{fig:stage2_zoctn_performance} shows annual accuracy measures of the different ZOC-TN models across the 15 test years. Overall, most models exhibit two drops in prediction accuracy around the two recessions 2008/2009 and 2020/2021. These periods serve as robustness checks for the models in the face of changing market conditions. We find that the spatio-temporal tree-boosted ZOC-TN model (ZOCTN\_ST\_Bst) stands out not only for its overall best performance, but its consistently high prediction accuracy across all years. In contrast to all other models, this spatio-temporal tree-boosting model's performance does not deteriorate during the two recessions. 
\begin{figure}[ht!]
  \centering
  \includegraphics[width=0.6\linewidth]{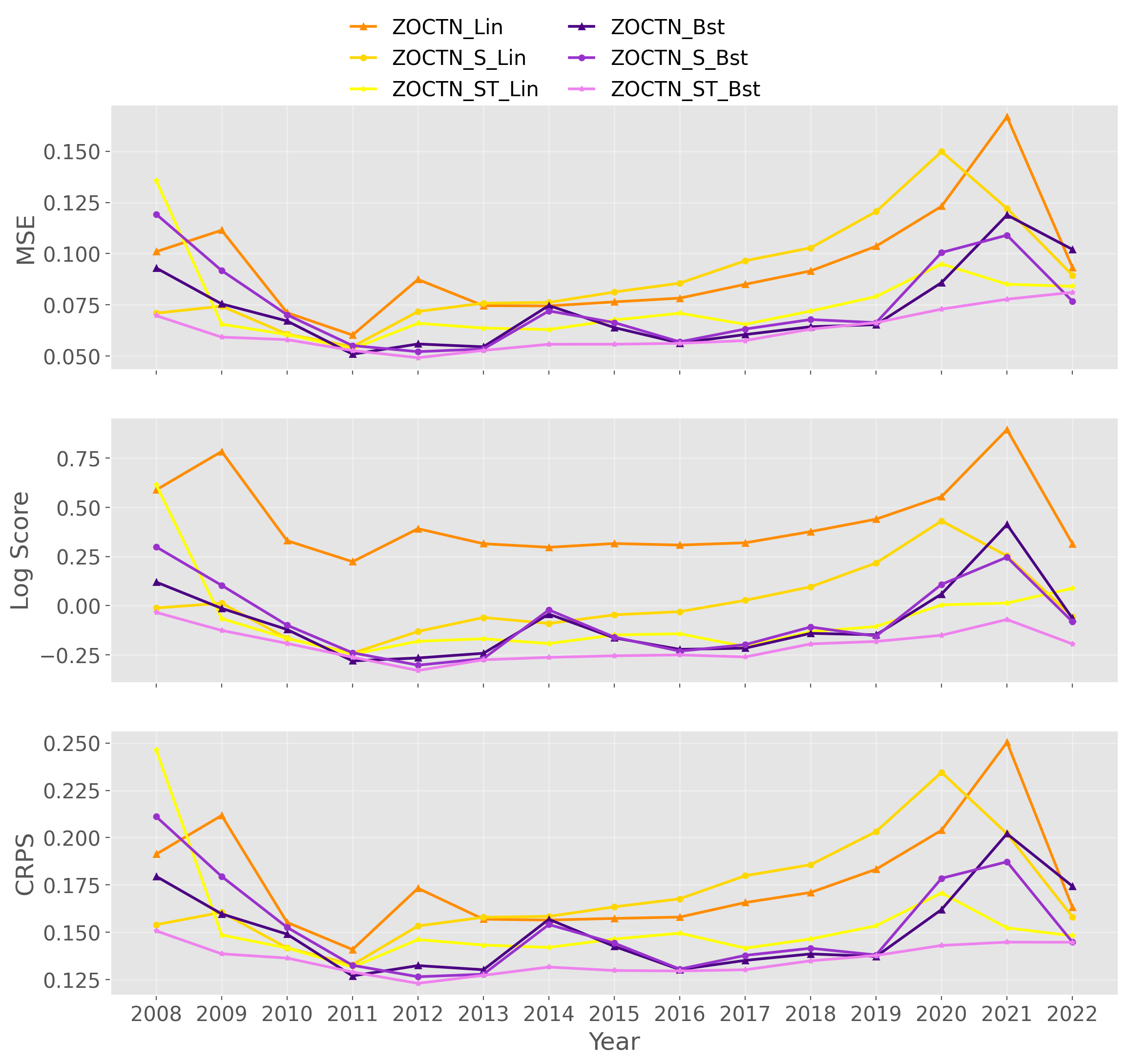}
  \caption{One-year-ahead prediction accuracy of ZOC-TN models vs. time.}
  \label{fig:stage2_zoctn_performance}
\end{figure}

In addition, Figure~\ref{fig:stage2_pitreliability} in Appendix~\ref{app:lgd_pit_qq} shows randomized PIT reliability residual diagrams to assess calibration with deviations from the horizontal zero line representing systematic over/under-dispersion or bias. Overall, we find that the linear and tree-boosted spatio-temporal ZOC-TN models (ZOCTN\_ST\_Lin and ZOCTN\_ST\_Bst) as well as the zero-one-inflated beta model (BEINF\_Lin) yield the best marginal calibration. Although the zero-one-inflated beta model shows low conditional prediction accuracy, it is nonetheless able to faithfully capture the marginal distribution of the data; similarly as observed in-sample with the suspended rootograms. Furthermore, we observe that predictions of the spatial linear and tree-boosted ZOC-TN models (ZOCTN\_S\_Lin and ZOCTN\_S\_Bst) are not well calibrated, which highlights the need for spatio-temporal models.

\subsection{Portfolio Loss and Capital Provisioning}
In what follows, we examine the economic implications of our models' prediction accuracy by assessing the models on the task of capital provisioning and $0.99$-quantile portfolio loss forecasting. Specifically, we use the models to forecast cumulative losses for a portfolio of mortgages resulting from all defaults occurring in the coming year, and we assess the prediction accuracy of these loss distributions by calculating the mean absolute error and the $0.99$-quantile loss \citep{Koenker1999}, $(L - q_{0.99})(0.99 - \mathbb{I}\{L\le q_{0.99}\})$, across all test years, 2008-2022. The $0.99$-quantile loss is of particular interest for risk management purposes as it evaluates the upper tail of predictive portfolio loss distributions. To generate predictive portfolio loss distribution, we sample the model's predictive LGD distribution for each loan defaulting in the coming year and multiply by the loan's UPB at the start of the year before summing these individual losses to find the forecast portfolio loss for the next year. This process is repeated $10,000$ times to generate an empirical loss distribution comprising $10,000$ observations. For models with GP components, we additionally first sample from the posterior predictive distribution of the GP.
%Under a ZOC-TN model, for example, the LGD of an upcoming default is modeled as a random variable,
% \[
%     \mathrm{LGD}_i \sim \text{ZOC-TN}(\mu_i, \sigma, a, b)
% \]
% for a location parameter, $\mu_i$, determined by the loan's features and the model type. Due to this, we estimate the loss distribution in a given year, $y$, empirically by simulating predictive loss $10,000$ times: $L^y_j = \sum_{i=1}^{N_y} \mathrm{LGD}_{ij}\times\mathrm{UPB}_i$, for $1 \le j \le 10,000$ and $N_y$ defaults occurring in year $y$. 
The results are shown in Table~\ref{tab:economic_metrics} in Appendix \ref{app:lgd_annual_performance}. Unlike the MAE, the $0.99$-quantile loss is simply summed across years, with a lower score representing better performance. We find that the spatio-temporal tree-boosted ZOC-TN model is superior to all other models in terms of the MAE and $0.99$-quantile loss. For instance, concerning the MAE results, we find that use of the spatio-temporal GP tree-boosted ZOC-TN model for forecasting next-year losses would have produced a reduction in error of around $\$277$M on average per year against the best independent linear model, $\$148$M against the spatio-temporal GP linear ZOC-TN model, and $\$39$M over the best independent tree-boosted model. %The total $0.99$-quantile loss appears broadly comparable across all models with a few notable exceptions being the independent linear BE-INF model with a loss over $5$x higher than the ZOC-TN model, and the spatial GP linear model, which although still outperformed by the spatio-temporal GP tree-boosted ZOC-TN model, produced a markedly lower loss than all other linear models, including the spatio-temporal GP linear model. %Although the independent tree-boosting models failed to offer a clear improvement in terms of $0.99$-quantile loss to the independent linear models, MAE for these models was roughly half that of their linear counterparts. 

\subsection{Model Interpretation}\label{subsec:freddiemac_interpretation}
In this final section, we  analyze potential drivers of LGD through the interpretation of the best-performing tree-boosting models. We begin by examining the posterior means for the random effects in the spatial and spatio-temporal GPBoost models presented in Figures~\ref{fig:sgpb_heatmap} and \ref{fig:stgpb_heatmap} in Appendix \ref{app:gp_interp}, respectively. The posterior distributions are calculated for models fit on the entire LGD dataset, with tree-boosting hyperparameters matching those of the 2022 models. In both cases, the posterior mean heatmaps closely mirror the plots of raw spatial and spatio-temporal average LGDs (see Section \ref{subsec:dataset_desc}). This means that a persistent portion of the LGD in the high-LGD regions identified in Section~\ref{subsec:dataset_desc} (Florida, the Rustbelt / Appalachia, New Mexico) cannot be explained by the included covariates alone. These high-residual clusters hint at geographical factors endemic to the housing market, or at least not correlated with the macroeconomic factors included in our analysis. One possible explanation for this is regional housing depressions caused by a net migration out of a region. Under this interpretation a persistent drop in demand for housing stock would drive down house prices and force banks to sell repossessed property for only a fraction of the original mortgage. This could feasibly be the case in the Rustbelt, as well as Appalachia where coal mining was an important part of the regional economy historically. This is also consistent with the broader economic literature, with \cite{harrison2023} finding that ``hypervacancy" in the Rustbelt, measured by metropolitan statistical areas with a long-term vacancy rate of more than $8\%$, recovered at a much slower rate between 2012-2019 than elsewhere in the country.
\begin{figure}[ht!]
  \centering
  \includegraphics[width=0.6\linewidth]{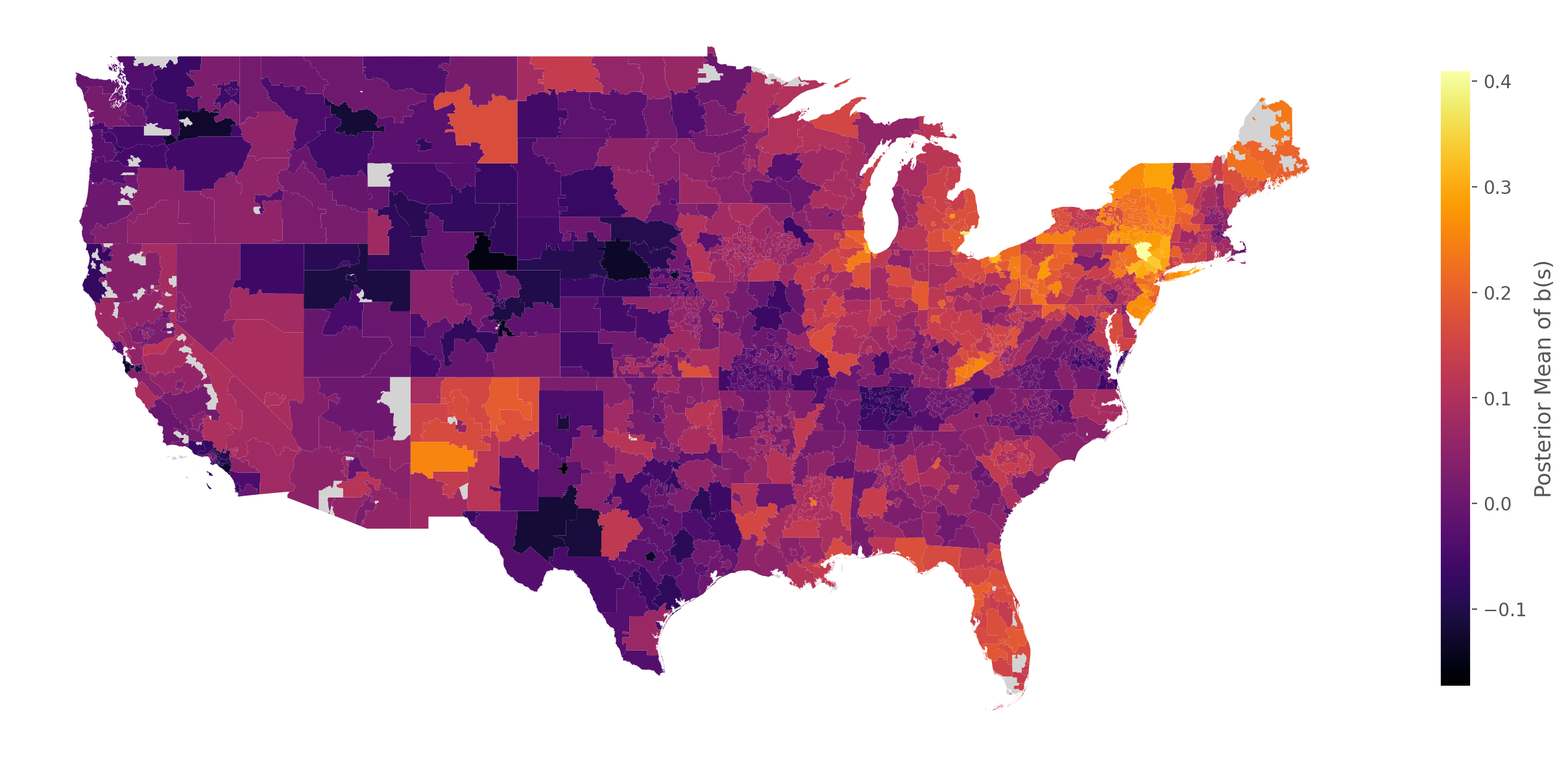}
  \caption{Posterior mean of latent Gaussian process in spatial tree-boosting model.}
  \label{fig:sgpb_heatmap}
\end{figure}

To understand potential nonlinear effects of covariates on the LGD, we analyze the spatio-temporal tree-boosted model trained on data up to the end of 2019 using SHAP (SHapley Additive exPlanations) values \citep{lundberg2017shap}. This model is chosen due to its high prediction accuracy and the clear difference in performance compared to the corresponding linear ZOC-TN model in the 2020 test fold (see Figure~\ref{fig:stage2_zoctn_performance}). This gap hints at the existence of nonlinear relationships in the data which our linear models are unable to capture. For computational convenience, as well as to prevent overcrowded figures, these SHAP plots are based on $10,000$ randomly selected training samples. According to the average absolute SHAP values shown in Figure~\ref{fig:shap_summary} in Appendix~\ref{app:gp_interp}, the five most important predictor variables are, in descending order, \texttt{original\_upb}, \texttt{insurance\_percent}, \texttt{ltv\_at\_default}, \texttt{loan\_purpose\_P}, and \texttt{ir\_spread}. Here, \texttt{original\_upb} denotes the original unpaid balance of the mortgage, and \texttt{loan\_purpose\_P} indicates whether the loan was a purchase transaction rather than a refinance transaction. Both \texttt{insurance\_percent} and \texttt{ltv\_at\_default} behave as one might expect, with greater insurance coverage on the loan resulting in lower proportional losses (negative SHAP value), while low loan-to-value ratios make it easier to recuperate the remaining unpaid balance.

Figure~\ref{fig:shap_dependence} reports SHAP dependence plots for the four most important numerical covariates. First, we find a clear inverse relationship between \texttt{original\_upb} and the LGD. This is in line with the results from the linear regression models where the coefficient of \texttt{original\_upb} is consistently negative across all models (see Appendix~\ref{app:lgd_ilms}). This can be explained as follows. If we assume, on average, that 1) under a prudent loan issuing process, the repayment progress prior to default, as a rate, should be independent of the loan size, and 2) through the liquidation process debt holders are able to recuperate the full outstanding amount minus some fixed administrative and servicing cost, we would expect to find a relationship as in Figure~\ref{fig:shap_dependence} due to the definition of the LGD. Besides, Figure~\ref{fig:shap_dependence} suggests distinctly nonlinear relationships between other covariates and the LGD. Of these, the dependence plots of \texttt{ir\_spread} and \texttt{insurance\_percent} appear roughly S-shaped, with the dispersion of the effect increasing on the tails, while the \texttt{ltv\_at\_default} shows a linear relationship between approximately $0$ and $100$, before plateauing beyond this point. The intuition behind this relationship is that when the home market value exceeds the outstanding loan value ($\texttt{ltv\_at\_default} < 100$), the underlying property can be sold immediately to recuperate the UPB, with the farther $\texttt{ltv\_at\_default}$ being away from $100$, the larger the cushion for incidentals and liquidation costs. Moreover, the concentration of highly insured loans around and above $\texttt{ltv\_at\_default} = 100$ provides some insight into the piecewise linear relationship, with issuers potentially adopting a different set of practices for the maintenance and liquidation of high-risk, high-LTV loans.
\begin{figure}[ht!]
  \centering
  \includegraphics[
    width=0.9\linewidth,
    trim={0 9.8cm 0 0},
    clip
  ]{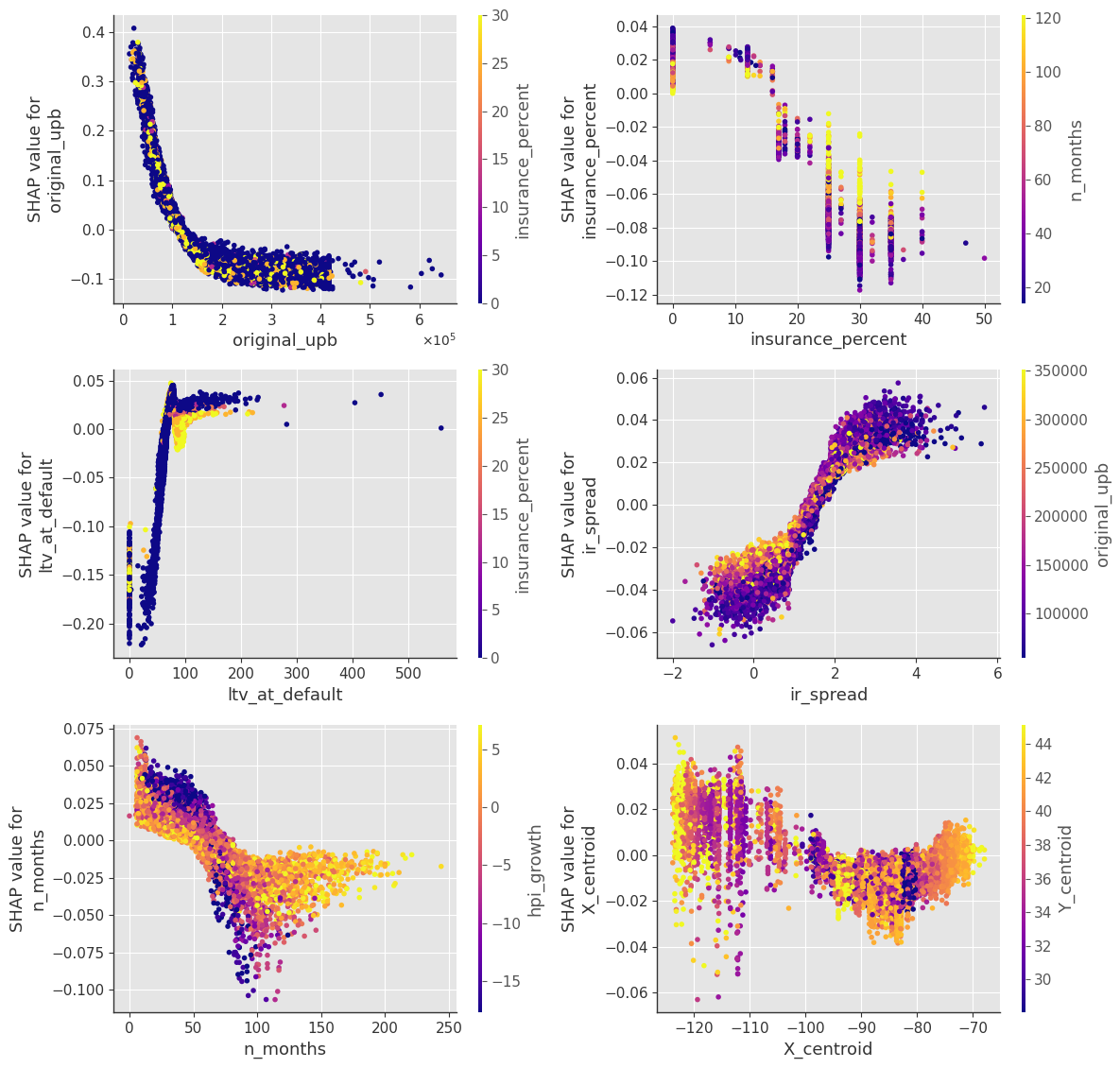}
  \caption{SHAP dependence plots for the four most important numerical covariates in the spatio-temporal tree-boosted model.}
  \label{fig:shap_dependence}
\end{figure}

\section{Conclusions}
This article introduces the zero-one censored transformed normal (ZOC-TN) distribution for bounded responses with possible mass at 0 and 1. The model combines a censored Gaussian variable with a two-parameter affine-logit transformation on the interior that can be seen as a low-dimensional approximation to a much broader class of possible smooth monotone logit-scale transformations. The simulation study and LGD application show that ZOC-TN is a flexible, computationally stable model for bounded responses with boundary mass. It performs competitively across diverse distributional settings, is substantially faster and more numerically reliable than the censored beta and shifted-gamma alternatives, and remains usable in richer tree-boosting and Gaussian process models. In the Freddie Mac application, the strongest results arise when the ZOC-TN likelihood is combined with nonlinear fixed effects and spatio-temporal random effects, suggesting that mortgage LGD is shaped by both complex covariate effects and residual dependence across space and time.

%These transformation parameters have a direct location-scale interpretation on the logit scale, and they reshape the interior distribution without altering the Gaussian-driven boundary masses. This decoupling gives the ZOC-TN likelihood more flexibility than an ordinary censored normal model, and allows it to flexibly represent various shapes. %We also establish large-sample properties for the independent linear model and derive an information decomposition showing that boundary observations identify the latent censored component, while interior observations identify the transformation parameters.

% In summary, the results here suggest that ZOC-TN is a practical addition to the fractional-response toolbox: it is more flexible than the standard censored normal, substantially cheaper than the strongest competing censored alternatives, and well suited for use as a likelihood within richer predictive models.

The paper leaves several directions for future work. On the methodological side, it would be natural to allow more than one distributional parameter to depend on the model features. For instance, introducing covariate dependence in the variance parameter would allow for a heteroscedastic ZOC-TN model. Another extension of the ZOC-TN model can be obtained by modifying the two-parameter affine-logit transformation to cover a broader set of monotone transformations. Concerning the latent variable which is transformed on $(0,1)$, other distributions such as extended-support beta and shifted gamma distributions can be used instead of a normal distribution. Alternatively, the censored construction could be abandoned altogether by estimating $p_0$ and $p_1$ separately and combining them with a transformed-normal over the interior to produce a two-tier hurdle model.

\bibliographystyle{abbrvnat}
\bibliography{bib_name.bib}
\clearpage

\appendix
%\chapter{Supplementary information}
% \label{app:supp}

\appendix

\if0\jbes{
\section*{Appendix}
}\fi

\setcounter{figure}{0}
\renewcommand{\thefigure}{A\arabic{figure}}
\setcounter{table}{0}
\renewcommand{\thetable}{A\arabic{table}}

\section{Proof of Proposition~\ref{prop:zoc-tn-interpretation}}\label{app:proof_of_zoctn_interpretation}
Here we prove the statements of Proposition~\ref{prop:zoc-tn-interpretation} on the interpretation of the interior location and concentration parameters, $(a, b)$, of the ZOC-TN model.
\begin{proof}
Since $\sigma^2>0$, the event $\{0<Z<1\}$ has positive probability, so the
conditional distribution of $Z$ given $0<Z<1$ is well defined. By construction,
$Z^\circ\in(0,1)$ almost surely. Hence $U=\logit(Z^\circ)$ is finite almost
surely. Also,
\[
Y^\circ
=
g_{a,b}(Z^\circ)
=
\expit\{a+b\logit(Z^\circ)\}
\in(0,1)
\]
almost surely. Therefore $V=\logit(Y^\circ)$ is finite almost surely, and
\[
V
=
\logit(Y^\circ)
=
\logit\left[\expit\{a+b\logit(Z^\circ)\}\right]
=
a+b\logit(Z^\circ)
=
a+bU.
\]
Thus
\[
V=a+bU
\qquad\text{almost surely}.
\]

% Part (i) follows from this affine identity. Whenever the relevant moments
% exist,
% \[
% \mathbb EV
% =
% \mathbb E(a+bU)
% =
% a+b\,\mathbb EU.
% \]
% For $m\ge 2$,
% \[
% V-\mathbb EV
% =
% b(U-\mathbb EU),
% \]
% and therefore
% \[
% \mathbb E\!\left[(V-\mathbb EV)^m\right]
% =
% b^m
% \mathbb E\!\left[(U-\mathbb EU)^m\right].
% \]
% The variance identity is the special case $m=2$.

Furthermore, let $a_2>a_1$ and fix $b>0$. Construct the two interior
responses from the same interior draw $Z^\circ$:
\[
Y^\circ_{a_j,b}
=
g_{a_j,b}(Z^\circ)
=
\expit\{a_j+b\logit(Z^\circ)\},
\qquad j=1,2.
\]
For every $z\in(0,1)$,
\[
g_{a_2,b}(z)
=
\expit\{a_2+b\logit(z)\}
>
\expit\{a_1+b\logit(z)\}
=
g_{a_1,b}(z),
\]
because $a_2>a_1$ and $\expit$ is strictly increasing. Hence
\[
Y^\circ_{a_2,b}
\ge
Y^\circ_{a_1,b}
\qquad\text{almost surely}.
\]
This coupling implies
\[
Y^\circ_{a_2,b}
\ge_{\mathrm{st}}
Y^\circ_{a_1,b}.
\]
Since $Y^\circ_{a,b}$ has the conditional distribution of
$Y_{a,b}$ given $0<Y_{a,b}<1$, the equivalent conditional stochastic
dominance statement follows.

% For part (ii), define
% \[
% h(u)=\expit(a+bu).
% \]
% Because $b>0$ and $\expit$ is strictly increasing, $h$ is strictly increasing.
% Since
% \[
% Y^\circ
% =
% h(U),
% \]
% strictly increasing transformations preserve unique medians. Therefore, if
% $m_U$ is the unique median of $U$, then the median of $Y^\circ$ is
% \[
% m_Y
% =
% h(m_U)
% =
% \expit(a+b\,m_U).
% \]
% Because $Y^\circ$ has the conditional distribution of $Y$ given $0<Y<1$, this
% is also the conditional median of $Y\mid(0<Y<1)$.
\end{proof}

\section{Proof of Proposition~\ref{prop:zoc-tn-local-density}}\label{app:local_density_proof}
The proof of local approximation of induced interior densities for the affine-logit transformation relies on the following result showing local approximation of the transformed variables on the logit scale.

\begin{proposition}[Local approximation on the logit scale]
\label{prop:zoc-tn-local-approx}
Let \(J\subset \mathbb{R}\) be an interval, and let \(h:J\to\mathbb{R}\) be twice continuously differentiable with \(h'(t)>0\) for all \(t\in\mathbb{R}\). Define
\[
T_h(z):=\mathrm{expit}\!\bigl(h(\mathrm{logit}\, z)\bigr),
\qquad z\in K,
\]
where \(K\subset(0,1)\) is compact and \(\mathrm{logit}(K)\subset J\). Fix \(t_0\in J\), and let
\[
b_0:=h'(t_0)>0,
\qquad
a_0:=h(t_0)-h'(t_0)t_0.
\]
Let \(H_K\) denote the convex hull of \(\mathrm{logit}(K)\cup\{t_0\}\), and assume \(H_K\subset J\).
Then the affine-logit transformation
\[
g_{a_0,b_0}(z):=\mathrm{expit}\!\bigl(a_0+b_0\mathrm{logit}\, z\bigr)
\]
satisfies
\[
\sup_{z\in K}\bigl|T_h(z)-g_{a_0,b_0}(z)\bigr|
\le
\frac{1}{8}
\sup_{t\in H_K} |h''(t)|
\cdot
\sup_{t\in\mathrm{logit}(K)} |t-t_0|^2.
\]
In particular, the ZOC-TN transformation provides a first-order local approximation to any smooth monotone transformation on the logit scale.
\end{proposition}

\begin{proof}
Let \(t=\mathrm{logit}\, z\). By Taylor's theorem, for every \(t\in \mathrm{logit}(K)\),
\[
h(t)=h(t_0)+h'(t_0)(t-t_0)+R(t),
\]
where the remainder satisfies
\[
|R(t)|
\le
\frac12
\sup_{u\in H_K} |h''(u)|\,|t-t_0|^2.
\]
By the definition of \(a_0\) and \(b_0\),
\[
h(t_0)+h'(t_0)(t-t_0)=a_0+b_0 t,
\]
and therefore
\[
|h(t)-(a_0+b_0 t)|
\le
\frac12
\sup_{u\in H_K} |h''(u)|\,|t-t_0|^2.
\]

Next, the $\mathrm{expit}$ function is globally Lipschitz with constant \(1/4\), since
\[
\mathrm{expit}'(u)=\mathrm{expit}(u)\{1-\mathrm{expit}(u)\}\le \frac14
\qquad\text{for all }u\in\mathbb{R}.
\]
Hence
\[
|T_h(z)-g_{a_0,b_0}(z)|
=
\left|\mathrm{expit}(h(t))-\mathrm{expit}(a_0+b_0 t)\right|
\le
\frac14\,|h(t)-(a_0+b_0 t)|.
\]
Combining the two bounds gives
\[
|T_h(z)-g_{a_0,b_0}(z)|
\le
\frac18
\sup_{u\in H_K} |h''(u)|\,|t-t_0|^2.
\]
Taking the supremum over \(z\in K\), equivalently \(t\in\mathrm{logit}(K)\), proves the claim.
\end{proof}

Proposition~\ref{prop:zoc-tn-local-approx} shows that, although the ZOC-TN family is not globally universal, its affine-logit transformation is a natural local approximation to a broad class of smooth monotone transformations on the unit interval. We now proceed with the proof of Proposition~\ref{prop:zoc-tn-local-density}.

\begin{proof}
Write
\[
T_h(z)=\mathrm{expit}(h(\mathrm{logit}\, z)),
\qquad
T_{a_0,b_0}(z)=\mathrm{expit}(a_0+b_0\mathrm{logit}\, z).
\]
By Proposition~\ref{prop:zoc-tn-local-approx},
\[
\sup_{z\in K_0}|T_h(z)-T_{a_0,b_0}(z)|
\le
C_1 \sup_{t\in \mathrm{logit}(K_0)}|t-t_0|^2
\]
for every compact \(K_0\subset(0,1)\) contained in the domain of interest, where \(C_1\) is finite because \(h''\) is continuous on the corresponding compact convex hull.

Since both transformations are strictly increasing, the corresponding densities satisfy the change-of-variables formulas
\[
f_h(y)=f_Z\!\bigl(T_h^{-1}(y)\bigr)\,\bigl|(T_h^{-1})'(y)\bigr|,
\qquad
f_{a_0,b_0}(y)=f_Z\!\bigl(T_{a_0,b_0}^{-1}(y)\bigr)\,\bigl|(T_{a_0,b_0}^{-1})'(y)\bigr|.
\]
Because \(T_h\) and \(T_{a_0,b_0}\) are \(C^1\) increasing diffeomorphisms on neighborhoods of the relevant compact sets, their inverses are \(C^1\) there as well. Moreover, on the compact set
\[
K_1:=T_h^{-1}(K)\cup T_{a_0,b_0}^{-1}(K),
\]
the derivatives of both transformations are bounded away from zero. The local approximation of \(T_h\) by \(T_{a_0,b_0}\), together with the inverse function theorem and compactness, implies
\[
\sup_{y\in K}\bigl|T_h^{-1}(y)-T_{a_0,b_0}^{-1}(y)\bigr|
\le C_2 \delta_K^2,
\qquad
\sup_{y\in K}\bigl|(T_h^{-1})'(y)-(T_{a_0,b_0}^{-1})'(y)\bigr|
\le C_3 \delta_K,
\]
where
\[
\delta_K:=\sup_{t\in \mathrm{logit}(T_h^{-1}(K)\cup T_{a_0,b_0}^{-1}(K))}|t-t_0|.
\]
The second bound is first order because it depends on the derivative approximation
\[
h'(t)-b_0=O(|t-t_0|)
\]
rather than only on the second-order approximation of \(h\) itself.

Using the density formulas and adding/subtracting a cross term,
\begin{align*}
|f_h(y)-f_{a_0,b_0}(y)|
&\le
\Big|f_Z\!\bigl(T_h^{-1}(y)\bigr)-f_Z\!\bigl(T_{a_0,b_0}^{-1}(y)\bigr)\Big|
\,\bigl|(T_h^{-1})'(y)\bigr| \\
&\quad
+
\bigl|f_Z\!\bigl(T_{a_0,b_0}^{-1}(y)\bigr)\bigr|
\Big|(T_h^{-1})'(y)-(T_{a_0,b_0}^{-1})'(y)\Big|.
\end{align*}
Since \(f_Z\) is bounded and Lipschitz on the relevant neighborhood, and the inverse derivatives are uniformly bounded on \(K\), the right-hand side is bounded by \(C_K\delta_K\), uniformly in \(y\in K\). Taking the supremum over \(K\) proves the claim.
\end{proof}

% --------------- PROOF OF THEOREM 2.2 ---------------
\section{Proof of Theorem~\ref{thm:zoctn_mle_qmle}}\label{proof:thm_zoctn_mle_qmle}
% TODO: CHANGE NOTATION OF LIKELIHOOD FUNCTIONS

The proof of Theorem~\ref{thm:zoctn_mle_qmle} is standard once the model-specific regularity conditions are verified. In particular, the only points requiring some care are the mixed discrete-continuous nature of the likelihood, the continuity and differentiability of the interior density after the monotone transformation, and the identification of the population maximizer. These technical details are fleshed out now.

Since the asymptotic argument is standard, it suffices to check that the ZOC-TN likelihood defines a well-behaved mixed discrete-continuous parametric family, that the criterion admits an integrable envelope, and that the population objective has a unique maximizer under correct specification and identifiability.

Throughout, let
\[
\boldsymbol{\theta}=(\boldsymbol{\beta},\sigma,a,b)\in\Theta\subset \mathbb{R}^p\times(0,\infty)\times\mathbb{R}\times(0,\infty),
\]
and write \(\mu_i=\mathbf{x}_i^\top\boldsymbol{\beta}\). For \(y\in(0,1)\), define
% \[
% z_{\boldsymbol{\theta}}(y):=\mathrm{expit}\!\left(\frac{\mathrm{logit}(y)-a}{b}\right).
% \]
\[
z_{a, b}(y):=g_{a, b}^{-1}(y)=\mathrm{expit}\!\left(\frac{\mathrm{logit}(y)-a}{b}\right).
\]
The ZOC-TN likelihood contribution can be written as
\[
p_{\boldsymbol{\theta}}(y\mid \mathbf{x})=
\begin{cases}
\Phi\!\left(-\dfrac{\mathbf{x}^\top\boldsymbol{\beta}}{\sigma}\right), & y=0,\\[1em]
\phi\!\left(\dfrac{z_{a,b}(y)-\mathbf{x}^\top\boldsymbol{\beta}}{\sigma}\right)
\dfrac{z_{a,b}(y)\{1-z_{a,b}(y)\}}{\sigma b\,y(1-y)}, & y\in(0,1),\\[1em]
1-\Phi\!\left(\dfrac{1-\mathbf{x}^\top\boldsymbol{\beta}}{\sigma}\right), & y=1.
\end{cases}
\]
We view this as a density with respect to the dominating measure
\[
\nu:=\delta_0+\lambda_{(0,1)}+\delta_1,
\]
where \(\lambda_{(0,1)}\) denotes Lebesgue measure on \((0,1)\). The corresponding log-likelihood contribution is \(\ell^{\text{ZOC-TN}}(\boldsymbol{\theta} \mid y,\mathbf{x})=\log p_{\boldsymbol{\theta}}(y\mid \mathbf{x})\).

For the verification below, we impose the following mild high-level assumptions.

\begin{assumption}
\label{ass:zoc-tn-regularity}
The following conditions hold.
\begin{enumerate}
    \item[(A1)] The parameter space \(\Theta\) is compact, and there exist constants \(\underline{\sigma}>0\) and \(\underline{b}>0\) such that \(\sigma\ge \underline{\sigma}\) and \(b\ge \underline{b}\) for all \(\boldsymbol{\theta}\in\Theta\).
    \item[(A2)] The covariate vector \(\mathbf{X}\) has finite first and second moments, and there exists a square-integrable random variable \(C(\mathbf{X})\) such that
    \[
    \sup_{\boldsymbol{\beta}\in\mathrm{proj}_{\boldsymbol{\beta}}(\Theta)} |\mathbf{X}^\top\boldsymbol{\beta}| \le C(\mathbf{X})
    \qquad\text{a.s.}
    \]
    In addition, the interior log-boundary terms satisfy
    \[
    \mathbb{E}\!\left[
    \mathbf{1}_{\{0<Y<1\}}
    \left\{
    |\log Y|+|\log(1-Y)|+|\logit(Y)|^2
    \right\}
    \right]<\infty.
    \]
    In particular, the covariate bound holds if \(\mathbf{X}\) is almost surely bounded.
    \item[(A3)] The ZOC-TN family is identifiable on \(\Theta\): if
    \(p_{\boldsymbol{\theta}}(\cdot\mid\mathbf{X})=
    p_{\boldsymbol{\theta}'}(\cdot\mid\mathbf{X})\) almost surely, then
    \(\boldsymbol{\theta}=\boldsymbol{\theta}'\). A sufficient condition is that
    the support of \(\mathbf{X}\) is not contained in any proper affine hyperplane,
    together with the usual identifiability of the common scale and transformation
    parameters from the boundary masses and the interior density.
    \item[(A4)] The ZOC-TN model is correctly specified whenever this is invoked, i.e.\ there exists \(\theta_0\in\Theta^\circ\) such that the conditional law of \(Y\mid \mathbf{X}\) is \(p_{\boldsymbol{\theta}_0}(\cdot\mid \mathbf{X})\).
    \item[(A5)] For the misspecified case, the population criterion
    \[
    M(\boldsymbol{\theta}):=\mathbb{E}[\ell^{\text{ZOC-TN}}(\boldsymbol{\theta} \mid Y,\mathbf{X})]
    \]
    has a unique maximizer \(\boldsymbol{\theta}^\star\in\Theta^\circ\).
\end{enumerate}
\end{assumption}

Assumption~\ref{ass:zoc-tn-regularity}(A1) guarantees that the scale parameters remain bounded away from zero, thereby excluding degeneracies in both the latent Gaussian density and the transformation Jacobian. Assumption~\ref{ass:zoc-tn-regularity}(A2) is a standard moment requirement ensuring integrability of the likelihood, score, and Hessian envelopes, including the logarithmic terms that arise near the endpoints of the interior interval. Assumption~\ref{ass:zoc-tn-regularity}(A3) records the required identifiability condition for the full regression model. Assumption~\ref{ass:zoc-tn-regularity}(A5) is the standard uniqueness condition for quasi-maximum likelihood under misspecification.

We first verify that \(p_{\boldsymbol{\theta}}(\cdot\mid \mathbf{x})\) indeed defines a proper mixed distribution.

\begin{lemma}
\label{lem:zoc-tn-proper}
For every \(\mathbf{x}\) and every \(\boldsymbol{\theta}\in\Theta\), \(p_{\boldsymbol{\theta}}(\cdot\mid \mathbf{x})\) is a probability density/mass function with respect to \(\nu\).
\end{lemma}

% TODO: NOTATION
\begin{proof}
The boundary terms are nonnegative by construction. For \(y\in(0,1)\), the transformation \(g_{a,b}(z)=\mathrm{expit}(a+b\mathrm{logit} z)\) is strictly increasing whenever \(b>0\), and its inverse is
\[
g_{a,b}^{-1}(y)=z_{a, b}(y)=\mathrm{expit}\!\left(\frac{\mathrm{logit}(y)-a}{b}\right).
\]
A direct calculation yields
\[
\frac{d}{dy}z_{a, b}(y)
=
\frac{z_{a, b}(y)\{1-z_{a, b}(y)\}}{b\,y(1-y)}.
\]
Hence the interior density is the change-of-variables transform of the uncensored Gaussian density on \((0,1)\). Therefore,
\[
\int_0^1 p_{\boldsymbol{\theta}}(y\mid \mathbf{x})\,dy
=
\int_0^1 \frac{1}{\sigma}\phi\!\left(\frac{z-\mathbf{x}^\top\boldsymbol{\beta}}{\sigma}\right)\,dz
=
\Phi\!\left(\frac{1-\mathbf{x}^\top\boldsymbol{\beta}}{\sigma}\right)-\Phi\!\left(-\frac{\mathbf{x}^\top\boldsymbol{\beta}}{\sigma}\right).
\]
Adding the masses at \(0\) and \(1\) gives
\[
\Phi\!\left(-\frac{\mathbf{x}^\top\boldsymbol{\beta}}{\sigma}\right)
+\Phi\!\left(\frac{1-\mathbf{x}^\top\boldsymbol{\beta}}{\sigma}\right)-\Phi\!\left(-\frac{\mathbf{x}^\top\boldsymbol{\beta}}{\sigma}\right)
+1-\Phi\!\left(\frac{1-\mathbf{x}^\top\boldsymbol{\beta}}{\sigma}\right)
=1.
\]
Thus \(p_{\boldsymbol{\theta}}(\cdot\mid \mathbf{x})\) is a proper probability law. 
\end{proof}

The next lemma verifies continuity of the criterion and existence of a sample maximizer.

\begin{lemma}
\label{lem:zoc-tn-continuity}
Under Assumption~\ref{ass:zoc-tn-regularity}(A1), for every \((y,\mathbf{x})\), the map \(\boldsymbol{\theta}\mapsto \ell^{\text{ZOC-TN}}(\boldsymbol{\theta} \mid y,\mathbf{x})\) is continuous on \(\Theta\). Consequently, for every sample, the maximizer
\[
\hat{\boldsymbol{\theta}}_N\in\arg\max_{\boldsymbol{\theta}\in\Theta}\sum_{i=1}^N \ell^{\text{ZOC-TN}}(\boldsymbol{\theta} \mid Y_i,\mathbf{X}_i) %\ell_{\boldsymbol{\theta}}(Y_i,\mathbf{X}_i)
\]
exists.
\end{lemma}

\begin{proof}
For \(y=0\) and \(y=1\), continuity follows from continuity of the normal c.d.f. and the fact that \(\sigma\) is bounded away from zero. For \(y\in(0,1)\), the map
\[
\boldsymbol{\theta}\mapsto z_{a, b}(y)=\expit\!\left(\frac{\logit(y)-a}{b}\right)
\]
is continuous because \(b\ge \underline b>0\), and all other factors in the interior density are continuous in \(\theta\). Since \(y\in(0,1)\) is fixed, the factor \(y(1-y)\) is strictly positive, so the logarithm is well-defined and continuous. Hence \(\ell^{\text{ZOC-TN}}(\boldsymbol{\theta} \mid y,\mathbf{x})\) is continuous on \(\Theta\).

Because \(\Theta\) is compact, the sample log-likelihood is continuous on a compact set and therefore attains its maximum by the Weierstrass theorem.
\end{proof}

We next verify an integrable envelope for the log-likelihood.

\begin{lemma}
\label{lem:zoc-tn-envelope}
Under Assumption~\ref{ass:zoc-tn-regularity}(A1)--(A2),
\[
\mathbb{E}\!\left[\sup_{\boldsymbol{\theta}\in\Theta} |\ell^{\text{ZOC-TN}}(\boldsymbol{\theta} \mid Y,\mathbf{X})|\right] < \infty.
\]
\end{lemma}

\begin{proof}
We consider the three cases separately.

If \(Y=0\), then
\[
\ell^{\text{ZOC-TN}}(\boldsymbol{\theta} \mid 0 ,\mathbf{X}) =\log \Phi\!\left(-\frac{\mathbf{X}^\top\boldsymbol{\beta}}{\sigma}\right).
\]
Since \(\sigma\ge \underline\sigma\) and \(|\mathbf{X}^\top\boldsymbol{\beta}|\le C(\mathbf{X})\) uniformly over \(\theta\), the argument of \(\Phi\) remains in an interval of the form \([-C(\mathbf{X})/\underline\sigma,\, C(\mathbf{X})/\underline\sigma]\). Standard tail bounds for the Gaussian c.d.f. imply that \(|\log\Phi(t)|\le c_1+c_2|t|^2\) for all \(t\in\mathbb{R}\). Hence
\[
\sup_{\boldsymbol{\theta}\in\Theta} |\ell^{\text{ZOC-TN}}(\boldsymbol{\theta} \mid 0, \mathbf{X})|
\le c_1+c_2 C(\mathbf{X})^2.
\]
The same reasoning applies to \(Y=1\), since
\[
\ell^{\text{ZOC-TN}}(\boldsymbol{\theta} \mid 1, \mathbf{X})=\log\!\left[1-\Phi\!\left(\frac{1-\mathbf{X}^\top\boldsymbol{\beta}}{\sigma}\right)\right].
\]

For \(Y\in(0,1)\),
\[
\ell^{\text{ZOC-TN}}(\boldsymbol{\theta} \mid Y, \mathbf{X})
=
-\log \sigma
+\log \phi\!\left(\frac{z_{a, b}(Y)-\mathbf{X}^\top\boldsymbol{\beta}}{\sigma}\right)
+\log z_{a, b}(Y)
+\log\{1-z_{a, b}(Y)\}
-\log b
-\log Y
-\log(1-Y).
\]
By compactness of \(\Theta\), both \(|\log \sigma|\) and \(|\log b|\) are uniformly bounded. Moreover, since \(a\) and \(b\) range over compact sets and \(b\) is bounded away from zero, there exists a constant \(C<\infty\) such that, uniformly in \(\theta\in\Theta\),
\[
-\log z_{a,b}(Y)-\log\{1-z_{a,b}(Y)\}
\le C\{1+|\logit(Y)|\}.
\]
Furthermore, using
\[
\log \phi(u)=-\frac{1}{2}\log(2\pi)-\frac{1}{2}u^2,
\]
and the fact that \(0<z_{a,b}(Y)<1\) and \(\sigma\geq \underline{\sigma}>0\), we obtain
\[
\sup_{\theta\in\Theta}
\left|
\log \phi\left(\frac{z_{a,b}(Y)-X^\top\beta}{\sigma}\right)
\right|
\le c_3+c_4\{1+C(X)^2\}.
\]
Finally, the terms \(-\log Y\) and \(-\log(1-Y)\) do not depend on \(\theta\). Hence, on \(\{Y\in(0,1)\}\),
\[
\sup_{\theta\in\Theta}|\ell^{\mathrm{ZOC\text{-}TN}}(\theta\mid Y,X)|
\le
c_5+c_6 C(X)^2+c_7|\logit(Y)|-\log Y-\log(1-Y).
\]
The right-hand side is integrable by Assumption~\ref{ass:zoc-tn-regularity}(A2), which controls both \(C(\mathbf{X})^2\) and the log-boundary terms on the interior. This yields the claim.
\end{proof}

\noindent
The preceding argument is standard: the only potentially delicate terms are the logarithmic boundary terms \(-\log Y\) and \(-\log(1-Y)\), but these are data-dependent constants on the interior and therefore do not interfere with uniformity in \(\boldsymbol{\theta}\).

We next verify differentiability of the interior likelihood with respect to the parameters.

\begin{lemma}
\label{lem:zoc-tn-differentiability}
Under Assumption~\ref{ass:zoc-tn-regularity}(A1)--(A2), the map \(\boldsymbol{\theta}\mapsto \ell^{\text{ZOC-TN}}(\boldsymbol{\theta} \mid y,\mathbf{x})\) is twice continuously differentiable for every \(y\in\{0,1\}\cup(0,1)\). Moreover, in a neighborhood of any interior point \(\theta^\circ\in\Theta^\circ\), the score and Hessian admit integrable envelopes.
\end{lemma}

\begin{proof}
For the boundary cases \(y\in\{0,1\}\), the likelihood contributions are compositions of smooth functions of \((\mathbf{x}^\top\boldsymbol{\beta},\sigma)\) with \(\sigma>0\), and hence are twice continuously differentiable.

For \(y\in(0,1)\), the map \(\boldsymbol{\theta}\mapsto z_{a, b}(y)\) is \(C^\infty\) because \(b>0\) and the expit function is smooth. In particular,
\[
\frac{\partial z_{a, b}(y)}{\partial a}
=
-\frac{1}{b} z_{a, b}(y)\{1-z_{a,b}(y)\},
\qquad
\frac{\partial z_{a, b}(y)}{\partial b}
=
-\frac{\logit(y)-a}{b^2} z_{a,b}(y)\{1-z_{a, b}(y)\}.
\]
Since \(0<z_{a, b}(y)<1\), these derivatives are finite, and the same holds for all second derivatives. The interior log-likelihood is a finite sum of compositions and products of smooth functions:
\[
-\log \sigma,\quad
\log \phi\!\left(\frac{z_{a, b}(y)-\mathbf{x}^\top\boldsymbol{\beta}}{\sigma}\right),\quad
\log z_{a, b}(y),\quad
\log\{1-z_{a, b}(y)\},\quad
-\log b.
\]
Therefore \(\ell^{\text{ZOC-TN}}(\boldsymbol{\theta} \mid y,\mathbf{x})\) is twice continuously differentiable in \(\boldsymbol{\theta}\).

To obtain envelopes, note that derivatives with respect to \(\boldsymbol{\beta}\) and \(\sigma\) are polynomial in \(\mathbf{x}\), \(\sigma^{-1}\), and \(z_{a, b}(y)-\mathbf{x}^\top\boldsymbol{\beta}\); derivatives with respect to \(a\) and \(b\) are polynomial in \(b^{-1}\), \(\logit(y)-a\), and \(z_{a, b}(y)\{1-z_{a, b}(y)\}\). Because \(\Theta\) is compact with \(\sigma\) and \(b\) bounded away from zero, and because \(0<z_{a, b}(y)\{1-z_{a, b}(y)\}\le 1/4\), all such derivatives are bounded by expressions of the form
\[
K_1(\mathbf{X})+K_2(\mathbf{X})\,|\logit(Y)|+K_3(\mathbf{X})\,\logit(Y)^2
\qquad\text{on }\{Y\in(0,1)\},
\]
for integrable random variables \(K_j(\mathbf{X})\) under Assumption~\ref{ass:zoc-tn-regularity}(A2). On the boundary, the derivatives are bounded by \(1+C(\mathbf{X})^m\) for suitable integers \(m\). Hence the score and Hessian admit integrable envelopes in a neighborhood of any interior parameter point.
\end{proof}

% \begin{remark}
If desired, one may strengthen Assumption~\ref{ass:zoc-tn-regularity}(A2) to bounded support of \(X\), which makes the envelope arguments immediate. This is not essential, but it leads to shorter technical proofs.
% \end{remark}

We now address the population objective.

\begin{lemma}
\label{lem:zoc-tn-population-maximizer}
Suppose Assumption~\ref{ass:zoc-tn-regularity}(A1)--(A4) holds. Then the population criterion
\[
M(\boldsymbol{\theta})=\mathbb{E}_{\boldsymbol{\theta}_0}[\ell_{\boldsymbol{\theta}}(Y,\mathbf{X})]
\]
is uniquely maximized at \(\boldsymbol{\theta}_0\).
\end{lemma}

\begin{proof}
Under correct specification,
\[
M(\boldsymbol{\theta})-M(\boldsymbol{\theta}_0)
=
\mathbb{E}_{\theta_0}\!\left[\log \frac{p_\theta(Y\mid \mathbf{X})}{p_{\theta_0}(Y\mid \mathbf{X})}\right]
=
-\mathbb{E}_{\theta_0}\!\left[\log \frac{p_{\theta_0}(Y\mid \mathbf{X})}{p_\theta(Y\mid \mathbf{X})}\right].
\]
The right-hand side equals minus the conditional Kullback--Leibler divergence from \(p_{\theta_0}(\cdot\mid \mathbf{X})\) to \(p_\theta(\cdot\mid X)\), and is therefore nonpositive. Equality holds if and only if
\[
p_\theta(\cdot\mid \mathbf{X})=p_{\theta_0}(\cdot\mid \mathbf{X})
\qquad\text{almost surely}.
\]
By identifiability of the ZOC-TN model under Assumption~\ref{ass:zoc-tn-regularity}(A3), this implies \(\boldsymbol{\theta}=\boldsymbol{\theta}_0\). Hence \(M(\boldsymbol{\theta})\) is uniquely maximized at \(\boldsymbol{\theta}_0\).
\end{proof}

The misspecified case is entirely standard: one simply replaces \(\boldsymbol{\theta}_0\) by the pseudo-true parameter \(\boldsymbol{\theta}^\star\), i.e.\ the unique maximizer of \(M(\boldsymbol{\theta})\), as imposed in Assumption~\ref{ass:zoc-tn-regularity}(A5). No model-specific modification is needed beyond the continuity and domination arguments above.

We summarize the discussion as follows.

\begin{proposition}
\label{prop:zoc-tn-regularity-summary}
Under Assumption~\ref{ass:zoc-tn-regularity}, the ZOC-TN likelihood satisfies the standard regularity conditions required for consistency and asymptotic normality of the maximum likelihood estimator under correct specification, and of the quasi-maximum likelihood estimator under misspecification. In particular:
\begin{enumerate}
    \item \(p_\theta(\cdot\mid \mathbf{x})\) is a well-defined mixed discrete-continuous probability law;
    \item \(\theta\mapsto \ell^{\text{ZOC-TN}}(\boldsymbol{\theta} \mid y,\mathbf{x})\) is continuous for every \((y,\mathbf{x})\), and the sample criterion attains its maximum on \(\Theta\);
    \item \(\sup_{\theta\in\Theta} |\ell^{\text{ZOC-TN}}(\boldsymbol{\theta} \mid y,\mathbf{x})|\) is integrable;
    \item in a neighborhood of the true or pseudo-true parameter, the score and Hessian are well-defined and admit integrable envelopes;
    \item under correct specification and identifiability, the population criterion has a unique maximizer at the true parameter.
\end{enumerate}
Accordingly, Theorem~\ref{thm:zoctn_mle_qmle} follows from standard \(M\)-estimation arguments.
\end{proposition}

\begin{proof}[Proof sketch]
Existence of \(\hat{\boldsymbol{\theta}}_N\) follows from continuity of \(\mathcal{L}_N^{\text{ZOC-TN}}(\boldsymbol{\theta})\) and compactness of \(\Theta\). Consistency is then obtained from a uniform law of large numbers for \(N^{-1}\mathcal{L}_N^{\text{ZOC-TN}}(\boldsymbol{\theta})\) together with uniqueness of the maximizer of the population criterion \(M(\theta)\). For asymptotic normality, one applies a Taylor expansion of the score around \(\boldsymbol{\theta}^\star\), uses consistency of \(\hat{\boldsymbol{\theta}}_N\), a law of large numbers for the Hessian, and a central limit theorem for the score. Under correct specification, the information identity yields the usual inverse Fisher information limit; under misspecification, one obtains the standard sandwich covariance matrix \(A^{-1}BA^{-1}\).
\end{proof}

% --------------------- PROOF OF PROPOSITION 2.3 ----------------
\section{Fisher Information of the ZOC-TN Linear Regression Model}\label{proof:zoc-tn-information}

\begin{proposition}[Fisher information for the ZOC-TN regression model]
\label{prop:zoc-tn-information}
Fix \(\mathbf{X}=\mathbf{x}\). Under correct specification, the conditional Fisher information may be decomposed as
\begin{align}
    \mathcal I(\boldsymbol{\theta}\mid \mathbf{x})
    & = -\mathbb{E}_{\boldsymbol{\theta}}\!\left[\nabla_{\boldsymbol{\theta}}^2 \ell_{\boldsymbol{\theta}}(Y,\mathbf{X})\mid \mathbf{X}=\mathbf{x}\right] \label{eq:zoctn_condfim}\\
    & =
\mathcal I_0(\boldsymbol{\theta}\mid \mathbf{x})
+
\mathcal I_1(\boldsymbol{\theta}\mid \mathbf{x})
+
\mathcal I_{\mathrm{int}}(\boldsymbol{\theta}\mid \mathbf{x}), \label{eq:zoctn_fim_decomp}
\end{align}
for boundary contributions $\mathcal I_0$ and $\mathcal I_1$, with $\mathcal I_{\mathrm{int}}$ corresponding to the interior $(0, 1)$. Moreover:

\begin{enumerate}
\item[(i)] $\mathcal I_0$ and $\mathcal I_1$ in \eqref{eq:zoctn_fim_decomp} have closed-form expressions with support only in the $(\boldsymbol{\beta},\sigma)$-block. Specifically, for $$\mu=\mathbf{x}^\top\boldsymbol{\beta},~~\alpha=-\frac{\mu}{\sigma}, ~~\gamma=\frac{1-\mu}{\sigma}:$$
\[
\mathcal{I}_0(\boldsymbol{\theta}\mid \mathbf{x})
=
\frac{\phi(\alpha)^2}{q_0\,\sigma^2}
\begin{pmatrix}
\mathbf{x}\mathbf{x}^\top & \alpha \mathbf{x} & 0 & 0\\
\alpha \mathbf{x}^\top & \alpha^2 & 0 & 0\\
0&0&0&0\\
0&0&0&0
\end{pmatrix},~~ \mathcal{I}_1(\boldsymbol{\theta}\mid \mathbf{x})
=
\frac{\phi(\gamma)^2}{q_1\,\sigma^2}
\begin{pmatrix}
\mathbf{x}\mathbf{x}^\top & \gamma \mathbf{x} & 0 & 0\\
\gamma \mathbf{x}^\top & \gamma^2 & 0 & 0\\
0&0&0&0\\
0&0&0&0
\end{pmatrix},
\]
where $$q_0=\Phi(\alpha), ~~ q_1=1-\Phi(\gamma).$$ This means that boundary observations contribute information only on $(\beta, \sigma)$, whereas information on the transformation parameters $(a, b)$ is identified exclusively through the interior observations.

\item[(ii)] The full conditional information matrix is computationally tractable: In addition to the explicit boundary contributions given above, evaluation of the interior contribution reduces to one-dimensional Gaussian-weighted integrals. In particular, the terms involving only \((\boldsymbol{\beta},\sigma,a)\) reduce to truncated normal moments, while the terms involving \(b\) require only analogous one-dimensional weighted integrals. That is, all entries of $\mathcal{I}_{\mathrm{int}}(\boldsymbol{\theta}\mid \mathbf{x})$ involving only \((\boldsymbol{\beta},\sigma,a)\) reduce to truncated normal moments
\[
M_k(\alpha,\gamma)=\int_\alpha^\gamma u^k\phi(u)\,du,
\]
while entries involving \(b\) additionally depend on the one-dimensional integrals
\[
J_k(\alpha,\gamma)=\int_\alpha^\gamma l(u)\,u^k\phi(u)\,du,
\qquad
K_k(\alpha,\gamma)=\int_\alpha^\gamma l(u)^2u^k\phi(u)\,du,
\]
for $l(u):=\logit(\mu + \sigma u)$ over $u \in (\alpha, \gamma)$.
\end{enumerate}
\end{proposition}
A complete derivation of the results in Proposition~\ref{prop:zoc-tn-information} is given below, including detailed expressions for $\mathcal{I}_{\mathrm{int}}(\boldsymbol{\theta}\mid \mathbf{x})$.

The information identity given in Equation~\eqref{eq:zoctn_condfim} is standard under correct specification, while the decomposition in Equation~\eqref{eq:zoctn_fim_decomp} follows from an application of the law of total expectation. Hence all that remains to be shown are the expressions for $\mathcal{I}_0, \mathcal{I}_1$, and $\mathcal{I}_{\mathrm{int}}$ presented in Proposition~\ref{prop:zoc-tn-information}, for which it suffices to compute the conditional score and take its outer product. This is done in the following proof.

% We restate the notation used in Proposition~\ref{prop:zoc-tn-information} here for convenience:
% \[
% \mu=\mathbf{x}^\top\boldsymbol{\beta},\qquad
% \alpha=-\frac{\mu}{\sigma},\qquad
% \gamma=\frac{1-\mu}{\sigma},
% \]
% and
% \[
% q_0=\Phi(\alpha),\qquad
% % q_{\mathrm{int}}=\Phi(\gamma)-\Phi(\alpha),\qquad
% q_1=1-\Phi(\gamma).
% \]

\begin{proof}
For the boundary masses,
\[
\ell_0(\boldsymbol{\theta}\mid \mathbf{x})=\log \Phi\!\left(-\frac{\mu}{\sigma}\right),
\qquad
\ell_1(\boldsymbol{\theta}\mid \mathbf{x})=\log\!\left(1-\Phi\!\left(\frac{1-\mu}{\sigma}\right)\right).
\]
Differentiation yields
\[
s_0(\boldsymbol{\theta};\mathbf{x})=
\begin{pmatrix}
-\dfrac{\phi(\alpha)}{q_0\,\sigma}\mathbf{x}\\[1em]
-\dfrac{\alpha\phi(\alpha)}{q_0\,\sigma}\\[0.8em]
0\\
0
\end{pmatrix},
\qquad
s_1(\boldsymbol{\theta};\mathbf{x})=
\begin{pmatrix}
\dfrac{\phi(\gamma)}{q_1\,\sigma}\mathbf{x}\\[1em]
\dfrac{\gamma\phi(\gamma)}{q_1\,\sigma}\\[0.8em]
0\\
0
\end{pmatrix}.
\]
Since \(Y=0\) occurs with probability \(q_0\) and \(Y=1\) with probability \(q_1\), the corresponding information contributions are
\[
\mathcal{I}_0(\boldsymbol{\theta}\mid \mathbf{x})=q_0\,s_0(\boldsymbol{\theta};\mathbf{x})s_0(\boldsymbol{\theta};\mathbf{x})^\top,
\qquad
\mathcal{I}_1(\boldsymbol{\theta}\mid \mathbf{x})=q_1\,s_1(\boldsymbol{\theta};\mathbf{x})s_1(\boldsymbol{\theta};\mathbf{x})^\top,
\]
which gives the stated formulas.

For the interior part, write \(z = \mu + \sigma u\), so that \(u \in (\alpha,\gamma)\) and the interior contribution has subdensity \(\phi(u)\,du\) on \((\alpha,\gamma)\). Rewriting the interior log-likelihood in terms of \(u\), differentiation gives
\[
s_{\mathrm{int}}(u;\mathbf{x})=
\begin{pmatrix}
\dfrac{u}{\sigma}\mathbf{x}\\[1em]
\dfrac{u^2-1}{\sigma}\\[1em]
\dfrac{\psi(u)}{b}\\[1em]
\dfrac{l(u)\psi(u)-1}{b}
\end{pmatrix}.
\]
where 
\[
l(u)=\mathrm{logit}(z), \qquad
r(u)=z(1-z),\qquad
\psi(u)=\frac{u\,r(u)}{\sigma}-(1-2z).
\]
Therefore
\[
\mathcal{I}_{\mathrm{int}}(\boldsymbol{\theta}\mid \mathbf{x})
=
\int_\alpha^\gamma \phi(u)\, s_{\mathrm{int}}(u;\mathbf{x})s_{\mathrm{int}}(u;\mathbf{x})^\top\,du.
\]

Finally, since \(z=\mu+\sigma u\),
\[
r(u)=z(1-z)=\mu(1-\mu)+\sigma(1-2\mu)u-\sigma^2u^2,
\]
and hence
\[
\psi(u)=\frac{u\,r(u)}{\sigma}-(1-2z)
=
\left(2\mu-1\right)
+\left(2\sigma+\frac{\mu(1-\mu)}{\sigma}\right)u
+\left(1-2\mu\right)u^2
-\sigma u^3.
\]
Thus \(\psi(u)\) is cubic in \(u\), which implies that all entries not involving \(l(u)\) reduce to linear combinations of truncated normal moments \(M_k\), whereas the \(b\)-block and its cross-terms involve the additional one-dimensional integrals \(J_k\) and \(K_k\). This proves the result.
\end{proof}

\clearpage
\section{Illustration of ZOC-TN density shapes}\label{app:density_shapes}
\begin{figure}[ht!]
  \centering
  \includegraphics[width=1\linewidth]{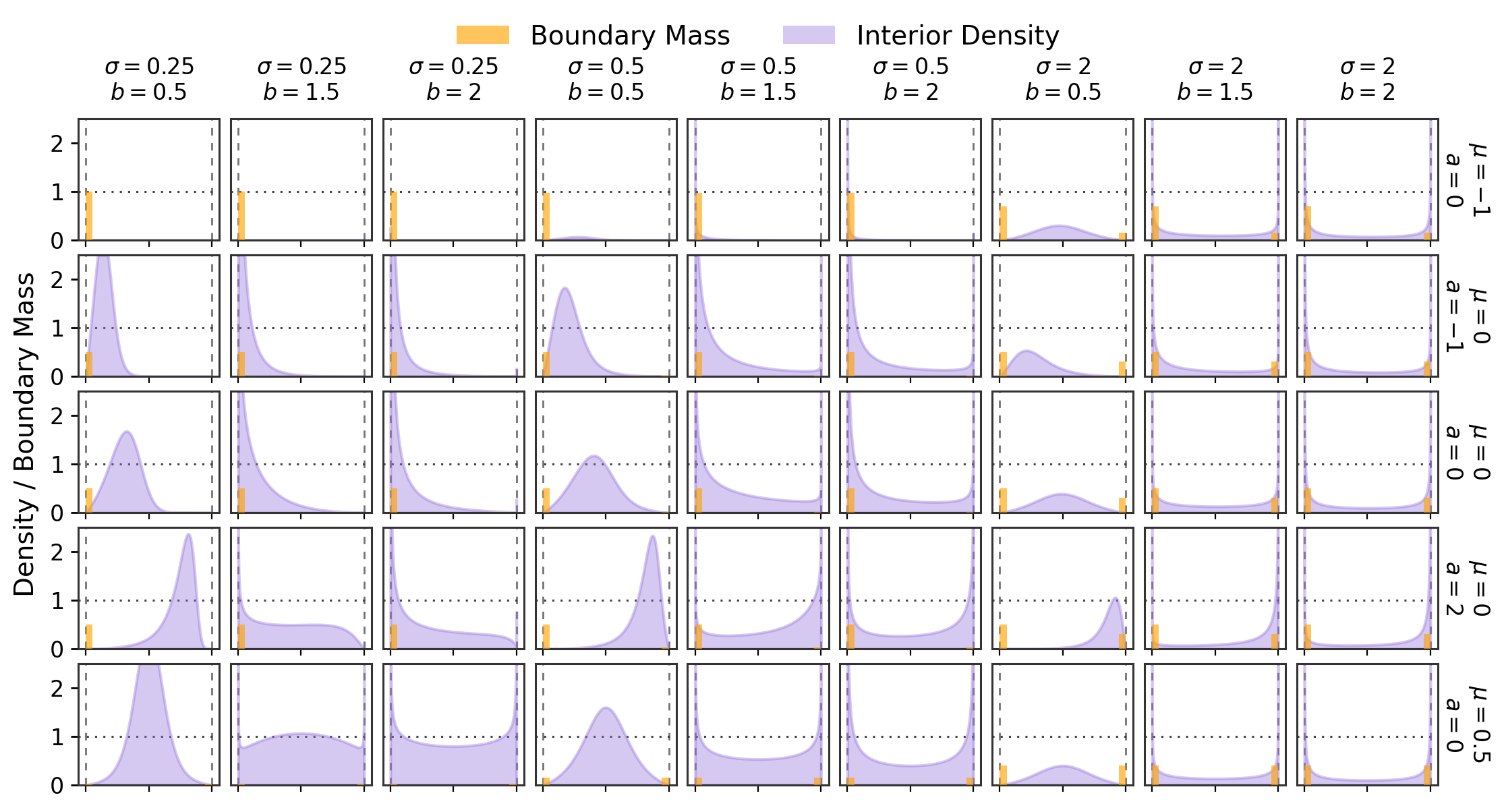}
  \caption{ ZOC-TN densities for select mean, variance, interior location, and interior concentration parameters: $\mu\in\{-1, 0, 0.5\}$, $\sigma \in \{0.25, 0.5, 2\}$, $a \in \{-1, 0, 2\}$, $b \in \{0.5, 1.5, 2\}$. The boundary probability mass (yellow bars) and the interior density (purple) are shown on the same scale.}
  \label{fig:zoctn_graphical_distributions}
\end{figure}

\newpage

\section{Overview of Benchmark Fractional Response Likelihoods}\label{app:simstudy_notation}

\subsection{Independent Linear Regression Models}
\begin{table}[ht!]
\caption{Independent linear regression fractional response models considered. The ZOC-TN model is included for comparison. The notation and parameterization of the beta and gamma distributions are discussed in subsequent sections.}
\label{tab:simstudy_models}
\centering
\renewcommand{\arraystretch}{1}
\scriptsize
\begin{tabular}{|p{1.5cm}| p{8.3cm}| p{2.2cm}| p{1.6cm}|}
\hline
\textbf{Name} & \textbf{(Quasi-)Density function} &  \textbf{Regression  parameter} & \textbf{Auxiliary parameters} \\
\hline

Quasi-B
&
\[
f^{\mathrm{BQ}}(y_i \mid \mathbf{x}_i)
=
p_i^{y_i}
+
(1 - p_i)^{(1-y_i)}
\]
&
$p_i = \mathrm{expit}(\hat{\mathbf{x}}_i^\top\boldsymbol{\beta})$
&
None
\\
\hline
ZOC-N
&
\[
f^{\mathrm{ZOC\text{-}N}}(y_i \mid \mathbf{x}_i)
=
\begin{cases}
\Phi\!\left(\dfrac{-\mu_i}{\sigma}\right), & y_i=0, \\[1.2ex]
\dfrac{1}{\sigma}\phi\!\left(\dfrac{y_i-\mu_i}{\sigma}\right), & y_i\in(0,1), \\[1.2ex]
1-\Phi\!\left(\dfrac{1-\mu_i}{\sigma}\right), & y_i=1
\end{cases}
\]
&
$\mu_i = \hat{\mathbf{x}}_i^\top\boldsymbol{\beta}$
&
\(\sigma\)
\\
\hline
BE-INF
&
\[
f^{\mathrm{BE\text{-}INF}}(y_i \mid \mathbf{x}_i)
=
\begin{cases}
\alpha(1-\gamma), & y_i=0, \\[1.2ex]
(1-\alpha)f_{(B)}(y\mid \mu_i, \varphi), & y\in(0,1) \\[1.2ex]
\alpha\gamma, & y_i=1
\end{cases}
\]
&
$\mu_i = \mathrm{expit}(\hat{\mathbf{x}}_i^\top\boldsymbol{\beta})$
&
\(\alpha, \gamma, \varphi\)
\\
\hline
ZOC-TN
&
\[
f^{\mathrm{ZOC\text{-}TN}}(y_i \mid \mathbf{x}_i)
=
\begin{cases}
\Phi\!\left(\dfrac{-\mu_i}{\sigma}\right), & y_i=0, \\[1.2ex]
\phi\!\left(
                \dfrac{z(y_i) - \mu_i}{\sigma}
            \right)
            \dfrac{z(y_i)(1-z(y_i))}{\,b\sigma y_i(1-y_i)},
        & y_i \in (0,1) \\[1.2ex]
1-\Phi\!\left(\dfrac{1-\mu_i}{\sigma}\right), & y_i=1
\end{cases}
\]
&
$\mu_i = \hat{\mathbf{x}}_i^\top\boldsymbol{\beta}$
&
$\sigma, a, b$
\\
\hline
ZOC-TB
&
\[
f^{\mathrm{ZOC\text{-}TB}}(y_i \mid \mathbf{x}_i)
=
\begin{cases}
F_{(B)}\!\left(\dfrac{u}{1+2u} \,\middle|\, \mu_i, \varphi\right), & y_i=0, \\[1.2ex]
f_{(B)}\!\left(y_i \,\middle|\, \mu_i,\varphi\right), & y_i\in(0,1), \\[1.2ex]
1-F_{(B)}\!\left(\dfrac{1+u}{1+2u} \,\middle|\, \mu_i,\varphi\right), & y_i=1
\end{cases}
\]
&
$\mu_i = \mathrm{expit}(\hat{\mathbf{x}}_i^\top\boldsymbol{\beta})$
&
\(\varphi,\ u\)
\\
\hline
ZOC-SG
&
\[
f^{\mathrm{ZOC\text{-}SG}}(y_i \mid \mathbf{x}_i)
=
\begin{cases}
F_{(G)}\!\left(\xi \,\middle|\, k,\dfrac{\mu_i}{k}\right), & y_i=0, \\[1.2ex]
f_{(G)}\!\left(y_i + \xi \,\middle|\, k,\dfrac{\mu_i}{k}\right), & y_i\in(0,1), \\[1.2ex]
1-F_{(G)}\!\left(1+\xi \,\middle|\, k,\dfrac{\mu_i}{k}\right), & y_i=1
\end{cases}
\]
&
$\mu_i = \exp(\hat{\mathbf{x}}_i^\top\boldsymbol{\beta})$
&
\(k,\ \xi\)
\\

\hline
\end{tabular}
\end{table}

\clearpage
\subsection{Zero-one Censored Transformed Beta (ZOC-TB)}
Introduced as the \textit{extended-support beta} distribution in \cite{Kosmidis2025}, the ZOC-TB likelihood adapts the beta distribution to accommodate both continuous observations over $(0, 1)$ and positive point-masses on the boundaries by recasting the fractional response $Y_i$ as a censored observation of a latent variable, $Y_i^*$, following the so-called four-parameter beta distribution: $Y_i^* \sim \text{B}4(\mu_i, \phi, -u, 1+u)$. This four-parameter beta distribution is specified by density
\begin{equation*}
    f_{(\text{B}4)}(y \; | \; \mu, \phi, u_1, u_2) = f_{(\text{B})}\left(\frac{y - u_1}{u_2 - u_1} \; | \; \mu, \phi \right)\frac{1}{u_2 - u_1}
\end{equation*}
for beta density $f_{(\text{B})}(\cdot \; |\; \mu, \phi)$ with $\mu \in (0, 1)$, $\phi > 0$, and $u_1 < u_2$ modified boundary parameters. In our case, we set $u_1 = -u$ and $u_2 = 1 + u$, with $u > 0$ representing a symmetric ad hoc exceedance to the left and right of $(0, 1)$. 

Under the parameterization used in Table~\ref{tab:simstudy_models}, we denote
\begin{equation*}
    f_{(\text{B})}(y \; |\; \mu, \phi) = \frac{\mathbb{I}\{0 < y < 1\}}{\text{B}(\mu\phi, (1 - \mu)\phi)}y^{\mu\phi - 1}(1 - y)^{(1 - \mu)\phi - 1}
\end{equation*}
with indicator function $\mathbb{I}\{\cdot\}$ and beta function, $\text{B}(\cdot, \cdot)$. This is consistent with the definition presented in \cite{Kosmidis2025}.

\subsection{Zero-one Censored Shifted Gamma (ZOC-SG)}
% TODO: FINISH
The zero-one censored shifted gamma (ZOC-SG) distribution \citep{Sigrist2011} extends the classical gamma regression framework to accommodate observations at the boundaries of the unit interval. This model considers a shifted latent response, $(Y_i^* + \xi) \sim \text{Gamma}\left(k, \theta_i\right)$ for given shape $k>0$, scale $\theta_i > 0$, and shift $\xi > 0$. 

For a given shape and scale parameter, the Gamma density function used in Table~\ref{tab:simstudy_models} is defined to be
\begin{equation*}
    f_{(\text{G})}(y \; |\; k, \theta) = \frac{1}{\theta^k\Gamma(k)}y^{k-1}e^{-y/\theta}, \;\; y>0
\end{equation*}
where $\Gamma(\cdot)$ is the gamma function. 

\section{Proof of Multiple Stationary Points on the Interior of ZOC-TN}\label{app:zoctn_vs_benchmark_proof}
The following proof demonstrates the third point outlined at the end of Section~\ref{subsec:zoctn_vs_benchmarks} on the structural advantages of the ZOC-TN likelihood when compared to its counterparts.

\begin{proof}
% For (i), observe that if \(a=0\) and \(b=1\), then
% \[
% g_{\mathrm{ZOC\text{-}TN}}(z;0,1)=\expit(\logit z)=z,
% \]
% so the interior transformation is the identity and the ZOC-TN likelihood
% reduces exactly to the ZOC-N likelihood.

% For (i), Bernoulli QMLE specifies only the conditional mean
% \[
% \mathbb{E}[Y \mid \mathbf{x}]=\expit(\mathbf{x}^\top\boldsymbol{\beta})\in(0,1),
% \]
% and hence does not define a mixed predictive law on \([0,1]\) with atoms at
% the boundaries. In contrast, the ZOC-TN model inherits boundary masses from
% the censored latent normal law, yielding 
% \begin{equation}\label{eq:zoctn_bdmass}
%     \Pr(Y=0\mid \mathbf{x})=\Phi\!\left(-\frac{\mu}{\sigma}\right),\qquad
%     \Pr(Y=1\mid \mathbf{x})=1-\Phi\!\left(\frac{1-\mu}{\sigma}\right),
% \end{equation}
    
% For (ii), the displayed boundary-mass formulas in Equation~\ref{eq:zoctn_bdmass} contain only \((\mu,\sigma)\). The parameters \(a\) and \(b\) enter only through the inverse transformation
% \[
%     z(y)=\expit\!\left(\frac{\logit(y)-a}{b}\right)
% \]
% in the interior density. Hence they alter only the interior shape.

The interior density of the censored shifted gamma model is
\[
f^{\text{ZOC-SG}}(y)\propto y^{k-1}e^{-y/\theta},\qquad y\in(0,1),
\]
so
\[
\frac{d}{dy}\log f^{\text{ZOC-SG}}(y)=\frac{k-1}{y}-\frac1\theta,
\]
which has at most one zero on \((0,1)\). Hence \(f^{\text{ZOC-SG}}\) has at most
one interior stationary point.

Likewise, the interior density of the transformed-beta /
extended-support-beta model is proportional to a beta density,
\[
f^{\text{ZOC-TB}}(y)\propto y^{\alpha-1}(1-y)^{\beta-1},\qquad y\in(0,1),
\]
so
\[
\frac{d}{dy}\log f^{\text{ZOC-TB}}(y)
=\frac{\alpha-1}{y}-\frac{\beta-1}{1-y},
\]
which also has at most one zero on \((0,1)\). Hence \(f^{\text{ZOC-TB}}\) has at
most one interior stationary point.
% the formulas for the derivatives of the log-densities of the
% censored shifted gamma and transformed-beta models are immediate from
% their interior densities. Each derivative has at most one zero on \((0,1)\),
% so each model has at most one interior stationary point.

% For the ZOC-TN model, the interior density is
% \[
% f_{\mathrm{TN}}(y\mid x)
% =
% \frac1\sigma
% \phi\!\left(\frac{z(y)-\mu}{\sigma}\right)
% \frac{z(y)(1-z(y))}{b\,y(1-y)},
% \qquad
% z(y)=\expit\!\left(\frac{\logit(y)-a}{b}\right).
% \]
% Writing \(y=\expit(\eta)\), differentiation gives
By contrast, for the ZOC-TN model, writing \(y=\expit(\eta)\) and
\[
z=\expit\!\left(\frac{\eta-a}{b}\right),
\]
the interior log-density satisfies
\[
\frac{d}{d\eta}\log f^{\text{ZOC-TN}}(y(\eta)\mid x)
=
-\frac{(z-\mu)z(1-z)}{b\sigma^2}
+\frac{1-2z}{b}
-(1-2y).
\]
This need not be monotone and may have multiple zeros. For example, if
\[
\mu=0.1,\qquad \sigma=0.2,\qquad a=2,\qquad b=1.5,
\]
then the above derivative is negative at \(y=0.2\), positive at \(y=0.4\),
and negative again at \(y=0.7\). By continuity, it therefore has at least two distinct zeros in \((0,1)\). Consequently, the ZOC-TN interior density can exhibit at least two interior stationary points, which is impossible under ZOC-SG, ZOC-TB, and BE-INF.
% \[
% \frac{d}{d\eta}\log f_{\mathrm{TN}}(y(\eta)\mid \mathbf{x})
% =
% -\frac{(z-\mu)z(1-z)}{b\sigma^2}
% +\frac{1-2z}{b}
% -(1-2y).
% \]
% For \((\mu,\sigma,a,b)=(0.1,0.2,2,1.5)\), direct evaluation shows that this
% quantity is negative at \(y=0.2\), positive at \(y=0.4\), and negative at
% \(y=0.7\). Hence, by the intermediate value theorem, it has at least two
% distinct zeros in \((0,1)\). Therefore the ZOC-TN interior density can have
% more than one interior stationary point, unlike the ZOC-SG and ZOC-TB
% families. This proves the claim.
\end{proof}

\section{Gradients for the Tree-boosted ZOC-TN Model}\label{app:zoc_tn_boosting_gradients}

This appendix gives the derivatives used for gradient tree boosting in the independent ZOC-TN model. For observation \(i\), write \(\mu_i=F(\mathbf{X}_i)\), and define
\[
    \alpha_i=-\frac{\mu_i}{\sigma},
    \qquad
    \gamma_i=\frac{1-\mu_i}{\sigma}.
\]
For \(y_i\in(0,1)\), let
\[
    z_i=z_{a,b}(y_i)
    =
    \operatorname{expit}\left(\frac{\logit(y_i)-a}{b}\right).
\]
The log-likelihood contribution is
\[
\ell_i
=
\ell^{\mathrm{ZOC-TN}}(\mu_i,\sigma,a,b\mid y_i).
\]

For boundary observations at zero,
\[
    \ell_i = \log \Phi(\alpha_i),
\]
and hence
\[
    \frac{\partial \ell_i}{\partial \mu_i}
    =
    -\frac{1}{\sigma}
    \frac{\phi(\alpha_i)}{\Phi(\alpha_i)},
\]
\[
    \frac{\partial^2 \ell_i}{\partial \mu_i^2}
    =
    -\frac{1}{\sigma^2}
    \frac{\phi(\alpha_i)}{\Phi(\alpha_i)}
    \left\{
        \alpha_i+
        \frac{\phi(\alpha_i)}{\Phi(\alpha_i)}
    \right\}.
\]

For boundary observations at one,
\[
    \ell_i = \log\{1-\Phi(\gamma_i)\},
\]
and therefore
\[
    \frac{\partial \ell_i}{\partial \mu_i}
    =
    \frac{1}{\sigma}
    \frac{\phi(\gamma_i)}{1-\Phi(\gamma_i)},
\]
\[
    \frac{\partial^2 \ell_i}{\partial \mu_i^2}
    =
    \frac{1}{\sigma^2}
    \frac{\phi(\gamma_i)}{1-\Phi(\gamma_i)}
    \left\{
        \gamma_i-
        \frac{\phi(\gamma_i)}{1-\Phi(\gamma_i)}
    \right\}.
\]

For interior observations \(y_i\in(0,1)\), the terms depending on \(\mu_i\) enter through the Gaussian density of \(z_i\). Thus,
\[
    \frac{\partial \ell_i}{\partial \mu_i}
    =
    \frac{z_i-\mu_i}{\sigma^2},
\]
and
\[
    \frac{\partial^2 \ell_i}{\partial \mu_i^2}
    =
    -\frac{1}{\sigma^2}.
\]

Equivalently, for the negative log-likelihood loss
\[
    \rho_i(F)
    =
    -\ell^{\mathrm{ZOC-TN}}(F(\mathbf{X}_i),\sigma,a,b\mid y_i),
\]
the first- and second-order quantities used in Newton boosting are
\[
    g_i
    =
    \frac{\partial \rho_i}{\partial F(\mathbf{X}_i)}
    =
    -\frac{\partial \ell_i}{\partial \mu_i},
    \qquad
    h_i
    =
    \frac{\partial^2 \rho_i}{\partial F(\mathbf{X}_i)^2}
    =
    -\frac{\partial^2 \ell_i}{\partial \mu_i^2}.
\]
% The tree fitted at boosting iteration \(m\) is therefore obtained from the local quadratic approximation
% \[
%     \sum_{i=1}^N
%     \left[
%         g_i f(\mathbf{X}_i)
%         +
%         \frac{1}{2}h_i f(\mathbf{X}_i)^2
%     \right],
% \]
% possibly with additional regularization on the tree structure and leaf weights. After fitting \(f_m\), the predictor is updated according to
% \[
%     F_m(\mathbf{x})
%     =
%     F_{m-1}(\mathbf{x})
%     +
%     \nu f_m(\mathbf{x}),
% \]
% where \(\nu\in(0,1]\) denotes the learning rate.

\clearpage

\section{Simulation Study: Experimental Settings}\label{app:simstudy_mds}
The covariates consist of three independent standard normal covariates and an intercept term:
$$
    \mathbf{X}_i = \begin{pmatrix}
        1 \\ \mathbf{Z}_i
    \end{pmatrix}, \mathbf{Z}_i \sim \mathcal{N}_3(\mathbf{0}, \mathbf{I}).
$$

\begin{table}[ht!]
\caption{Data generating processes used in simulation study. }\label{tab:simstudy_dgps}
\centering
\renewcommand{\arraystretch}{1.5}
\scriptsize
\begin{tabular}{|p{1.5cm}| p{7.0cm}| p{5.0cm}|}
\hline
\textbf{DGP} & \textbf{Bounded response model} &  \textbf{True parameters} \\\hline
ZOC-TN&
\[Y_i \sim \text{ZOC-TN}(\mu_i = ({\mathbf{X}}_i^\top\boldsymbol{\beta}_{\text{ZOC-TN}}), \sigma^2, a, b) \]
&
$\boldsymbol{\beta}_{\text{ZOC-TN}} = (0.5, 0.005, -0.015, 0.01)^\top$, $(\sigma, a, b) = (0.25, 0.0, 1.5)$
\\\hline
ZOC-TB&
\[Y_i \sim \text{ZOC-TB}(\mu_i = \mathrm{expit}({\mathbf{X}}_i^\top\boldsymbol{\beta}_{\text{ZOC-TB}}), \varphi, u)\]
&
$\boldsymbol{\beta}_{\text{ZOC-TB}} = (0.5, 0.005, -0.015, 0.01)^\top$, $(\varphi, u) = (5.0, 0.2)$
\\\hline
ZOC-SG&
\[Y_i \sim \text{ZOC-SG}(\mu_i = \mathrm{exp}({\mathbf{X}}_i^\top\boldsymbol{\beta}_{\text{ZOC-SG}}), k, \xi)\]
&
$\boldsymbol{\beta}_{\text{ZOC-SG}} = (0.05, 0.005, -0.015, 0.01)^\top$, $(k, \xi) = (20, 0.6)$
\\\hline
Right-Skew (MD)&
\[Y_i \sim \text{MD}_{\text{R}}(p_0({\mathbf{X}}_i), p_1({\mathbf{X}}_i), \mu_i = \mathrm{expit}({\mathbf{X}}_i^\top\boldsymbol{\beta}_{\text{R}}), \kappa)\]
&
$\boldsymbol{\beta}_{\text{R}} = (1.4, 0.1, 0.7, -0.05)^\top$, $\kappa = 15.0$
\\\hline
W-Shape (MD)&
\[Y_i \sim \text{MD}_{\text{W}}(p_0({\mathbf{X}}_i), p_1({\mathbf{X}}_i), \mu_i = \mathrm{expit}({\mathbf{X}}_i^\top\boldsymbol{\beta}_{\text{W}}), \mathbf{m}, \boldsymbol{\kappa})\]
&
$\boldsymbol{\beta}_{\text{W}} = (0.0, 0.1, 0.05, -0.15)^\top$, \; \;\;\; \;\;\; $\mathbf{m} = (0.15, 0.5, 0.85)$, $\boldsymbol{\kappa} = (5.0, 20.0, 5.0)$
\\\hline
\end{tabular}
\end{table}

The two Right-Skew and W-Shape mixture distributions contain point masses at $0$ and $1$ specified as
\begin{equation*}
\begin{split}
        \mathrm{Pr}(Y_i = 0 \mid \mathbf{X}_i) &= p_0(\mathbf{X}_i) = \frac{\exp(\mathbf{X}_i^\top\boldsymbol{\eta}_0)}{1 + \exp(\mathbf{X}_i^\top\boldsymbol{\eta}_0) + \exp(\mathbf{X}_i^\top\boldsymbol{\eta}_1)}, \\
        \mathrm{Pr}(Y_i = 1 \mid \mathbf{X}_i) &= p_1(\mathbf{X}_i) = \frac{\exp(\mathbf{X}_i^\top\boldsymbol{\eta}_1)}{1 + \exp(\mathbf{X}_i^\top\boldsymbol{\eta}_0) + \exp(\mathbf{X}_i^\top\boldsymbol{\eta}_1)},
\end{split}
\end{equation*}
where 
\begin{align*}
    \boldsymbol{\eta}_0 &= (-4.5, -3.0, 1.0, -0.5)^\top, \\ 
    \boldsymbol{\eta}_1 &= (-2.0, 1.0, -0.5, 0.25)^\top.
\end{align*}
The probability of $Y_i$ belonging to the continuous component is then $p_c(\mathbf{X}_i) = 1 - p_0(\mathbf{X}_i) - p_1(\mathbf{X}_i)$.

 % In the case of the right-skewed mixture distribution $\text{MD}_{\text{R}}$, this interior density consists of only a beta distribution, while the W-Shape mixture distribution $\text{MD}_{\text{W}}$ consists of a mixture of three beta components.
% Above, we introduce two mixture distributions referred to as the Right-Skew and W-Shape mixture distributions, denoted $\text{MD}_{\text{R}}$ and $\text{MD}_{\text{W}}$ respectively. Under these models the distribution of the bounded random variable $Y_i \in [0, 1]$ depends on its corresponding feature set, $\mathbf{X}_i \in \mathbb{R}^3$, consisting of three covariates. 

% As mentioned in the main text, the boundary mass functions are defined equivalently for both mixture DGPs. For a given $\mathbf{x}_i = (1 \;x_1\; x_2\;x_3)^\top$,
% \begin{equation*}
%     \mathrm{Pr}(Y_i=0 \mid \mathbf{x}_i) = p_0(\mathbf{x}_i), \;\; \mathrm{Pr}(Y_i=1 \mid \mathbf{x}_i) = p_1(\mathbf{x}_i)
% \end{equation*}
% where
% \begin{equation*}
%     p_0(\mathbf{x}_i) = \frac{\exp(\mathbf{x}_i^\top\boldsymbol{\eta}_0)}{1 + \exp(\mathbf{x}_i^\top\boldsymbol{\eta}_0) + \exp(\mathbf{x}_i^\top\boldsymbol{\eta}_1)}, \;\; p_1(\mathbf{x}_i) = \frac{\exp(\mathbf{x}_i^\top\boldsymbol{\eta}_1)}{1 + \exp(\mathbf{x}_i^\top\boldsymbol{\eta}_0) + \exp(\mathbf{x}_i^\top\boldsymbol{\eta}_1)}.
% \end{equation*}

Over $Y_i \in (0, 1)$, the Right-Skew interior is given by a single right-skewed beta distribution while the W-Shape interior is made from a mixture of three beta distributions. Starting with $\text{MD}_{\text{R}}$,
\begin{equation*}
    Y_i \mid (Y_i \in (0,1), \; \mathbf{x}_i) \sim \text{Beta}(\mu = \expit(\mathbf{x}_i^\top \boldsymbol{\beta}_R), \phi=\kappa)
\end{equation*}
using the same $\text{Beta}(\mu, \phi)$ parameterization as in Appendix~\ref{app:simstudy_notation}. For data generation, we set
\begin{equation*}
    \boldsymbol{\beta}_R = (1.4, 0.1, 0.7, -0.05)^\top,\; \kappa = 15.0
\end{equation*}

Construction of the W-Shape mixture distribution involves a mixture of three beta components. In order to produce the triple-peaked empirical distribution, we introduce a latent class variable, $S_i \in \{1, 2,3\}$, conditional on $Y_i \in (0, 1)$:
\begin{equation*}
    \mathrm{Pr}(S= j \mid Y_i \in (0, 1), \;\mathbf{x}_i) = \pi_j(\mathbf{x}_i)
\end{equation*}
such that $\pi_1(\mathbf{x}_i) + \pi_2(\mathbf{x}_i) + \pi_3(\mathbf{x}_i) = 1$. We compute these probabilities through a similar softmax method as before,
\begin{equation*}
    \pi_j(\mathbf{x}_i) = \frac{r_j(\mathbf{x}_i)}{r_1(\mathbf{x}_i) + r_2(\mathbf{x}_i) + r_3(\mathbf{x}_i)}
\end{equation*}
where 
\begin{equation*}
    r_1(\mathbf{x}_i) = \exp(-\tanh \mathbf{x}_i^\top\boldsymbol{\beta}_W),\; r_2(\mathbf{x}_i) = \exp(c),\; r_3(\mathbf{x}_i) = \exp(\tanh \mathbf{x}_i^\top\boldsymbol{\beta}_W)
\end{equation*}
Under this setting $c$ can be interpreted as the relative strength of the central peak against the two shoulders. From construction we see that smaller values of $(\mathbf{x}_i^\top\boldsymbol{\beta}_W)$ are more likely to be assigned to class $S_i = 1$, while larger values are more likely to fall into $S_i = 3$. The hyperbolic tangent function is applied to bound the weight contribution to $\pm 1$. We define
\begin{equation*}
    \boldsymbol{\beta}_W = (0.0, 0.1, 0.05, -0.15)^\top,\; c = 0.8
\end{equation*}

Each of these classes corresponds to its own beta distribution with prescribed mean and concentration parameters. Hence
\begin{equation*}
    Y_i \mid (Y_i \in (0,1),\; S_i =j) \sim \text{Beta}(m_j, \kappa_j)
\end{equation*}
For simulation these parameters are set to
\begin{align*}
    (m_1, m_2, m_3) &= (0.15, 0.5, 0.85)^\top \\
    (\kappa_1, \kappa_2, \kappa_3) &= (5.0, 20.0, 5.0)^\top
\end{align*}
Again, these configurations were selected to achieve the ``head and shoulder" peaks of the $W$-Shape interior, with $S_2$ responsible for the main central mode, and $S_1$, $S_3$ generating the two flanks.

\begin{figure}[ht!]
  \centering
  \includegraphics[width=1\linewidth]{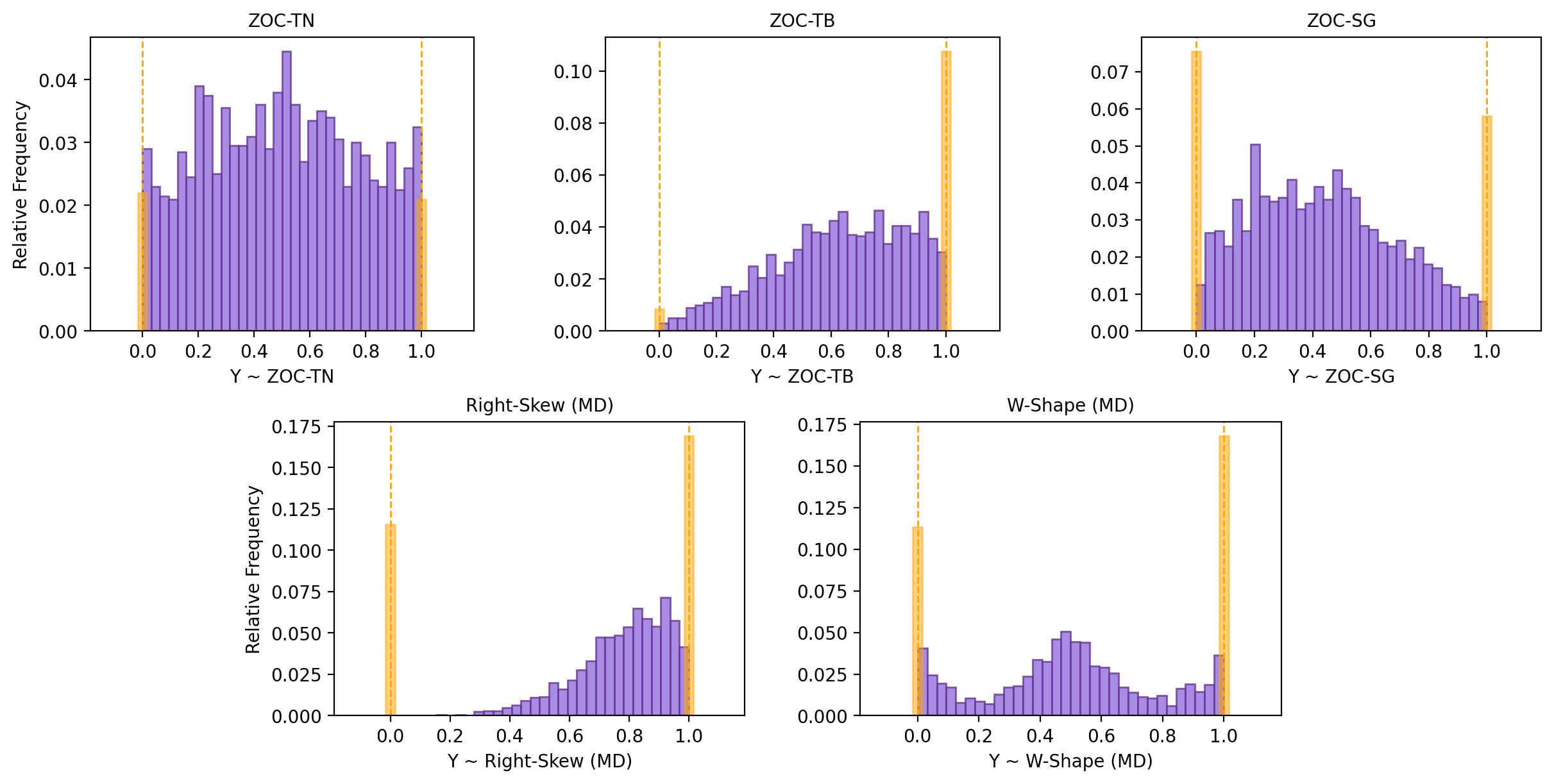}
  \caption{Randomly selected datasets used in simulation study. Each dataset consists of $2,000$ ($1,000$ train $+$ $1,000$ test) observations. Yellow bars represent boundary observations, purple bins cover the interior.}
  \label{fig:simulated_empirical_distribution}
\end{figure}

\clearpage
\section{Simulation Study Results}\label{app:simstudy_vanillapit}
\begin{table}[ht!]
\caption{Simulation study model performance across different DGPs. Average performance and (standard deviation) reported across 100 simulations. Runtimes reported in seconds.}
\label{tab:simstudy_summary_all}
\if1\jbes
  \tiny
\else
  \scriptsize
\fi
\centering
\begin{tabular}{lrrrrrrr}
\toprule
Method & Runtime & BIC & MSE & Log Score & CRPS & $\mathrm{AE}(p_0)$ & $\mathrm{AE}(p_1)$ \\
\midrule
\multicolumn{8}{l}{\rule{0pt}{2.4ex}\textbf{ZOC-TN Dataset}} \\
\midrule
\cellcolor{gray!20}ZOC-TN   & \cellcolor{gray!20} 0.156 & \cellcolor{gray!20} $\mathbf{448}$  & \cellcolor{gray!20} $\mathbf{0.0846}$ & \cellcolor{gray!20} $\mathbf{0.209}$ & \cellcolor{gray!20} $\mathbf{0.169}$ & \cellcolor{gray!20} 4.34 & \cellcolor{gray!20} 4.18 \\
    \cellcolor{gray!20}     & \cellcolor{gray!20} (0.0836) & \cellcolor{gray!20} (51.1) & \cellcolor{gray!20} (0.00267) & \cellcolor{gray!20} (0.0235) & \cellcolor{gray!20} (0.00290) & \cellcolor{gray!20} (2.45) & \cellcolor{gray!20} (2.23) \\
ZOC-TB   & 1.14  & 452  & $\mathbf{0.0846}$ & 0.213 & 0.170 & $\mathbf{4.05}$ & $\mathbf{3.89}$ \\
         & (0.120) & (51.3) & (0.00267) & (0.0232) & (0.00271) & (1.89) & (1.87) \\
ZOC-SG   & 9.64  & 575  & $\mathbf{0.0846}$ & 0.274 & 0.170 & 23.9 & 28.6 \\
         & (10.5) & (44.9) & (0.00266) & (0.0223) & (0.00295) & (5.27) & (4.38) \\
Quasi-B  & $\mathbf{0.0369}$ & 1411 & $\mathbf{0.0846}$ & 0.693 & 0.252 & 23.1 & 23.0 \\
         & (0.0346) & (1.72) & (0.00267) & (0.000915) & (0.000986) & (0.136) & (0.137) \\
ZOC-N    & 0.116 & 569  & $\mathbf{0.0846}$ & 0.273 & 0.170 & 26.4 & 26.2 \\
         & (0.00679) & (44.7) & (0.00266) & (0.0222) & (0.00300) & (4.48) & (4.06) \\
BE-INF   & 0.154 & 462  & 0.0847 & 0.217 & 0.433 & 5.19 & 5.00 \\
         & (0.102) & (52.1) & (0.00270) & (0.0233) & (0.0249) & (2.28) & (1.92) \\
\midrule
\multicolumn{8}{l}{\rule{0pt}{2.4ex}\textbf{ZOC-TB Dataset}} \\
\midrule
ZOC-TN  & 0.150 & 670  & $\mathbf{0.0699}$ & 0.324 & $\mathbf{0.151}$ & 2.28 & 12.1 \\
        & (0.0778) & (61.2) & (0.00260) & (0.0283) & (0.00299) & (1.52) & (4.80) \\
\cellcolor{gray!20}ZOC-TB  & \cellcolor{gray!20} 1.23  & \cellcolor{gray!20} $\mathbf{662}$  & \cellcolor{gray!20} $\mathbf{0.0699}$ & \cellcolor{gray!20} $\mathbf{0.323}$ & \cellcolor{gray!20} $\mathbf{0.151}$ & \cellcolor{gray!20} $\mathbf{1.63}$ & \cellcolor{gray!20} 12.2 \\
    \cellcolor{gray!20}    & \cellcolor{gray!20}(0.171) & \cellcolor{gray!20}(61.2) & \cellcolor{gray!20}(0.00260) & \cellcolor{gray!20}(0.0282) & \cellcolor{gray!20}(0.00304) & \cellcolor{gray!20}(0.959) & \cellcolor{gray!20}(4.66) \\
ZOC-SG  & 24.3  & 682  & $\mathbf{0.0699}$ & 0.331 & 0.152 & 3.30 & 18.4 \\
        & (7.56) & (58.9) & (0.00258) & (0.0266) & (0.00300) & (1.76) & (7.29) \\
Quasi-B & $\mathbf{0.0328}$ & 1305 & $\mathbf{0.0699}$ & 0.640 & 0.226 & 7.68 & 123 \\
        & (0.00369) & (10.8) & (0.00260) & (0.00603) & (0.00300) & (0.0144) & (0.151) \\
ZOC-N   & 0.137 & 675  & $\mathbf{0.0699}$ & 0.331 & 0.152 & 3.32 & 18.4 \\
        & (0.108) & (58.9) & (0.00258) & (0.0266) & (0.00295) & (1.78) & (7.28) \\
BE-INF  & 0.135 & 674  & $\mathbf{0.0699}$ & 0.327 & 0.173 & 2.11 & $\mathbf{9.16}$ \\
        & (0.00493) & (61.2) & (0.00259) & (0.0284) & (0.00982) & (1.61) & (5.46) \\
\midrule
\multicolumn{8}{l}{\rule{0pt}{2.4ex}\textbf{ZOC-SG Dataset}} \\
\midrule
ZOC-TN   & 0.145 & 861  & $\mathbf{0.0855}$ & 0.419 & $\mathbf{0.169}$ & 7.78 & $\mathbf{7.95}$ \\
        & (0.00496) & (62.6) & (0.00252) & (0.0253) & (0.00281) & (3.54) & (4.22) \\
ZOC-TB   & 1.47  & 884  & 0.0856 & 0.432 & $\mathbf{0.169}$ & 19.1 & 20.0 \\
        & (0.464) & (62.4) & (0.00251) & (0.0253) & (0.00261) & (6.28) & (6.12) \\
\cellcolor{gray!20}ZOC-SG   & \cellcolor{gray!20}0.426 & \cellcolor{gray!20}$\mathbf{855}$  & \cellcolor{gray!20}$\mathbf{0.0855}$ & \cellcolor{gray!20}$\mathbf{0.418}$ & \cellcolor{gray!20}$\mathbf{0.169}$ & \cellcolor{gray!20}$\mathbf{7.24}$ & \cellcolor{gray!20}8.42 \\
    \cellcolor{gray!20}    & \cellcolor{gray!20}(0.106) & \cellcolor{gray!20}(62.3) & \cellcolor{gray!20}(0.00252) & \cellcolor{gray!20}(0.0252) & \cellcolor{gray!20}(0.00263) & \cellcolor{gray!20}(3.42) & \cellcolor{gray!20}(4.19) \\
Quasi-B & $\mathbf{0.0404}$ & 1400 & $\mathbf{0.0855}$ & 0.687 & 0.249 & 65.5 & 63.5 \\
        & (0.0622) & (4.65) & (0.00253) & (0.00242) & (0.00143) & (0.188) & (0.253) \\
ZOC-N   & 0.120 & 879  & $\mathbf{0.0855}$ & 0.432 & $\mathbf{0.169}$ & 22.6 & 18.3 \\
        & (0.00542) & (62.5) & (0.00251) & (0.0251) & (0.00264) & (6.20) & (6.18) \\
BE-INF  & 0.143 & 879  & 0.0856 & 0.427 & 0.392 & 7.51 & 9.22 \\
        & (0.0606) & (62.1) & (0.00259) & (0.0262) & (0.00916) & (3.00) & (3.03) \\
\midrule
\multicolumn{8}{l}{\rule{0pt}{2.4ex}\textbf{Right-Skewed (Mixture Distribution) Dataset}} \\
\midrule
ZOC-TN   & 0.157 & 347 & 0.0583 & 0.162 & $\mathbf{0.119}$ & 86.3 & $\mathbf{77.6}$ \\
        & (0.0644) & (75.1) & (0.00308) & (0.0331) & (0.00340) & (5.02) & (6.59) \\
ZOC-TB   & 1.34 & 904 & 0.0608 & 0.440 & 0.131 & $\mathbf{80.1}$ & 94.0 \\
        & (0.176) & (63.6) & (0.00279) & (0.0275) & (0.00341) & (7.44) & (8.05) \\
ZOC-SG   & 109 & 933 & 0.0597 & 0.454 & 0.130 & 83.8 & 112 \\
        & (23.1) & (62.6) & (0.00254) & (0.0278) & (0.00314) & (7.25) & (8.21) \\
Quasi-B & $\mathbf{0.0363}$ & 1034 & $\mathbf{0.0568}$ & 0.506 & 0.166 & 108 & 162 \\
        & (0.00768) & (17.0) & (0.00286) & (0.00922) & (0.00386) & (6.69) & (5.04) \\
ZOC-N   & 0.127 & 926 & 0.0597 & 0.454 & 0.130 & 83.8 & 112 \\
        & (0.0101) & (62.6) & (0.00255) & (0.0277) & (0.00322) & (7.25) & (8.15) \\
BE-INF  & 0.145 & $\mathbf{59.5}$ & 0.0868 & $\mathbf{0.0131}$ & 0.204 & 143 & 113 \\
        & (0.00819) & (85.9) & (0.00542) & (0.0366) & (0.00806) & (7.22) & (4.50) \\
\midrule
\multicolumn{8}{l}{\rule{0pt}{2.4ex}\textbf{W-Shape (Mixture Distribution) Dataset}} \\
\midrule
ZOC-TN  & 0.184 & $\mathbf{1238}$ & 0.0991 & $\mathbf{0.602}$ & $\mathbf{0.177}$ & 73.6 & 21.9 \\
       & (0.0617) & (39.6) & (0.00405) & (0.0199) & (0.00410) & (4.38) & (4.76) \\
ZOC-TB  & 1.63 & 1249 & 0.0993 & 0.611 & 0.179 & $\mathbf{71.2}$ & $\mathbf{20.3}$ \\
       & (0.178) & (38.9) & (0.00410) & (0.0198) & (0.00433) & (4.91) & (4.85) \\
ZOC-SG  & 1.72 & 1268 & 0.0987 & 0.618 & $\mathbf{0.177}$ & 80.7 & 23.4 \\
       & (5.76) & (39.3) & (0.00392) & (0.0206) & (0.00405) & (4.66) & (5.24) \\
Quasi-B     & $\mathbf{0.0419}$ & 1311 & $\mathbf{0.0982}$ & 0.643 & 0.229 & 108 & 162 \\
       & (0.00530) & (12.0) & (0.00361) & (0.00773) & (0.00370) & (6.62) & (4.83) \\
ZOC-N   & 0.148 & 1265 & 0.0986 & 0.619 & $\mathbf{0.177}$ & 82.1 & 22.7 \\
       & (0.0237) & (39.8) & (0.00390) & (0.0207) & (0.00401) & (4.81) & (4.50) \\
BE-INF  & 0.169 & 1549 & 0.120 & 0.758 & 0.343 & 142 & 113 \\
       & (0.0112) & (42.3) & (0.00476) & (0.0271) & (0.0200) & (6.85) & (3.84) \\
\bottomrule
\end{tabular}
\end{table}

\begin{figure}[ht!]
  \centering
  \includegraphics[width=1.0\linewidth]{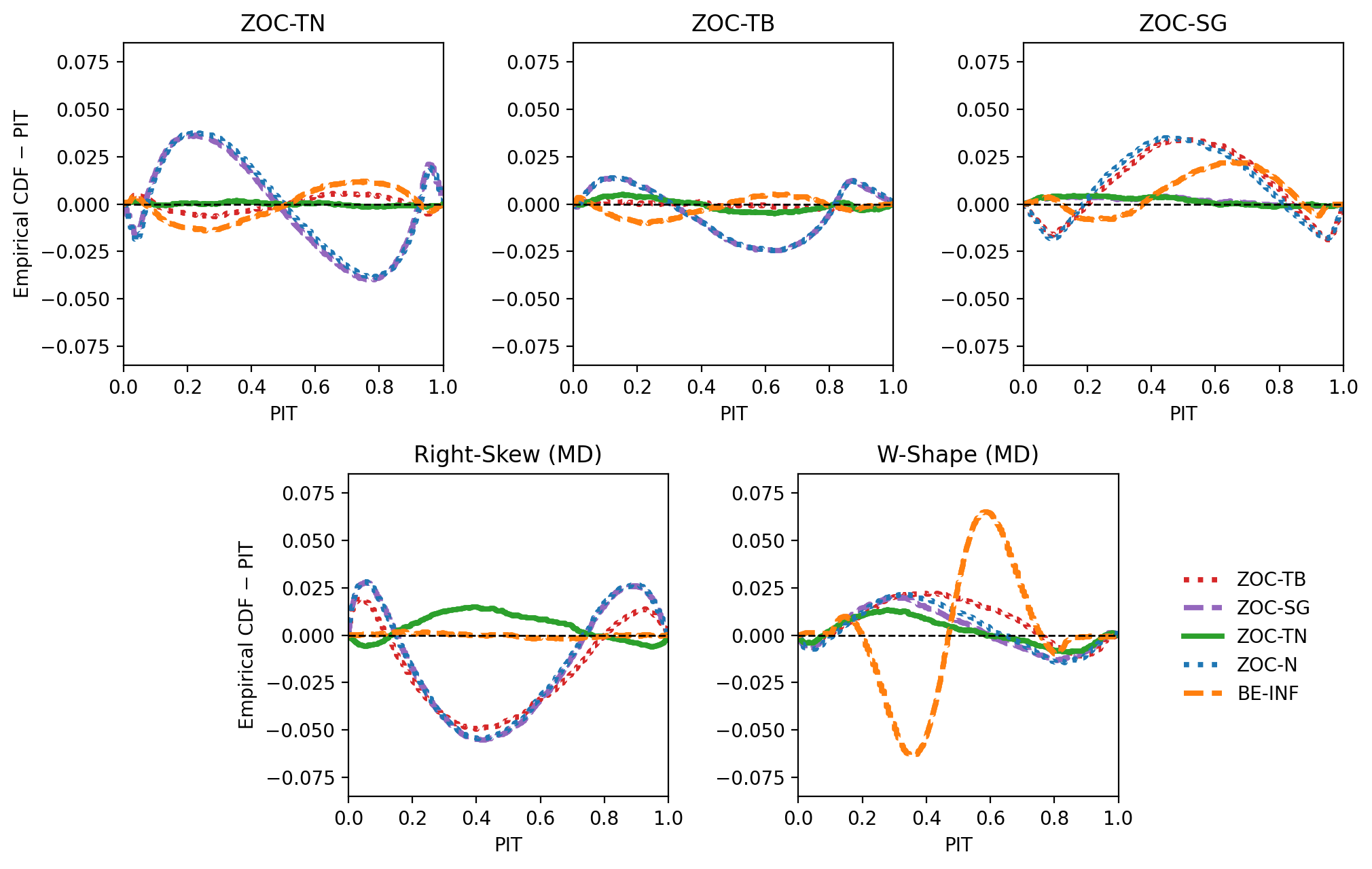}
  \caption{PIT reliability residual diagrams for competing models on different DGPs; predictions pooled across all test sets. Each curve plots the difference between the empirical PIT CDF and the uniform reference distribution. Deviations away from zero indicate miscalibration.}
  \label{fig:simstudy_pitresidual}
\end{figure}

In addition to the PIT reliability residual diagrams presented in Figure~\ref{fig:simstudy_pitresidual}, we provide the vanilla randomized PIT reliability QQ-plots in Figure~\ref{fig:simstudy_pitvanilla}.

\begin{figure}[ht!]
  \centering
  \includegraphics[width=0.8\linewidth]{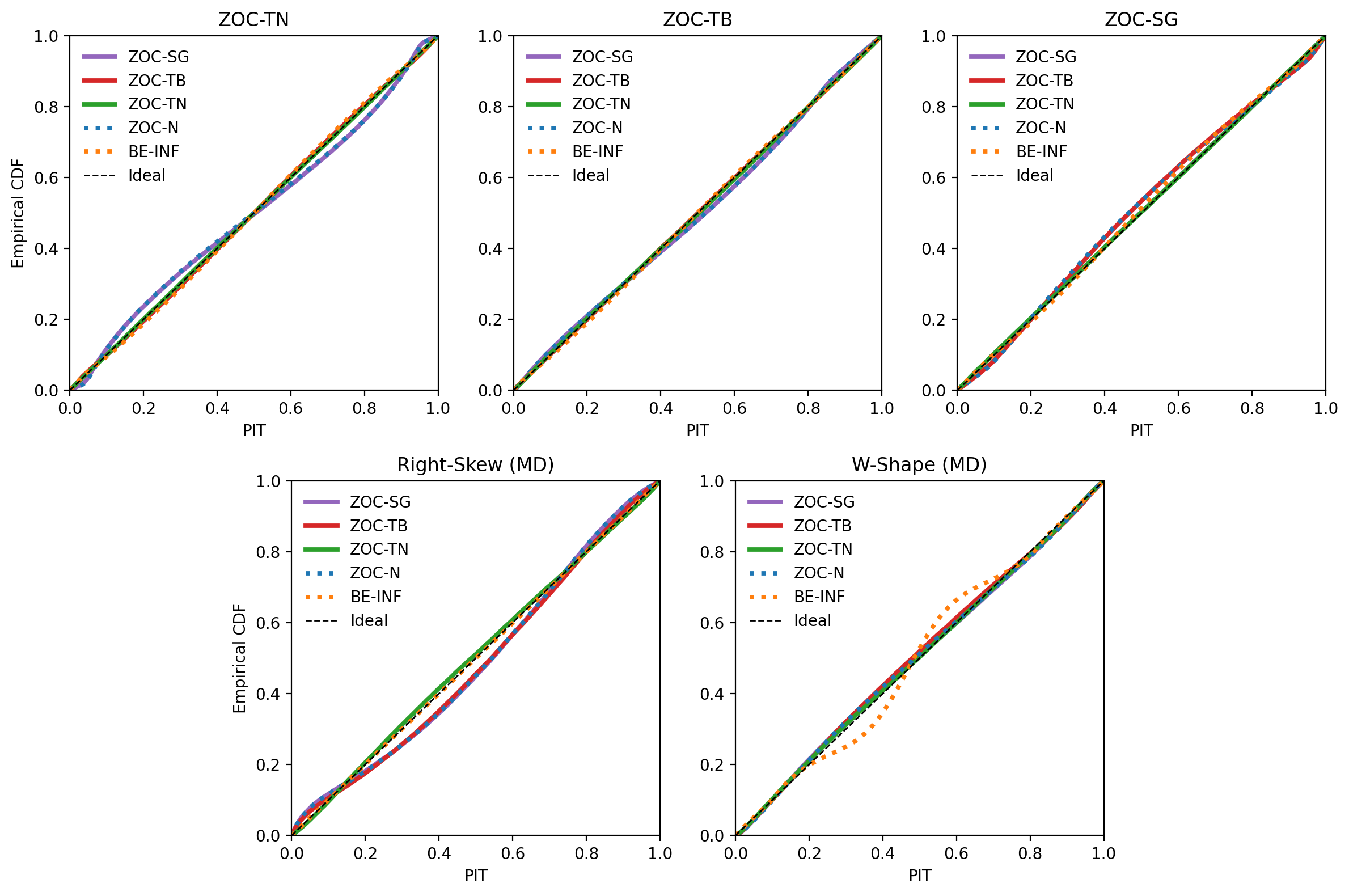}
  \caption{Randomized PIT diagrams for model test-set predictions, divided by DGP, pooled across all $100$ simulations. Curves show $\text{Uniform}(0,1)$ QQ-plots for empirical PIT CDFs of all independent linear models evaluated except Bernoulli QMLE (Quasi-B). Random uniform samples are drawn for probability masses at $0$ and $1$. The dashed black diagonal indicates perfect calibration.}
  \label{fig:simstudy_pitvanilla}
\end{figure}

\clearpage
\section{SFLLD LGD Exploratory Data Analysis}\label{app:data_exploration}

\begin{table}[ht!]
\centering
\scriptsize
\renewcommand{\arraystretch}{1.15}
\setlength{\tabcolsep}{6pt}
\caption{Description of predictor variables used in SFLLD LGD forecasting.}
\label{tab:predictors}
\begin{tabularx}{\linewidth}{@{} l >{\raggedright\arraybackslash}X c @{}}
\toprule
\textbf{Variable} & \textbf{Description} & \textbf{Data type} \\
\midrule
\texttt{credit\_score} &
Borrower credit score. &
Integer \\
\texttt{X\_centroid} &
Centroid longitude of property ZIP3 code. &
Decimal \\
\texttt{Y\_centroid} &
Centroid latitude of the property ZIP3 code. &
Decimal \\
\texttt{occupancy} &
Occupancy status: owner-occupied (P), second home (S), or investment property (I). &
Categorical \\
\texttt{nr\_units} &
Number of units in the property (1-4). &
Integer \\
\texttt{loan\_purpose} &
Loan purpose: cash-out refinance (C), no cash-out refinance (N), or purchase (P). &
Categorical \\
\texttt{first\_time\_homebuyer} &
Indicator for a first-time homebuyer (no residential property ownership in the prior three years): (Y/N). &
Categorical \\
\texttt{msa} &
Indicator for whether the property is located in a metropolitan statistical area (MSA). &
Categorical \\
\texttt{insurance\_percent} &
Percent of loss coverage provided by the insurer in the event of default. &
Percent \\
\texttt{original\_debt\_to\_income} &
Original debt-to-income ratio (monthly debt payments divided by monthly income). &
Decimal \\
\texttt{original\_ltv} &
Original combined loan-to-value (LTV): (first + secondary mortgage balances) divided by appraised value. &
Decimal \\
\texttt{original\_upb} &
Original unpaid principal balance (UPB) of the mortgage. &
Decimal \\
\texttt{ir\_spread} &
Interest-rate spread. &
Decimal \\
\texttt{multiple\_borrowers} &
Indicator for more than one obligated borrower on the loan. &
Categorical \\
\texttt{n\_months} &
Loan age in months. &
Integer \\
\texttt{ltv\_at\_default} &
Loan-to-value (LTV) of mortgage at time of default: current UPB divided by appraised value. &
Decimal \\
\texttt{gdp\_growth} &
Annual state-level GDP growth rate. &
Percent \\
\texttt{ln\_income\_per\_capita} &
Log per-capita personal income by state. &
Decimal \\
\texttt{ln\_expenditures\_per\_capita} &
Log per-capita personal consumption expenditures by state. &
Decimal \\
\texttt{unemployment\_rate} &
Annual state-level unemployment rate. &
Percent \\
\texttt{hpi\_growth} &
Annual state-level house price index (HPI) growth rate. &
Percent \\
\texttt{construction\_growth} &
Annual state-level GDP growth attributable to the construction industry. &
Percent \\
\texttt{gross\_operating\_surplus\_growth} &
Annual state-level growth in gross operating surplus (all industries). &
Percent \\
\texttt{inflation\_rate} &
Annual U.S. inflation rate. &
Percent \\
\texttt{djia\_growth} &
Annual growth rate of the year-end Dow Jones Industrial Average (DJIA). &
Percent \\
\bottomrule
\end{tabularx}
\end{table}

% TODO: ADD SUMMARY COUNTS FOR CATEGORICAL VARIABLES
\begin{table}[ht!]
\caption{Summary statistics for the numeric predictor variables.}
\label{tab:summary_stats}
\centering
\scriptsize
\begin{tabular}{lrrrrrr}
\toprule
 & Min. & Q.25\% & Median & Mean & Q.75\% & Max. \\
\midrule
X\_centroid & -123.758 & -96.044 & -85.557 & -91.015 & -81.387 & -67.866 \\
Y\_centroid & 25.531 & 34.570 & 39.217 & 38.123 & 41.833 & 48.428 \\
credit\_score & 300 & 652 & 689 & 690.8 & 731 & 850 \\
nr\_units & 1 & 1 & 1 & 1.04 & 1 & 4 \\
insurance\_percent & 0.000 & 0.000 & 0.000 & 7.906 & 22.000 & 55.000 \\
original\_debt\_to\_income & 1.000 & 32.000 & 38.000 & 38.294 & 45.000 & 65.000 \\
original\_loan\_to\_value & 6.000 & 75.000 & 80.000 & 81.301 & 90.000 & 949.000 \\
original\_upb & 6000 & 91000 & 142000 & 163546.2 & 217000 & 1233000 \\
ir\_spread & -4.580 & 0.845 & 1.475 & 1.473 & 2.125 & 6.960 \\
n\_months & 0.0 & 31.0 & 50.0 & 56.922 & 75.0 & 278.0 \\
ltv\_at\_default & 0.000 & 67.809 & 75.373 & 74.903 & 83.814 & 727.078 \\
gdp\_growth & -15.019 & 0.666 & 3.070 & 2.326 & 4.306 & 25.184 \\
ln\_income\_per\_capita & 9.937 & 10.480 & 10.576 & 10.595 & 10.689 & 11.415 \\
ln\_expenditures\_per\_capita & 9.723 & 10.306 & 10.412 & 10.402 & 10.486 & 11.188 \\
unemployment\_rate & 2.1 & 5.4 & 7.2 & 7.530 & 9.7 & 13.5 \\
hpi\_growth & -23.07 & -6.35 & -2.45 & -2.807 & 2.22 & 31.71 \\
gdp\_construction\_growth & -33.149 & -10.848 & -2.374 & -3.159 & 5.351 & 42.042 \\
gross\_operating\_surplus\_growth & -22.903 & 0.614 & 3.482 & 3.182 & 5.779 & 34.548 \\
inflation\_rate & -0.356 & 1.465 & 2.069 & 2.044 & 3.157 & 4.698 \\
DJIA\_growth & -33.84 & -0.61 & 7.26 & 3.603 & 18.82 & 26.50 \\
\bottomrule
\end{tabular}
\end{table}

\begin{table}[ht!]
\centering
\caption{Summary statistics for the categorical predictor variables.}
\label{tab:categorical_summary}
\scriptsize
\begin{tabular}{lll}
\toprule
 & Level & Count \\
\midrule
\texttt{occupancy} & I & 30155 \\
 & P & 326802 \\
 & S & 9902 \\
\texttt{loan\_purpose} & C & 136180 \\
 & N & 122975 \\
 & P & 107704 \\
\texttt{first\_time\_homebuyer} & 0 & 335117 \\
 & 1 & 31742 \\
msa & 0 & 78585 \\
 & 1 & 288274 \\
\texttt{multiple\_borrowers} & 1 & 218866 \\
 & 2 & 147993 \\
\bottomrule
\end{tabular}
\end{table}

\begin{figure}[ht!]
  \centering
  \includegraphics[width=0.6\linewidth]{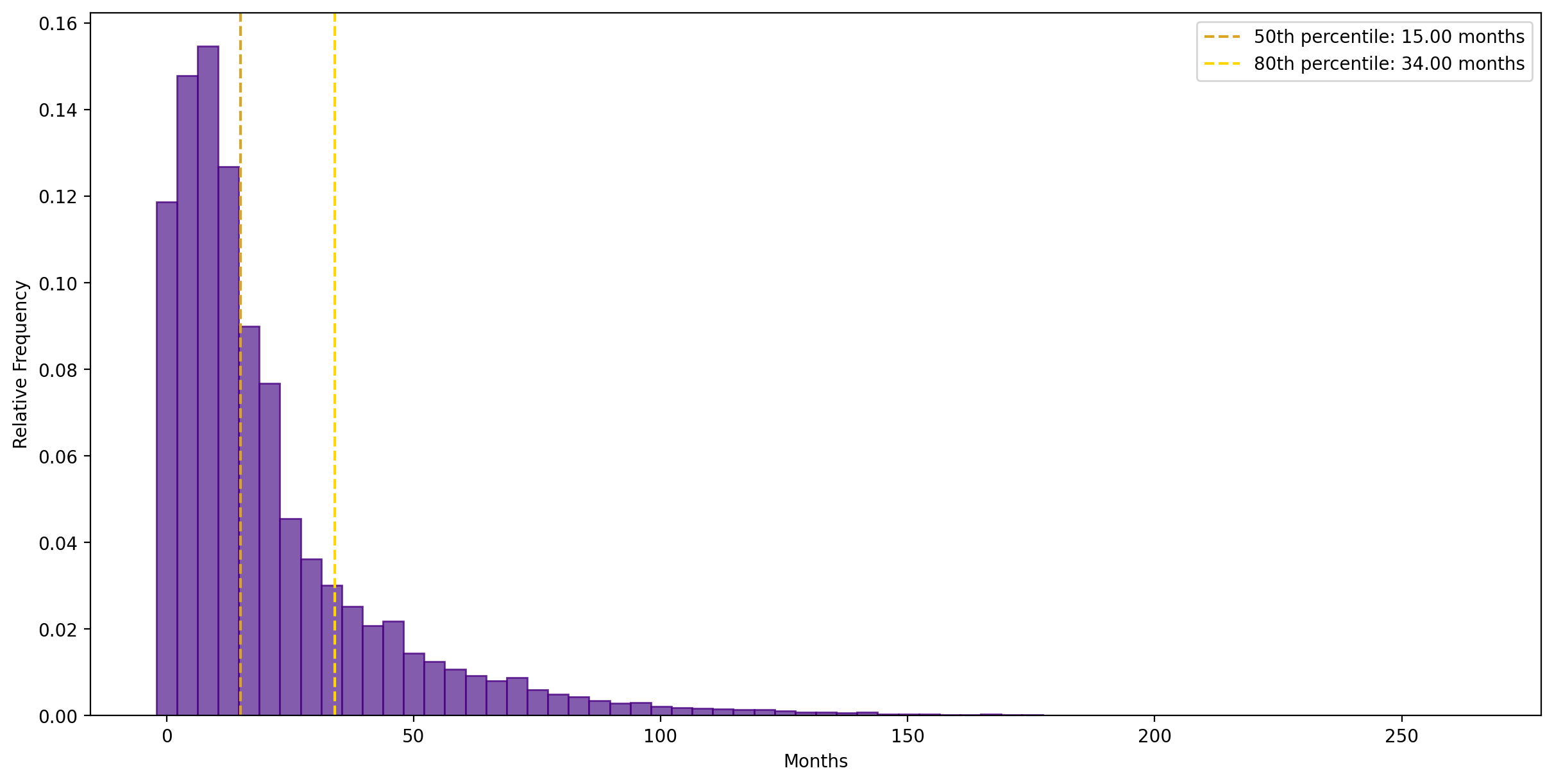}
  \caption{Empirical distribution of liquidation time in SFLLD: Time elapsed in months between default occurrence and final settlement.}
  \label{fig:lgd_liquidationdist}
\end{figure}

\begin{figure}[ht!]
  \centering
  \includegraphics[width=1\linewidth]{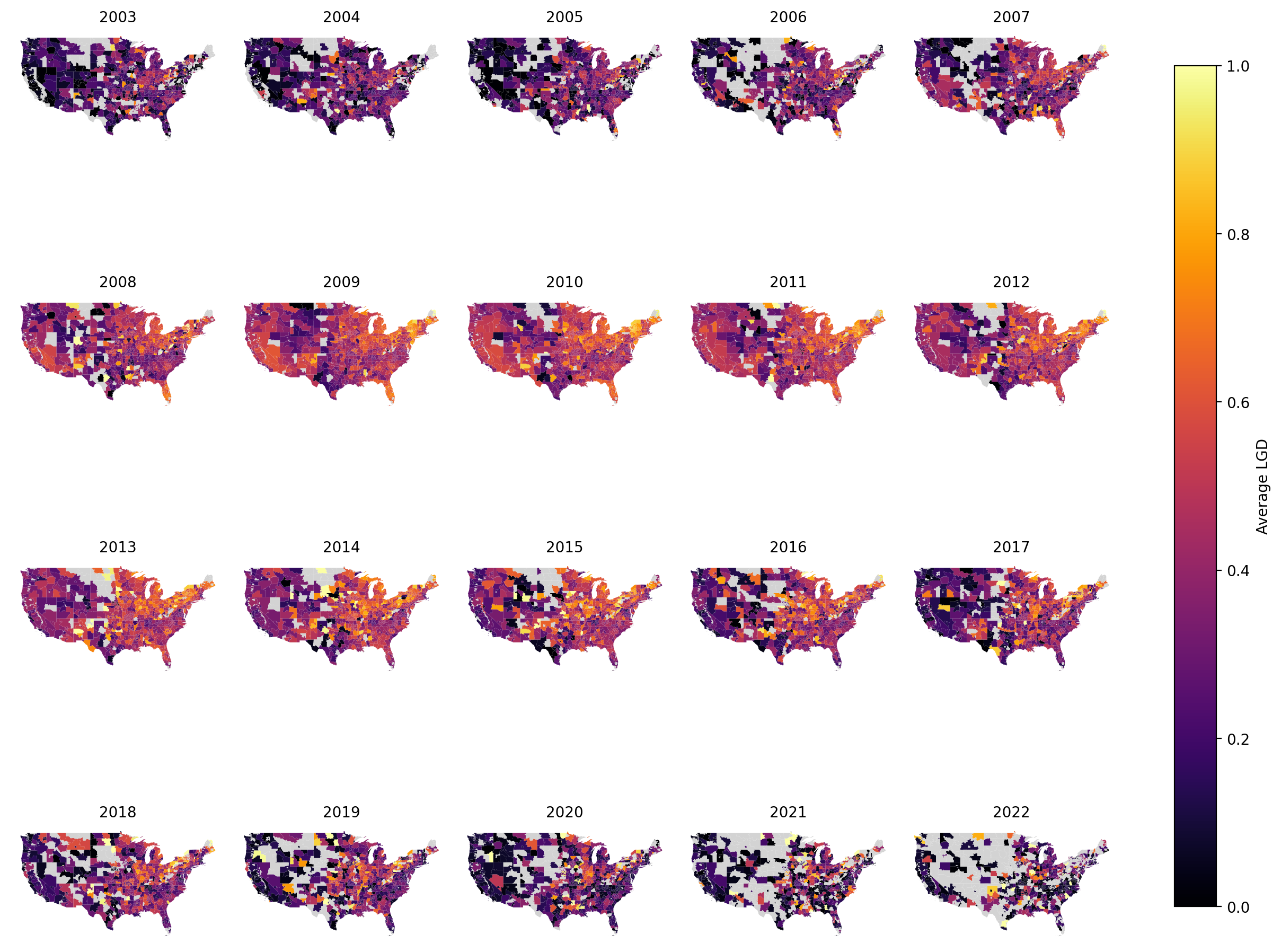}
  \caption{Heatmap of average LGD by ZIP3 locality across 2003-2022. Data for 2000-2002 not shown due to space constraints. No data available for gray regions.}
  \label{fig:lgd_spatiotemporal_heatmap}
\end{figure}

\clearpage
\section{GP Covariance Smoothness Sensitivity Checks}\label{app:gp_cov_smoothness}
In addition to the reported GP results in Section~\ref{subsec:stage2_results} with fixed $\nu=1.5$, we present here results for spatial models with smoothness estimated, and spatio-temporal models at $\nu = 1.5 \pm 1$. These results are shown in Table~\ref{tab:smoothness_results}. Under the spatial linear model, average fitted smoothness was found to be $\bar{\nu} = 0.448$, while under the spatial tree-boosting model, $\bar{\nu} = 0.132$. Despite this, in the case of the linear models, it is clear that variable smoothness provides at best marginal improvements to performance across MSE, log score and CRPS.

\begin{table}[ht!]
\centering
\caption{Out-of-sample forecasting results on LGD dataset for GP models with varying smoothness parameter.}
\label{tab:smoothness_results}
\footnotesize
\begin{tabular}{lrrr}
\toprule
& MSE & Log Score & CRPS \\
\midrule
\textit{GP Linear Models} \\
\midrule
\rowcolor{gray!15}
Spatial, $\nu=1.5$ & $0.0721$ & $-0.0816$ & $0.154$ \\
Spatial, $\nu$ estimated & $0.0718$ & $-0.0845$ & $0.1539$ \\

\rowcolor{gray!15}
Spatio-temporal, $\nu=1.5$ & $0.0689$ & $-0.0979$ & $0.151$ \\
Spatio-temporal, $\nu=0.5$ & $0.0964$ & $0.220$ & $0.188$ \\
Spatio-temporal, $\nu=2.5$ & $0.0631$ & $-0.143$ & $0.143$ \\
\midrule
\textit{GP Tree-Boosting} \\
\midrule
\rowcolor{gray!15}
Spatial, $\nu=1.5$ & $0.0722$ & $-0.0919$ & $0.153$ \\
Spatial, $\nu$ estimated & $0.0858$ & $0.191$ & $0.173$ \\
\rowcolor{gray!15}
Spatio-temporal, $\nu=1.5$ & $0.0572$ & $-0.210$ & $0.134$ \\
Spatio-temporal, $\nu=0.5$ & $0.0854$ & $0.0258$ & $0.172$ \\
Spatio-temporal, $\nu=2.5$ & $0.0590$ & $-0.192$ & $0.137$ \\
\bottomrule
\end{tabular}
\end{table}

\clearpage
\section{SFLLD LGD Independent Linear Model Details}\label{app:lgd_ilms}

% \subsection{Linear Model Coefficients}
% Tables~\ref{tab:lm_coeffs}, \ref{tab:lm_coeffs_continued1}, and \ref{tab:lm_coeffs_continued2} contain the estimated coefficients, and their corresponding standard errors, of the independent linear models evaluated in Table~\ref{tab:stage1_model_performance}. Models were fit on the entire LGD dataset.

\begin{figure}[ht!]
  \centering
  \includegraphics[width=0.7\linewidth]{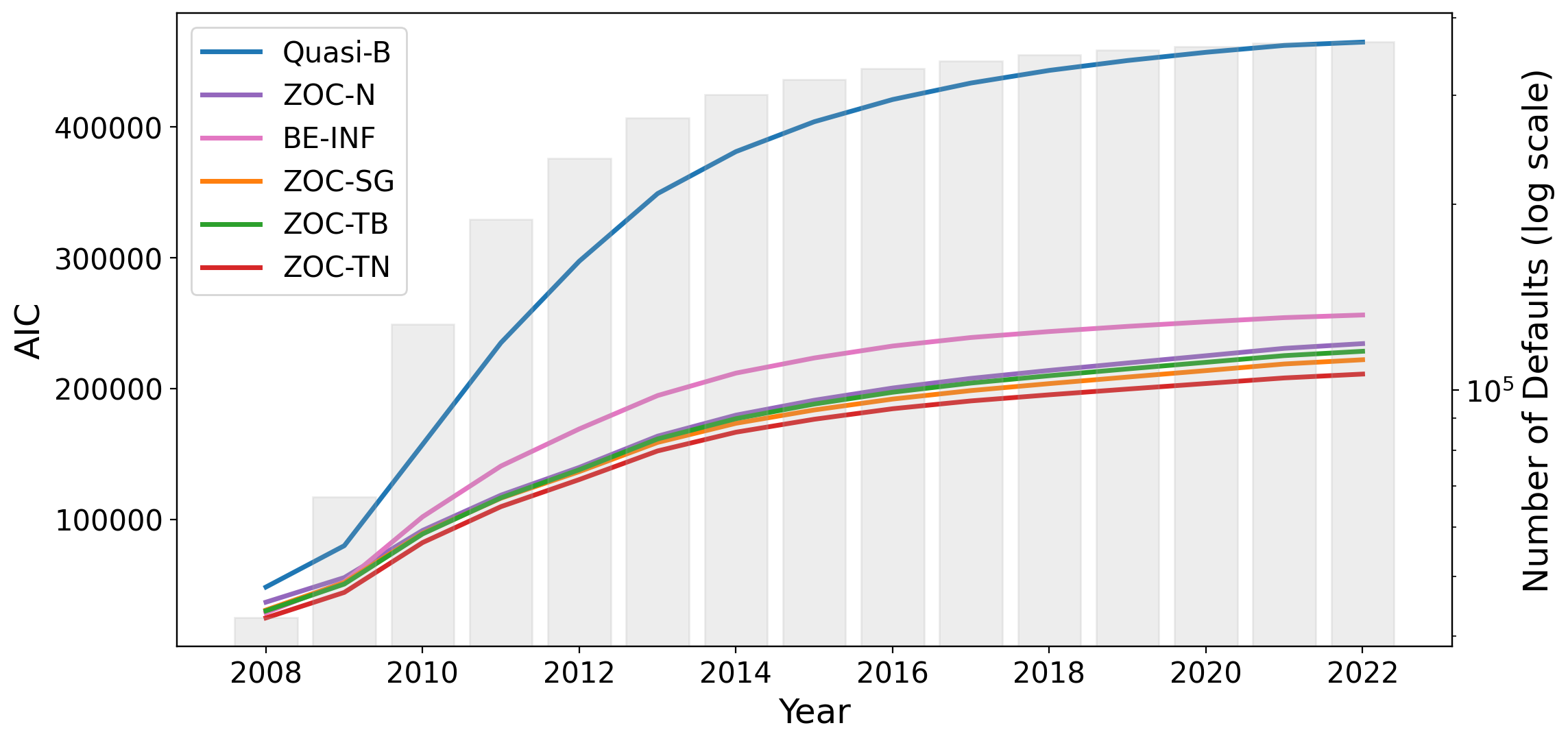}
  \caption{Evolution of linear independent model AIC scores across expanding window datasets.}
  \label{fig:stage1_aic}
\end{figure}

\begin{table}[ht!]
\centering
\caption{Estimated parameters for independent linear Quasi-B and BE-INF models, fitted on the full LGD dataset}
\label{tab:lm_coeffs}

\begin{minipage}[t]{0.48\textwidth}
\centering
\scriptsize
\subcaption{Quasi-B Model}
\begin{tabular}{lrr}
\toprule
 & Parameter & SE \\
\midrule
intercept & -0.0781 & 0.00351 \\
X\_centroid & 0.252 & 0.00402 \\
Y\_centroid & 0.102 & 0.00386 \\
credit\_score & -0.0806 & 0.00368 \\
nr\_units & 0.0767 & 0.00397 \\
insurance\_percent & -0.326 & 0.00457 \\
original\_debt\_to\_income & -0.0108 & 0.00361 \\
original\_loan\_to\_value & 0.0501 & 0.00708 \\
original\_upb & -0.363 & 0.00453 \\
ir\_spread & 0.134 & 0.00439 \\
n\_months & -0.0566 & 0.00474 \\
ltv\_at\_default & 0.176 & 0.00689 \\
gdp\_growth & -0.135 & 0.0126 \\
ln\_income\_per\_capita & -0.0597 & 0.0126 \\
ln\_expenditures\_per\_capita & 0.151 & 0.0125 \\
unemployment\_rate & 0.0793 & 0.00536 \\
hpi\_growth & -0.321 & 0.00714 \\
gdp\_construction\_growth & -0.0438 & 0.00714 \\
gross\_operating\_surplus\_growth & 0.0628 & 0.00902 \\
inflation\_rate & 0.0986 & 0.00718 \\
DJIA\_growth & -0.0294 & 0.00594 \\
occupancy\_P & -0.179 & 0.00451 \\
occupancy\_S & -0.0552 & 0.00413 \\
loan\_purpose\_N & -0.132 & 0.00415 \\
loan\_purpose\_P & -0.182 & 0.00476 \\
first\_time\_homebuyer\_1 & 0.0163 & 0.00409 \\
MSA\_1 & -0.0274 & 0.0036 \\
number\_of\_borrowers\_2 & -0.0671 & 0.00363 \\
\bottomrule
\end{tabular}
\end{minipage}
\hfill
\begin{minipage}[t]{0.48\textwidth}
\centering
\scriptsize
\subcaption{BE-INF Model}
\begin{tabular}{lrr}
\toprule
 & Parameter & SE \\
\midrule
intercept & -0.217 & 0.00181 \\
X\_centroid & 0.19 & 0.002 \\
Y\_centroid & 0.0638 & 0.00195 \\
credit\_score & -0.0484 & 0.00189 \\
nr\_units & 0.0463 & 0.00206 \\
insurance\_percent & -0.249 & 0.00234 \\
original\_debt\_to\_income & 0.00619 & 0.00187 \\
original\_loan\_to\_value & 0.0668 & 0.0035 \\
original\_upb & -0.226 & 0.0023 \\
ir\_spread & 0.1 & 0.00225 \\
n\_months & -0.069 & 0.00245 \\
ltv\_at\_default & 0.146 & 0.0033 \\
gdp\_growth & -0.0523 & 0.00658 \\
ln\_income\_per\_capita & -0.0927 & 0.00645 \\
ln\_expenditures\_per\_capita & 0.152 & 0.0064 \\
unemployment\_rate & 0.0735 & 0.00273 \\
hpi\_growth & -0.272 & 0.00359 \\
gdp\_construction\_growth & -0.0405 & 0.00364 \\
gross\_operating\_surplus\_growth & 0.0124 & 0.00476 \\
inflation\_rate & 0.0565 & 0.00368 \\
DJIA\_growth & -0.0231 & 0.00305 \\
occupancy\_P & -0.11 & 0.0024 \\
occupancy\_S & -0.0298 & 0.00215 \\
loan\_purpose\_N & -0.107 & 0.00215 \\
loan\_purpose\_P & -0.145 & 0.00244 \\
first\_time\_homebuyer\_1 & 0.00981 & 0.00208 \\
MSA\_1 & -0.0356 & 0.00185 \\
number\_of\_borrowers\_2 & -0.0387 & 0.00185 \\
alpha & 0.538 & 0.00122 \\
gamma & 0.897 & 0.000908 \\
phi & 2.74 & 0.0061 \\
\bottomrule
\end{tabular}
\end{minipage}

% \vspace{0.5cm}

\end{table}

\begin{table}[ht!]
\centering
\caption{Estimated parameters for independent linear ZOC-N and ZOC-TN models, fitted on the full LGD dataset}
\label{tab:lm_coeffs_continued1}

\begin{minipage}[t]{0.48\textwidth}
\centering
\scriptsize
\subcaption{ZOC-N Model}
\begin{tabular}{lrr}
\toprule
 & Parameter & SE \\
\midrule
intercept & 0.491 & 0.000496 \\
X\_centroid & 0.0627 & 0.000558 \\
Y\_centroid & 0.0266 & 0.00054 \\
credit\_score & -0.021 & 0.000522 \\
nr\_units & 0.0184 & 0.00053 \\
insurance\_percent & -0.0808 & 0.000632 \\
original\_debt\_to\_income & -0.00366 & 0.000512 \\
original\_loan\_to\_value & 0.0114 & 0.000912 \\
original\_upb & -0.0907 & 0.000623 \\
ir\_spread & 0.0345 & 0.00062 \\
n\_months & -0.015 & 0.000665 \\
ltv\_at\_default & 0.0409 & 0.000881 \\
gdp\_growth & -0.038 & 0.0018 \\
ln\_income\_per\_capita & -0.0156 & 0.00178 \\
ln\_expenditures\_per\_capita & 0.0418 & 0.00177 \\
unemployment\_rate & 0.023 & 0.000753 \\
hpi\_growth & -0.0736 & 0.000984 \\
gdp\_construction\_growth & -0.0104 & 0.00101 \\
gross\_operating\_surplus\_growth & 0.0181 & 0.00129 \\
inflation\_rate & 0.0263 & 0.00101 \\
DJIA\_growth & -0.00818 & 0.00084 \\
occupancy\_P & -0.0486 & 0.000627 \\
occupancy\_S & -0.0158 & 0.000583 \\
loan\_purpose\_N & -0.0335 & 0.000591 \\
loan\_purpose\_P & -0.0448 & 0.000672 \\
first\_time\_homebuyer\_1 & 0.00386 & 0.000575 \\
MSA\_1 & -0.00598 & 0.000511 \\
number\_of\_borrowers\_2 & -0.0177 & 0.000514 \\
sigma & 0.296 & 0.000383 \\
\bottomrule
\end{tabular}
\end{minipage}
\hfill
\begin{minipage}[t]{0.48\textwidth}
\centering
\scriptsize
\subcaption{ZOC-TN Model}
\begin{tabular}{lrr}
\toprule
 & Parameter & SE \\
\midrule
intercept & 0.544 & 0.000682 \\
X\_centroid & 0.059 & 0.000524 \\
Y\_centroid & 0.0255 & 0.000506 \\
credit\_score & -0.0187 & 0.000489 \\
nr\_units & 0.0162 & 0.000497 \\
insurance\_percent & -0.0769 & 0.000594 \\
original\_debt\_to\_income & -0.00266 & 0.000479 \\
original\_loan\_to\_value & 0.0123 & 0.000854 \\
original\_upb & -0.0803 & 0.000588 \\
ir\_spread & 0.0329 & 0.000581 \\
n\_months & -0.0137 & 0.000623 \\
ltv\_at\_default & 0.0399 & 0.000824 \\
gdp\_growth & -0.0307 & 0.00168 \\
ln\_income\_per\_capita & -0.0165 & 0.00167 \\
ln\_expenditures\_per\_capita & 0.0415 & 0.00166 \\
unemployment\_rate & 0.0234 & 0.000706 \\
hpi\_growth & -0.0704 & 0.000923 \\
gdp\_construction\_growth & -0.011 & 0.000949 \\
gross\_operating\_surplus\_growth & 0.0133 & 0.00121 \\
inflation\_rate & 0.0235 & 0.000949 \\
DJIA\_growth & -0.00728 & 0.000787 \\
occupancy\_P & -0.0441 & 0.000589 \\
occupancy\_S & -0.014 & 0.000546 \\
loan\_purpose\_N & -0.0311 & 0.000555 \\
loan\_purpose\_P & -0.042 & 0.00063 \\
first\_time\_homebuyer\_1 & 0.00349 & 0.000538 \\
MSA\_1 & -0.00634 & 0.000478 \\
number\_of\_borrowers\_2 & -0.0158 & 0.000482 \\
sigma & 0.278 & 0.000393 \\
a & -0.369 & 0.00361 \\
b & 1.21 & 0.00209 \\
\bottomrule
\end{tabular}
\end{minipage}

\end{table}

\begin{table}[ht!]
\centering
\caption{Estimated parameters for independent linear ZOC-SG and ZOC-TB models, fitted on the full LGD dataset}
\label{tab:lm_coeffs_continued2}

\begin{minipage}[t]{0.48\textwidth}
\centering
\scriptsize
\subcaption{ZOC-SG Model}
\begin{tabular}{lrr}
\toprule
 & Parameter & Std (Hessian) \\
\midrule
intercept & 0.155 & 0.008 \\
X\_centroid & 0.0554 & 0.000664 \\
Y\_centroid & 0.0229 & 0.000506 \\
credit\_score & -0.0188 & 0.000479 \\
nr\_units & 0.0165 & 0.000463 \\
insurance\_percent & -0.0729 & 0.00082 \\
original\_debt\_to\_income & -0.00305 & 0.000438 \\
original\_loan\_to\_value & 0.00905 & 0.000775 \\
original\_upb & -0.0781 & 0.000797 \\
ir\_spread & 0.03 & 0.000586 \\
n\_months & -0.0106 & 0.000548 \\
ltv\_at\_default & 0.0376 & 0.000851 \\
gdp\_growth & -0.0329 & 0.00147 \\
ln\_income\_per\_capita & -0.0123 & 0.00152 \\
ln\_expenditures\_per\_capita & 0.0367 & 0.00154 \\
unemployment\_rate & 0.0186 & 0.00069 \\
hpi\_growth & -0.0647 & 0.000998 \\
gdp\_construction\_growth & -0.0111 & 0.000916 \\
gross\_operating\_surplus\_growth & 0.0157 & 0.00101 \\
inflation\_rate & 0.0252 & 0.000865 \\
DJIA\_growth & -0.00468 & 0.000696 \\
occupancy\_P & -0.0407 & 0.000644 \\
occupancy\_S & -0.0124 & 0.000529 \\
loan\_purpose\_N & -0.0292 & 0.000574 \\
loan\_purpose\_P & -0.0392 & 0.000677 \\
first\_time\_homebuyer\_1 & 0.00357 & 0.000479 \\
MSA\_1 & -0.00621 & 0.000449 \\
number\_of\_borrowers\_2 & -0.016 & 0.000476 \\
k & 14.8 & 0.25 \\
xi & 0.681 & 0.00939 \\
\bottomrule
\end{tabular}
\end{minipage}
\hfill
\begin{minipage}[t]{0.48\textwidth}
\centering
\scriptsize
\subcaption{ZOC-TB Model}
\begin{tabular}{lrr}
\toprule
 & Parameter & SE \\
\midrule
intercept & -0.0196 & 0.00141 \\
X\_centroid & 0.182 & 0.0018 \\
Y\_centroid & 0.0802 & 0.00159 \\
credit\_score & -0.0596 & 0.0015 \\
nr\_units & 0.0557 & 0.00164 \\
insurance\_percent & -0.229 & 0.00208 \\
original\_debt\_to\_income & -0.0109 & 0.00145 \\
original\_loan\_to\_value & 0.0419 & 0.00279 \\
original\_upb & -0.258 & 0.00212 \\
ir\_spread & 0.0991 & 0.00182 \\
n\_months & -0.0386 & 0.0019 \\
ltv\_at\_default & 0.117 & 0.00267 \\
gdp\_growth & -0.104 & 0.00513 \\
ln\_income\_per\_capita & -0.0366 & 0.00508 \\
ln\_expenditures\_per\_capita & 0.111 & 0.00507 \\
unemployment\_rate & 0.0644 & 0.00217 \\
hpi\_growth & -0.211 & 0.00299 \\
gdp\_construction\_growth & -0.0279 & 0.00287 \\
gross\_operating\_surplus\_growth & 0.049 & 0.00366 \\
inflation\_rate & 0.0714 & 0.00289 \\
DJIA\_growth & -0.0239 & 0.00238 \\
occupancy\_P & -0.139 & 0.00192 \\
occupancy\_S & -0.0456 & 0.00168 \\
loan\_purpose\_N & -0.0935 & 0.00172 \\
loan\_purpose\_P & -0.126 & 0.00198 \\
first\_time\_homebuyer\_1 & 0.0105 & 0.00163 \\
MSA\_1 & -0.016 & 0.00145 \\
number\_of\_borrowers\_2 & -0.0496 & 0.00148 \\
phi & 4.83 & 0.0539 \\
u & 0.193 & 0.00315 \\
\bottomrule
\end{tabular}
\end{minipage}

\end{table}

% \subsection{Fitting Runtimes}
% % TODO ADD TEXT
% \begin{figure}[ht!]
%   \centering
%   \includegraphics[width=0.6\linewidth]{Figures/Freddiemac_Application/lm_times.png}
%   \caption{Independent linear model fitting times over expanding LGD dataset window.}
%   \label{fig:lm_speeds}
% \end{figure}

\clearpage
\subsection{Test for comparing the two-tiered Tobit model to the ZOC-TN model }\label{app_tests_tobit}
From the results seen so far, the transformation parameters of the ZOC-TN model seem to provide two valuable additional degrees of freedom to match the interior of the empirical LGD distribution when compared to the vanilla two-tiered Tobit model. We verify this empirically by performing both log-likelihood ratio and Wald tests to evaluate the significance of the transformation parameters, $a$ and $b$. This is done by noting that the ZOC-N likelihood is a nested model of the ZOC-TN likelihood, recovered when $(a, b) =(0, 1)$. Using the negative log-likelihood and estimated parameters and covariance matrix of the fitted ZOC-TN model found in Table~\ref{tab:stage1_model_performance} and Appendix~\ref{app:lgd_ilms}, we compute the test statistics for the two tests as
% TODO: CHANGE NOTATION
\begin{align*}
\begin{aligned}
    L &= 2 (\ell_{\mathrm{ZOC-TN}} - \ell_{\mathrm{ZOC-N}}), &
    W &= (\hat{\mathbf{a}}_{\mathrm{ZOC-TN}} - \mathbf{a}_{\mathrm{ZOC-N}})^\top
    \widehat{\mathrm{Var}}(\hat{\mathbf{a}}_{\mathrm{ZOC-TN}})^{-1}
    (\hat{\mathbf{a}}_{\mathrm{ZOC-TN}} - \mathbf{a}_{\mathrm{ZOC-N}}) \\
      &\approx 23,290, &
      &\approx 19,988
\end{aligned}
\end{align*}
where $\hat{\mathbf{a}}_{\mathrm{ZOC-TN}} = (\hat{a}, \hat{b})$, and $\mathbf{a}_{\mathrm{ZOC-N}} = (0, 1)$. Under the two tests, both test statistics are assumed to follow a $\chi^2_2$ distribution, producing p-values $\ll 0.01$ and providing strong evidence that the transformation meaningfully improves fit.
%TODO: ADD NOTE ON RELATIVE SPEED
% TODO: Merge with Goodness-of-fit table
% \begin{table}
% \centering
% \caption{Runtimes associated with fitting independent linear models on full LGD dataset. Standard error (S.E.) calculations return information on uncertainty of estimate parameters, based on Hessian of negative log-likelihood. }
% \label{tab:stage1_times}
% \footnotesize
% \begin{tabular}{lrr}
% \toprule
%   & No S.E. Calculation & With S.E. Calculation \\
% \midrule
% Quasi-B & $0.51$s & $0.898$s \\
% ZOC-N & $1.24$s & $1.84$s \\
% BE-INF & $1.29$s & $1.85$s \\
% ZOC-SG & $1$m $0.8$s &	$1$m $3.4$s  \\
% ZOC-TB & $2$m $35$s & $2$m $40$s \\
% ZOC-TN & $2.06$s & $55.9$s \\
% \bottomrule
% \end{tabular}
% \end{table}

\clearpage
\section{Additional Results for LGD Prediction}\label{app:lgd_annual_performance}

Figures \ref{fig:lgd_annual_performance_glm} and \ref{fig:lgd_annual_performance_tb} show annual forecasting accuracy of the independent linear and tree-boosting models, respectively, on the Freddie Mac LGD dataset. 

\begin{figure}[ht!]
  \centering
  \includegraphics[width=0.8\linewidth]{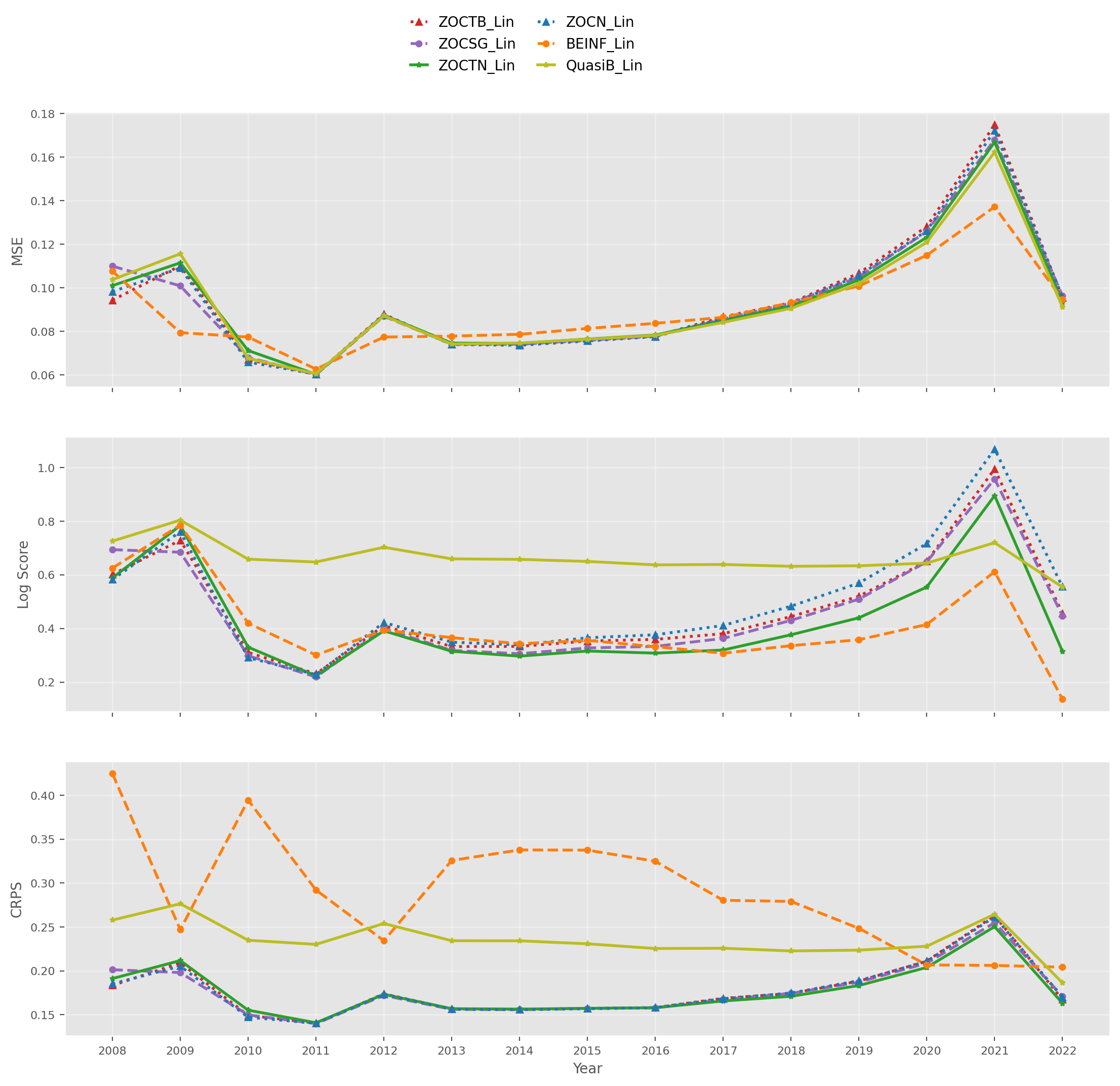}
  \caption{One-year-ahead prediction performance of independent linear models across all folds, 2008-2022.}
  \label{fig:lgd_annual_performance_glm}
\end{figure}

\begin{figure}[ht!]
  \centering
  \includegraphics[width=0.8\linewidth]{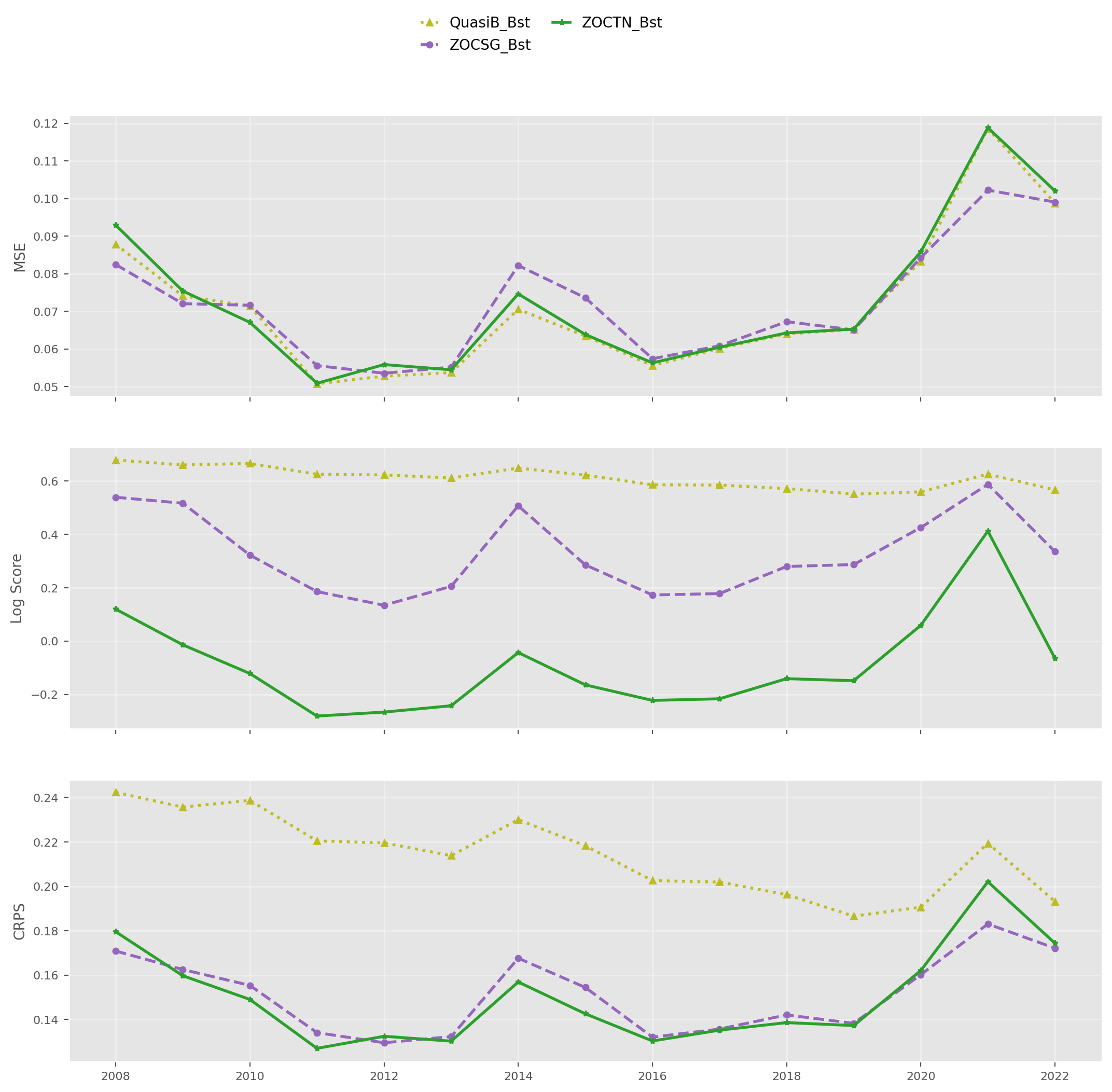}
  \caption{One-year-ahead prediction performance of independent tree-boosting models across all folds, 2008-2022.}
  \label{fig:lgd_annual_performance_tb}
\end{figure}

\begin{table}[ht!]
\centering
\caption{Average portfolio LGD distribution prediction accuracy metrics. }
\label{tab:economic_metrics}
\footnotesize
\begin{tabular}{lrr}
\toprule
 & Mean Absolute Error & 99Q Loss \\
\midrule
\multicolumn{3}{l}{\textit{Independent Linear Models}} \\
\midrule
Quasi-B & $4.22 \times 10^8$ & $1.62 \times 10^9$\\
ZOC-N & $4.14 \times 10^8$ & $1.50 \times 10^9$\\
BE-INF & $7.73 \times 10^8$ & $1.00 \times 10^{10}$\\
ZOC-TN & $4.33 \times 10^8$ & $1.80 \times 10^9$\\
ZOC-SG & $4.11 \times 10^8$ & $1.72 \times 10^9$\\
ZOC-TB & $4.22 \times 10^8$ & $1.54 \times 10^9$\\
\midrule
\multicolumn{3}{l}{\textit{GP Linear Models}} \\
\midrule
Spatial (\texttt{ZOCTN\_S\_Lin}) & $3.51 \times 10^8$ & $1.04 \times 10^9$\\
Spatio-temporal (\texttt{ZOCTN\_ST\_Lin}) & $2.82 \times 10^8$ & $1.90 \times 10^9$\\
\midrule
\multicolumn{3}{l}{\textit{Independent Tree-Boosting}} \\
\midrule
Quasi-B & $2.01 \times 10^8$ & $1.88 \times 10^9$\\
ZOC-TN & $1.73 \times 10^8$ & $1.80 \times 10^9$\\
ZOC-SG & $2.12 \times 10^8$ & $2.16 \times 10^9$\\
\midrule
\multicolumn{3}{l}{\textit{GP Tree-Boosting}} \\
\midrule
Spatial (\texttt{ZOCTN\_S\_Bst}) & $3.59 \times 10^8$ & $4.23 \times 10^9$ \\
Spatio-temporal (\texttt{ZOCTN\_ST\_Bst}) & $1.34 \times 10^8$ & $8.90\times 10^8$ \\
\bottomrule
\end{tabular}
\end{table}

% \begin{figure}[ht!]
%   \centering
%   \includegraphics[width=0.8\linewidth]{Figures/Freddiemac_Application/econ_metrics_glm_bigger.png}
%   \caption{Predictive portfolio loss distribution metrics for independent linear models across all folds.}
%   \label{fig:econ_metrics_glms}
% \end{figure}

\clearpage
\section{Log Score and CRPS}\label{app:ls_crps}
The log score is computed as the mean negative log-likelihood of the test set under a model's predictive distribution. This is done in the straightforward way, as prescribed by Equation~\eqref{eq:ind_ll} in the case of the independent models, while we employ a Monte Carlo estimator in the case of GP models. This is because to calculate the likelihood under the GP models (Equation~\eqref{eq:gp_likelihood}), we must marginalize out the latent Gaussian process, $\mathcal{G}_i := \mathcal{G}(\mathbf{s}_i)$. This is done by approximating the posterior distribution of $\mathcal{G}_i$ to be normal, $\mathcal{G}_i\sim \mathcal{N}(\omega_i, \nu_i^2)$, and estimating
\begin{align*}
    p^{\mathrm{model}}(Y_i \mid \mathbf{x}_i, \mathcal{D}_{\mathrm{train}}) &= \int p(Y_i \mid \mathbf{x}_i, \mathcal{G}_i) p(\mathcal{G}_i \mid \mathcal{D}_{\mathrm{train}}) d\mathcal{G}_i \\
    &\approx \frac{1}{N_{\mathrm{MC}}} \sum_{j=1}^{N_{\mathrm{MC}}} p(Y_i \mid \mathbf{x}_i, \mathcal{G}_i^j), \; \mathcal{G}_i^j \sim \mathcal{N}(\omega_i, \nu_i^2)
\end{align*}

For CRPS, a sampling method is employed for both independent and GP models, targeting the identity
\begin{align*}
    \mathrm{CRPS}(F^{\mathrm{model}}_i, Y_i) &= \mathbb{E}_{Z \sim F}[|Z - Y_i|] - \frac{1}{2}\mathbb{E}_{Z, Z' \sim F}[|Z - Z'|] \\
    &\approx \frac{1}{N_{\mathrm{MC}}}\sum_{j=1}^{N_{\mathrm{MC}}} |Z_i^j - Y_i| - \frac{1}{2N_{\mathrm{MC}}^2}\sum_{j, k=1}^{N_{\mathrm{MC}}} |Z_i^j - Z_i^k|, \; Z_i^j, Z_i^k \sim F^{\mathrm{model}}_i
\end{align*}
with $F^{\mathrm{model}}_i$ being the predictive distribution of the $i^{\mathrm{th}}$ observation under the model, and $Y_i$ the corresponding test-set observation. For the independent linear models, samples are drawn through the appropriate generating process for the respective likelihood, while in the GP case a single sample from the posterior GP distribution is first drawn for each $Z_i^j, Z_i^k$ in order to simulate from $F^{\mathrm{model}}_i$.

\clearpage
\section{SFLLD LGD Hyperparameters}\label{app:results}

\begin{table}[ht!]
\centering
\caption{List of tuned tree-boosting hyperparameters.}
\label{tab:tuning_param_desc}
\scriptsize
\begin{tabularx}{\textwidth}{lXX}
\hline
\textbf{Parameter} & \textbf{Description} & \textbf{Grid values} \\
\hline
\texttt{learning\_rate} & Shrinkage applied to each boosting step & \{0.001, 0.01, 0.1, 1, 10\} \\
\texttt{min\_data\_in\_leaf} & Minimum observations allowed in a terminal leaf & \{1, 10, 100, 1000\} \\
\texttt{num\_leaves} & Maximum number of terminal leaves per tree & \{4, 8, 16, 32, 64, 128, 256, 512\} \\
\texttt{lambda\_l2} & $L_2$ regularization penalty on leaf weights & \{0, 0.1, 1, 10\} \\
\texttt{feature\_fraction} & Fraction of features randomly sampled for each tree & \{0.70, 0.80, 0.90\} \\
\texttt{line\_search\_step\_length} & Whether line search is used for the step length & \{\texttt{True}, \texttt{False}\} \\
\hline
\end{tabularx}
\end{table}

\begin{figure}[ht!]
  \centering
  \includegraphics[width=0.8\linewidth]{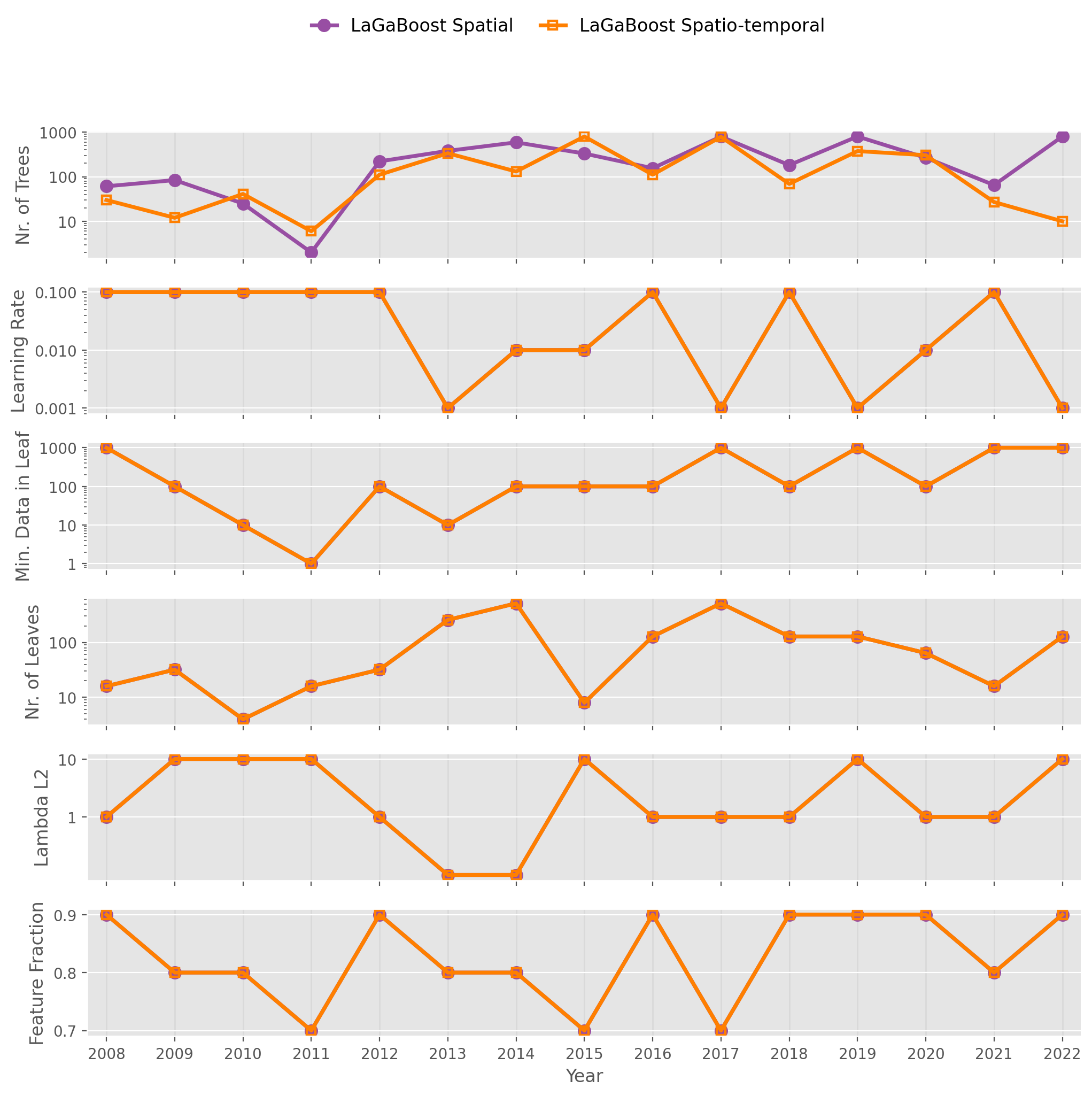}
  \caption{Tuned tree-boosting hyperparameters of spatial and spatio-temporal GP tree-boosting models. Configurations tuned via random grid search with $100$ samples. Due to computational constraints, only number of boosting iterations ($\texttt{N\_trees}$) is tuned independently for spatio-temporal models.}
  \label{fig:gpb_hps}
\end{figure}

\clearpage
\section{SFLLD LGD PIT Reliability Diagrams}\label{app:lgd_pit_qq}

\begin{figure}[ht!]
  \centering
  \includegraphics[width=1\linewidth]{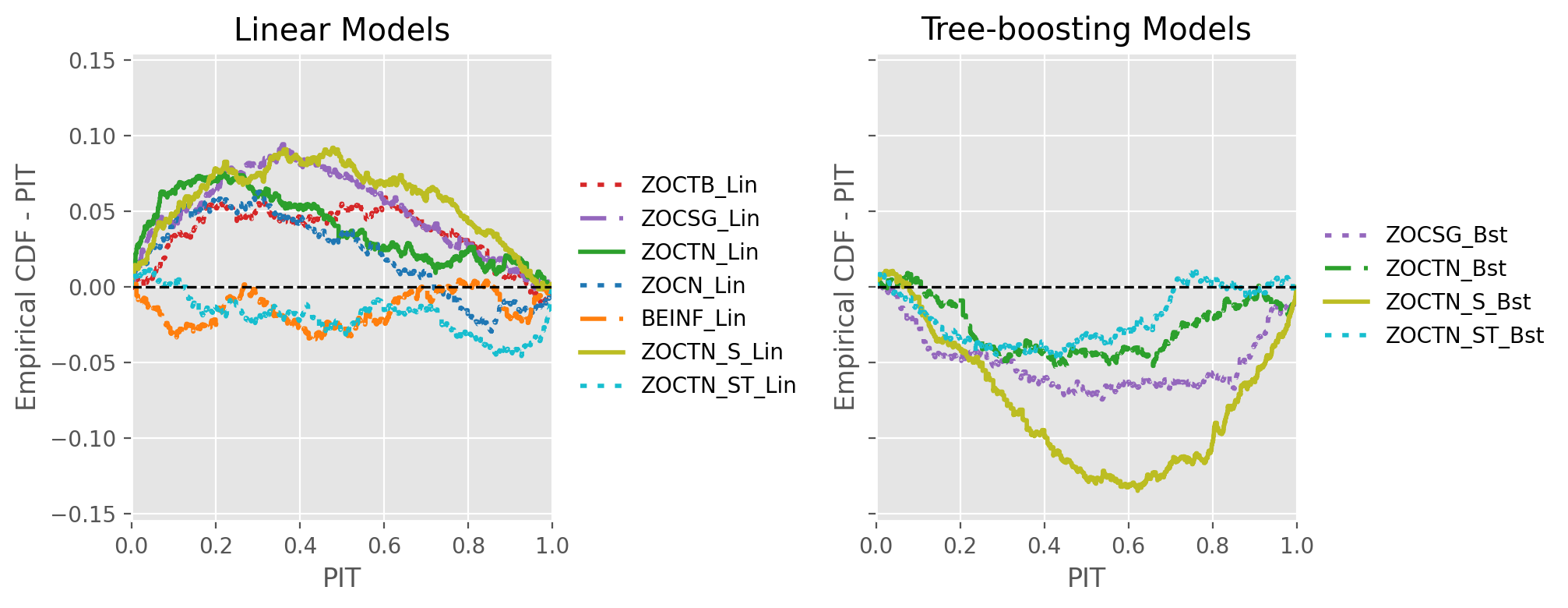}
  \caption{PIT reliability residual diagrams pooled across all test data sets.}
  \label{fig:stage2_pitreliability}
\end{figure}

Figure~\ref{fig:stage2_pitvanilla} presents the PIT reliability diagrams shown in Figure~\ref{fig:stage2_pitreliability} as $\text{Uniform}(0,1)$ QQ-plots.
\begin{figure}[ht!]
  \centering
  \includegraphics[width=0.8\linewidth]{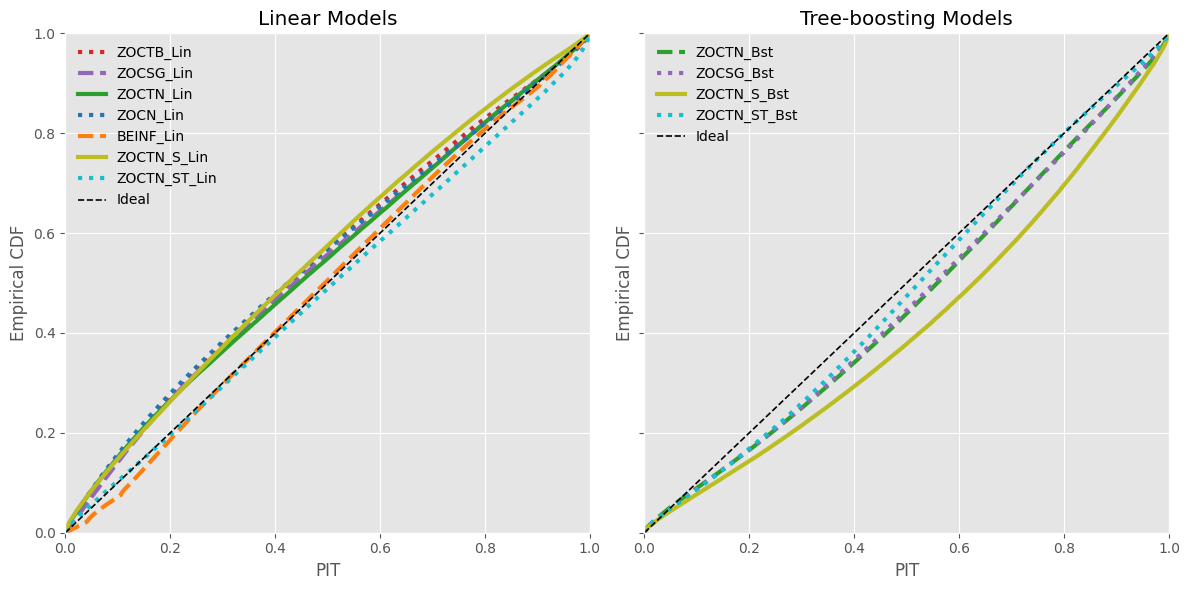}
  \caption{Randomized PIT diagrams of LGD forecasts of linear (left) and tree-boosting (right) models. Curves show $\text{Uniform}(0,1)$ QQ-plots of empirical PIT CDFs for one-year-ahead defaults, pooled across all evaluation folds, 2008-2022. Random uniform samples are drawn for probability masses at $0$ and $1$. The dashed black diagonal indicates perfect calibration.}
  \label{fig:stage2_pitvanilla}
\end{figure}

\clearpage
\section{SFLLD LGD Model Interpretation}\label{app:gp_interp}

\begin{figure}[ht!]
  \centering
  \includegraphics[width=1\linewidth]{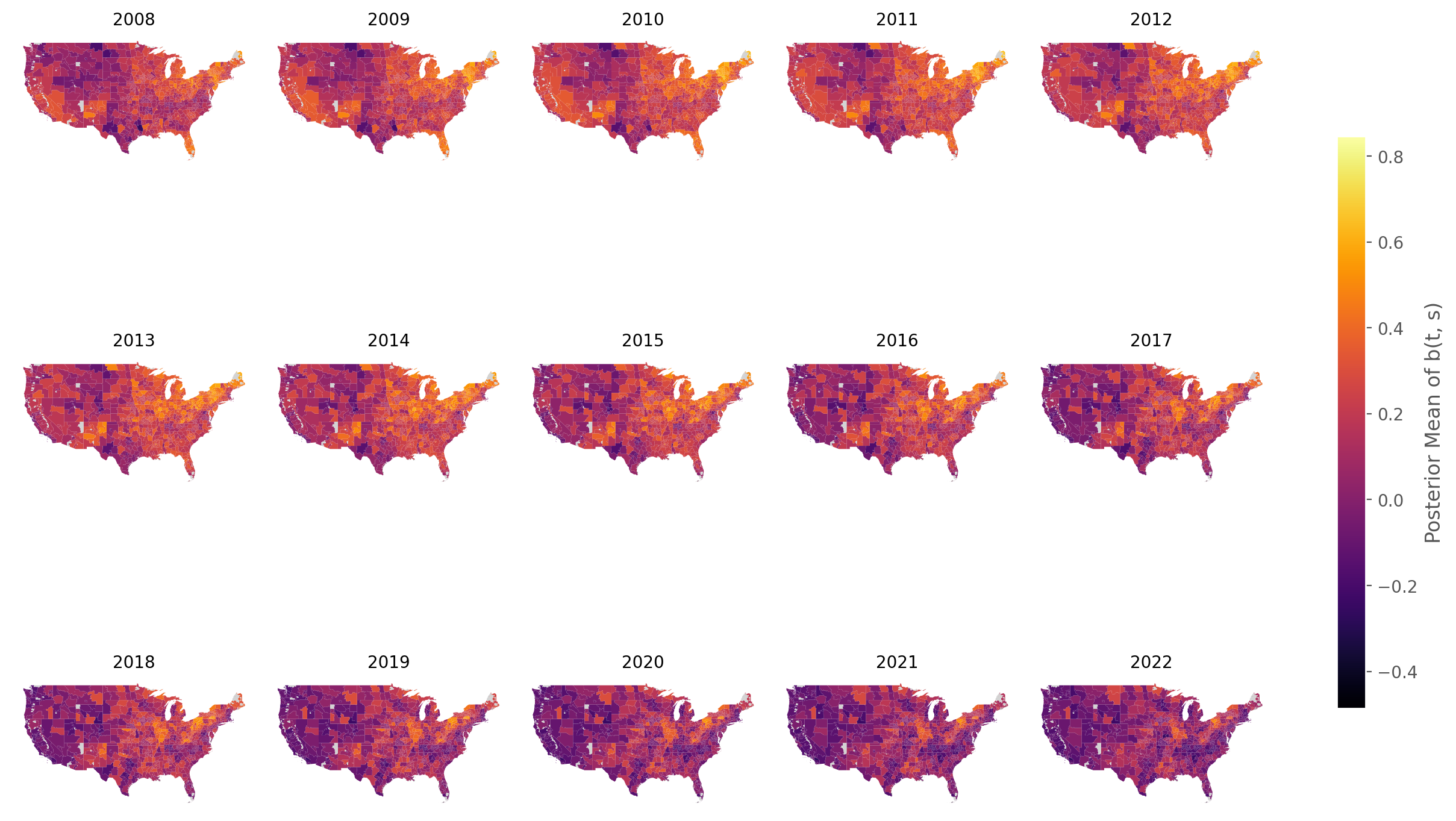}
  \caption{Posterior mean of latent Gaussian process in spatio-temporal tree-boosting model.}
  \label{fig:stgpb_heatmap}
\end{figure}

\begin{figure}[ht!]
  \centering
  \includegraphics[width=0.6\linewidth]{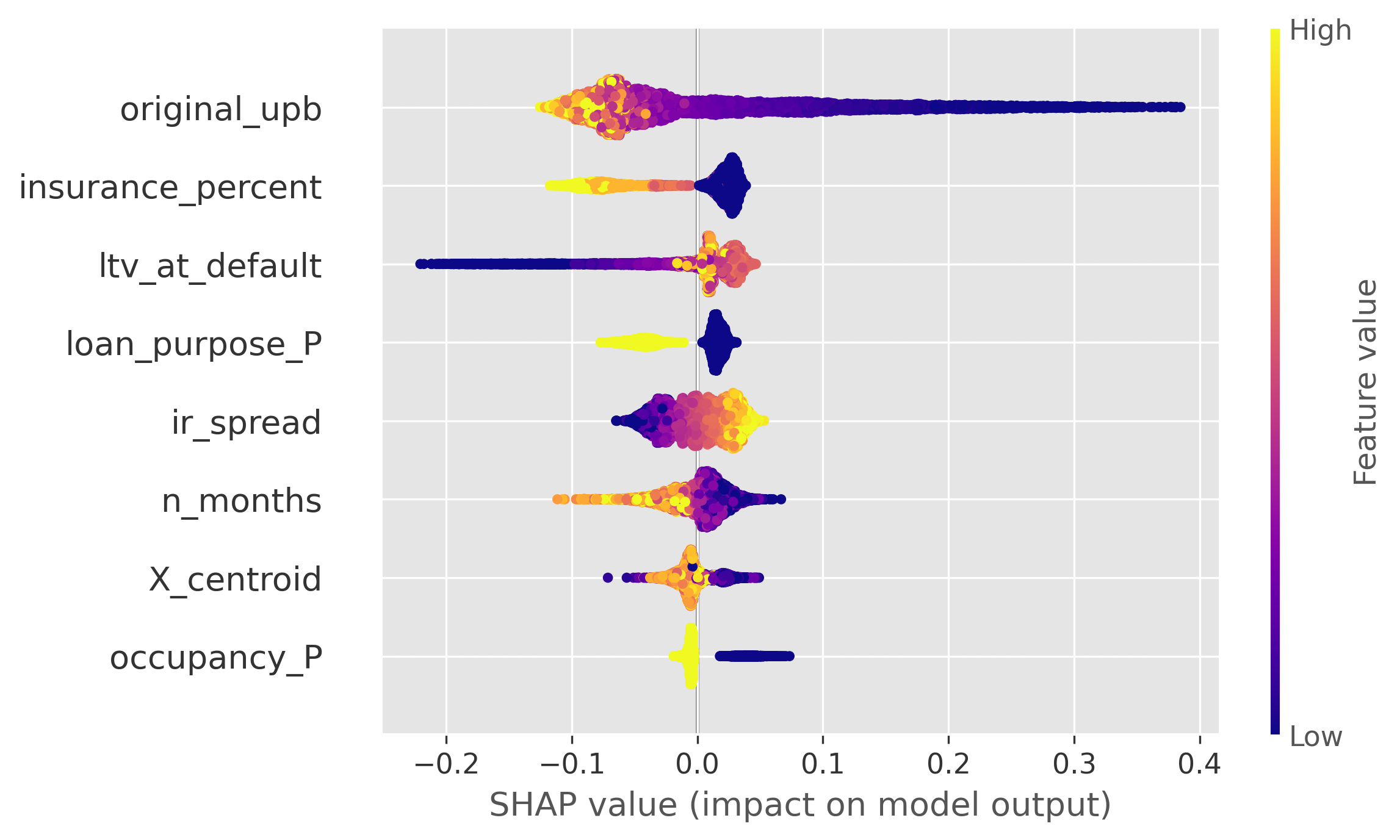}
  \caption{SHAP values for spatio-temporal GP tree-boosted model trained on data up to 2020. Predictor variables listed in descending order of average absolute SHAP value, top eight shown.}
  \label{fig:shap_summary}
\end{figure}

\begin{figure}[ht!]
  \centering
  \includegraphics[width=0.8\linewidth]{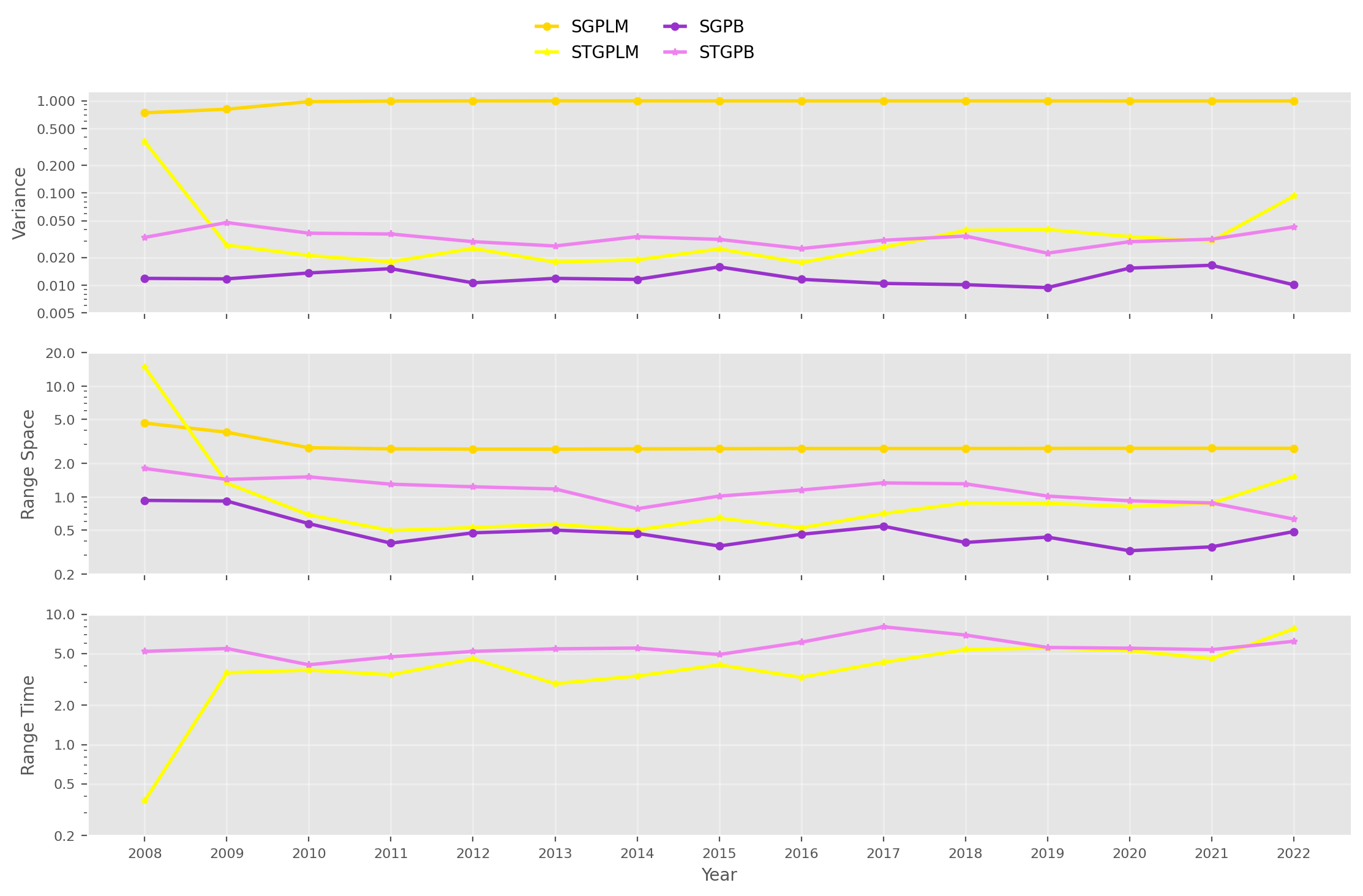}
  \caption{Estimated covariance parameters for GP models: spatial (S) / spatio-temporal (ST), linear (GPLM) / tree-boosting (GPB). Annual estimates correspond to distinct models trained on data up to the start of the year.}
  \label{fig:gpmodel_covpars}
\end{figure}

\end{document}